\documentclass[]{llncs}
\usepackage{pslatex}

\usepackage{paralist}
\usepackage{amssymb}
\usepackage{amsmath}
\usepackage{accents}
\usepackage{mathtools}
\usepackage{color}
\usepackage{stmaryrd}
\usepackage{leftidx}
\usepackage{hyperref}
\usepackage{float}
\usepackage{bbold}

\usepackage{tikz}
\usepackage{algorithm}
\usepackage{algorithmicx}
\usepackage{algcompatible}
\usepackage{bussproofs} 

\usepackage{lineno}

\newtheorem{fact}{Fact}

\newcommand{\ie}{\textit{i.e.,~}}
\newcommand{\eg}{\textit{e.g.,~}}
\newcommand{\resp}{\textit{resp.~}}

\newcommand*{\figref}[1]{\hyperref[#1]{\mbox{Figure}~\ref*{#1}}}
\newcommand*{\secref}[1]{\hyperref[#1]{\mbox{Section}~\ref*{#1}}}
\newcommand*{\lemref}[1]{\hyperref[#1]{\mbox{Lemma}~\ref*{#1}}}
\newcommand*{\thmref}[1]{\hyperref[#1]{\mbox{Theorem}~\ref*{#1}}}
\newcommand*{\propref}[1]{\hyperref[#1]{\mbox{Proposition}~\ref*{#1}}}
\newcommand*{\defref}[1]{\hyperref[#1]{\mbox{Definition}~\ref*{#1}}}
\newcommand*{\exref}[1]{\hyperref[#1]{\mbox{Example}~\ref*{#1}}}
\newcommand*{\appref}[1]{\hyperref[#1]{\mbox{Appendix}~\ref*{#1}}}
\newcommand*{\eqnref}[1]{\hyperref[#1]{\mbox{Equation}~\ref*{#1}}}

\newcommand{\mycomment}[2]{\textcolor{blue}{[#1]}{\textcolor{red}{#2}}}

\newcommand{\nat}{\mathbb{N}}
\newcommand{\ints}{\mathbb{Z}}

\newcommand{\interv}[2]{{[{#1},{#2}]}}
\newcommand{\isdef}{\stackrel{\scriptscriptstyle{\mathsf{def}}}{=}}
\newcommand{\iffdef}{\stackrel{\scriptscriptstyle{\mathsf{def}}}{\iff}}

\newcommand{\set}[1]{{\{ {#1} \}}}

\newcommand{\pow}[1]{{\mathrm{pow}({#1})}}

\newcommand{\dom}[1]{{\mathrm{dom}({#1})}}
\newcommand{\img}[1]{{\mathrm{img}({#1})}}
\newcommand{\cardof}[1]{|\!|{#1}|\!|}
\newcommand{\sizeof}[1]{|{#1}|}
\newcommand{\finsubseteq}{\subseteq_{\mathit{fin}}}
\newcommand{\arrow}[2]{\xrightarrow{{\scriptscriptstyle #1}}_{{\scriptscriptstyle #2}}}
\newcommand{\finmap}{{\rightharpoonup_{\mathit{fin}}}}
\renewcommand{\vec}[1]{\mathbf{#1}}
\newcommand{\proj}[2]{{#1}\!\downharpoonleft_{#2}}

\newcommand{\fnid}{\mathrm{id}}
\newcommand{\automata}{\mathcal{A}}

\newcommand{\sem}[1]{[\![{#1}]\!]}

\newcommand{\pspace}{\textsf{PSPACE}}
\newcommand{\twoexptime}{\textsf{2EXPTIME}}
\newcommand{\bigO}{\mathcal{O}}
\newcommand{\twonfa}{$2$\textsf{NFA}}

\newcommand{\anet}{\mathsf{N}}
\newcommand{\amarkednet}{\mathcal{N}}
\newcommand{\places}{\mathsf{Q}}
\newcommand{\placeof}[1]{\places_{#1}}
\newcommand{\trans}{\mathsf{T}}
\newcommand{\transof}[1]{\trans_{#1}}
\newcommand{\weight}{\mathsf{W}}
\newcommand{\weightof}[1]{\weight_{#1}}
\newcommand{\initmark}{\mathrm{init}}
\newcommand{\initmarkof}[1]{\initmark_{#1}}
\newcommand{\pre}[1]{\text{$\leftidx{^\bullet}{\text{${#1}$}}$}}
\newcommand{\post}[1]{\text{${#1}$}^\bullet}
\newcommand{\prepost}[1]{\leftidx{^\bullet}{\!\text{${#1}$}}^\bullet}
\newcommand{\amark}{\mathsf{m}}
\newcommand{\markset}{\mathsf{M}}
\newcommand{\reach}[1]{\mathrm{reach}({#1})}
\newcommand{\cover}[1]{\mathrm{cover}({#1})}
\newcommand{\fire}[1]{\stackrel{#1}{\leadsto}}

\newcommand{\ptype}{\mathsf{p}}
\newcommand{\ptypeof}[1]{{\mathrm{ptype}({#1})}}
\newcommand{\ptypes}{\mathcal{P}}
\newcommand{\obstransof}[1]{\trans^\mathit{obs}_{#1}}
\newcommand{\inttransof}[1]{\trans^\mathit{int}_{#1}}
\newcommand{\graph}{\mathsf{G}}
\newcommand{\verts}{\mathsf{V}}
\newcommand{\vertof}[1]{\verts_{#1}}
\newcommand{\edges}{\mathsf{E}}
\newcommand{\edgeof}[1]{\edges_{#1}}
\newcommand{\asys}{\mathsf{S}}
\newcommand{\parsys}{\mathbf{S}}
\newcommand{\valpha}{\Lambda}
\newcommand{\vlab}{\lambda}
\newcommand{\vlabof}[1]{\vlab_{#1}}
\newcommand{\ealpha}{\Delta}
\newcommand{\abeh}{\beta}
\newcommand{\behof}[1]{\abeh({#1})}
\newcommand{\pps}{\mathsf{pps}}
\newcommand{\ppstext}{\textsf{PPS}}

\newcommand{\hr}{\mathsf{HR}}

\newcommand{\hrtext}{\textsf{HR}}
\newcommand{\vrtext}{\textsf{VR}}
\newcommand{\sourcelabels}{\Sigma}
\newcommand{\open}[1]{\check{#1}}
\newcommand{\slabs}{\tau}
\newcommand{\typeof}[1]{\mathrm{type}({#1})}
\newcommand{\asrc}{\sigma}

\newcommand{\sources}{\xi}

\newcommand{\sgraph}[4]{{#1}^{#4}_{{#2},{#3}}}
\newcommand{\restrict}[2]{\mathsf{restrict}^{#2}_{{#1}}}
\newcommand{\rename}[2]{\mathsf{rename}^{#2}_{{#1}}}
\newcommand{\pop}[1]{\oplus^{#1}}

\newcommand{\grammar}{\Gamma}
\newcommand{\rules}{\Pi}
\newcommand{\nonterm}{\Xi}
\newcommand{\step}[1]{\Rightarrow_{#1}}
\newcommand{\algof}[1]{\mathcal{#1}}
\newcommand{\universeOf}[1]{\mathbb{#1}}
\newcommand{\aop}{\mathsf{op}}
\newcommand{\signature}{\mathsf{F}}
\newcommand{\kerof}[1]{\sim_{\scriptscriptstyle{\mathrm{ker}(#1)}}}
\newcommand{\srcrel}[2]{\kerof{\overline{#1}^{-1}\cup\overline{#2}^{-1}}}
\newcommand{\systems}{\universeOf{S}}
\newcommand{\behaviors}{\universeOf{B}}
\newcommand{\alangof}[3]{\mathcal{L}_{#1}^{\scriptscriptstyle{#2}}({#3})}

\newcommand{\paramreach}[3]{\mathsf{Reach}\ifthenelse{\equal{#1}{}\AND\equal{#2}{}\AND\equal{#3}{}}{}{({#1},{#2},{#3})}}
\newcommand{\paramcover}[3]{\mathsf{Cover}\ifthenelse{\equal{#1}{}\AND\equal{#2}{}\AND\equal{#3}{}}{}{({#1},{#2},{#3})}}

\newcommand{\mtarget}{\amark_\mathsf{tgt}}
\newcommand{\fpof}[1]{\mathsf{fp}_{#1}}

\newcommand{\afold}{\phi}
\newcommand{\foldof}[1]{\afold({#1})}
\newcommand{\foldrel}[1]{\operatorname{{\equiv}}_{[#1]}}
\newcommand{\absof}[1]{{#1}^\sharp}
\newcommand{\finabsof}[1]{{#1}^\flat}
\newcommand{\initof}[1]{\mathfrak{I}({#1})}
\newcommand{\annot}[2]{{#1}^{#2}}
\newcommand{\foldpn}[1]{\mathfrak{N}({#1})}
\newcommand{\zreach}[2]{\mathrm{reach}^0_{#1}({#2})}

\newcommand{\send}{\mathsf{send}}
\newcommand{\recv}{\mathsf{recv}}

\algnewcommand{\AND}{\textbf{and} }
\algnewcommand{\OR}{\textbf{or} }
\algnewcommand{\WHEN}{\textbf{when} }
\algnewcommand{\INCR}{\textbf{increasing}}
\algnewcommand{\DECR}{\textbf{decreasing}}
\algnewcommand{\ASSERT}{\STATE\textbf{assert} }
\algnewcommand{\LET}{\STATE\textbf{let} }
\algblockdefx[NAME]{CASE}{ENDCASE}%
  [1]{#1 $\Rightarrow$ \textbf{begin} }%
  {\textbf{end}}
\algblockdefx[NAME]{MATCH}{ENDMATCH}%
  [1]{\textbf{match} #1 \textbf{with} }%
  {\textbf{end match}}
\algnewcommand{\MARK}[1]{\hspace*{1em}{\color{blue}#1}}

\usetikzlibrary{arrows,shapes,decorations,automata,backgrounds,fit,matrix}
\usetikzlibrary{calc,decorations.pathreplacing}
\usetikzlibrary{petri}

\tikzstyle{petri-p}=[circle,thick,inner sep=0pt,minimum size=5mm]
\tikzstyle{petri-t}=[rectangle,thick,inner sep=0pt,minimum width=5mm,minimum height=1mm]
\tikzstyle{petri-t2}=[rectangle,thick,inner sep=0pt,minimum width=1mm,minimum height=5mm]
\tikzstyle{petri-tok}=[circle,inner sep=0pt,minimum size=4pt, color=black,fill=black]
\tikzstyle{petri-small-tok}=[circle,inner sep=0pt,minimum size=3pt, color=black,fill=black]        
\tikzstyle{gnode}=[circle,inner sep=0pt,minimum size=4pt, color=black,fill=black]

\usepackage{environ}
\NewEnviron{hide}{}
\newif\ifLongVersion\LongVersiontrue
\ifLongVersion
\usepackage[createShortEnv,conf={normal}]{proof-at-the-end}
\newenvironment{lemmaAtEnd}[1][]{\begin{lemmaE}}{\end{lemmaE}}
\newenvironment{proofSketch}{\hide}{\endhide}
\else
\usepackage[createShortEnv,conf={end,no link to proof,restate}]{proof-at-the-end}

\newenvironment{proofSketch}{\hide}{\endhide}
\fi

\begin{document}


\title{Counting Abstraction for the Verification of Structured Parameterized Networks}

 \author{Marius Bozga\inst{1} \and Radu Iosif\inst{1} \and Arnaud Sangnier\inst{2} \and Neven Villani\inst{1}}
 \institute{Univ. Grenoble Alpes, CNRS, Grenoble INP, VERIMAG, 38000, France \and
   DIBRIS, Univ. of Genova, Italy}

 \maketitle

\begin{abstract}
  We consider the verification of parameterized networks of replicated
  processes whose architecture is described by hyperedge-replacement
  graph grammars. Due to the undecidability of verification problems
  such as reachability or coverability of a given configuration, in
  which we count the number of replicas in each local state, we
  develop two orthogonal verification techniques. We present a
  counting abstraction able to produce, from a graph grammar
  describing a parameterized system, a finite set of Petri nets that
  over-approximate the behaviors of the original system. The counting
  abstraction is implemented in a prototype tool, evalutated on a
  non-trivial set of test cases. Moreover, we identify a decidable
  fragment, for which the coverability problem is in \twoexptime{}
  and \pspace-hard.
\end{abstract}

\section{Introduction}

The architecture is an important design aspect related to the
functionality of a computer network. For instance, the code of a
consensus protocol changes depending on whether it is used in a ring
or a clique-shaped network. The architecture is also a factor on which
the traffic balance and overall efficiency of communication
depend. Being able to formally model network architectures is a key
enabler for the use of verification algorithms that prove absence of
error scenarios in a distributed environment (\eg deadlocks, races
or mutual exclusion violations), or convergence towards a desired goal
(\eg building a spanning-tree or electing a leader).

The impressive size of present day networks requires
\emph{parameterized models}, that describe infinite families of
networks having an unbounded number of nodes. The problem of
\emph{parameterized verification} (\ie proving correctness for any
number of processes) often amounts to model-checking a small cut-off
of the system (see~\cite{BloemJacobsKhalimovKonnovRubinVeithWidder15}
for a survey). In cases where a cut-off does not exist or is too
large, symbolic representations of invariants (\ie sets of
configurations closed under local and communication actions) using \eg
boolean constraints~\cite{DBLP:conf/concur/EsparzaRW22},
well-structured transitions
systems~\cite{DBLP:conf/lics/AbdullaCJT96}, monadic second-order
logic~\cite{DBLP:journals/jlap/BozgaIS21} or finite-state
automata~\cite{DBLP:conf/birthday/LinR21} can be used to decide a
parameterized safety problem in a matter of seconds. This is because,
in particular, parameterized verification methods do not suffer from
the state explosion problem of classical model-checking techniques,
that scale poorly in the number of concurrent processes.

However, a current limitation is that most techniques rely on
hard-coded network topologies, typically
cliques~\cite{GermanSistla92}, rings~\cite{ClarkeGrumbergBrowne86} or
combinations thereof~\cite{DBLP:journals/dc/AminofKRSV18}. Many
architecture description languages have been developped by the
software engineering community (see, \eg \cite{bradbury2004survey} and
\cite{rumpe2017classification} for surveys) to support network design
but, in general, these languages lack support for verification. Only
few recent parameterized verification techniques take architectures as
input of the problem, described using, \eg
first-order~\cite{DBLP:conf/tacas/BozgaEISW20} or separation
logic~\cite{DBLP:journals/tcs/BozgaIS23}. Such descriptive
(logic-based) languages are typically hard to use, because of the
generality of their semantics, that requires complex frame conditions
to specify what is actually \emph{not} part of the architecture.

\vspace*{-.5\baselineskip}
\paragraph{Contributions} We consider the parametric verification problem
for process networks specified by \emph{graph grammars} that use
the operations of the standard \emph{hyperedge-replacement} (\hrtext)
algebra of graphs~\cite{courcelle_engelfriet_2012} to describe how
graphs are inductively build from smaller subgraphs. This constructive
aspect of graph grammars makes them appealing for network design,
because recursive specification of types and datastructures are
widespread among programmers. Graph grammars are, moreover, at the
core of a solid theory (see~\cite{courcelle_engelfriet_2012} for a
comprehensive survey). In principle, \hrtext\ graph grammars can
specify families of graphs having bounded tree-width, such as chains,
rings, stars, trees (of unbounded rank) and beyond, \eg overlaid
structures such as trees or stars with certain nodes linked in a
list. Since cliques and grids are families of unbounded tree-width,
neither can be specified using \hrtext\ graph grammars.
\ifLongVersion However, the relation between the expressiveness of
\emph{vertex-replacement} (\vrtext) and that of
\hrtext\ grammars\footnote{Each \vrtext\ grammar can be transformed
  into an \hrtext\ grammar modulo an elimination of epsilon edges,
  similar to the epsilon edge elimination for finite automata.}
~\cite{COURCELLE1995275} could provide a generalization of our
techniques to cliques, bipartite graphs and, in general, architectures
specified using \vrtext\ grammars. We consider this for future
work. \fi

Because the parameterized verification problem is undecidable, even
for chain-like networks, we consider two orthogonal lines of
work. First, following the seminal work of German and Sistla
\cite{GermanSistla92}, where identities of processes communicating
through rendez-vous in a clique is ignored to keep only the number of
processes in each state, we define a \emph{counting abstraction} that
folds the infinite set of networks specified using a \hrtext{} grammar
into a finite set of Petri nets, which subsumes the behaviors of the
original set of networks. As a consequence, if a set of places is not
covered by an execution of some of the resulting Petri nets, then it
is not covered by the original set of behaviors. These coverability
problems can be used to express mutual exclusion, and can encode other
properties (\eg finite-valued consensus). Even though, computing the
counting abstraction for clique networks is fairly
simple~\cite{GermanSistla92}, it is less trivial for families of
networks specified by a grammar. We circumvent these technical
challenges by defining appropriate \hrtext{} algebras, in which the
abstraction can be computed by a finite Kleene iteration of the
grammar. This line of work is motivated by several recent advances on
the theory~\cite{DBLP:journals/sosym/FinkelL15} and tool
support~\cite{DBLP:conf/apn/Wolf18a} for the reachability and
coverability problem for Petri nets. The abstraction has been
implemented in a prototype tool and a number of experiments showing
the effectiveness of the method have been carried out.

Second, we define a decidable fragment of the original problem, by
restricting the local behavior of the nodes to \emph{pebble-passing
  systems}, where a finite (but unbounded) number of identical pebbles
can be moved from one node to another. We inspire ourselves from
token-passing
systems\cite{aminof-paramtoken-vmcai14,aminof-model-ijcar16} for which
a restriction on the behavior of each proccess allows to get
decidability results of some verification problems. Note that in our
case, the processes definition are simple, but we allow an unbounded
number of tokens/pebbles. Interestingly, our decidable restriction
applies only to the local behavior and does not restrict the family of
networks considered, other than that they must be the language of a
given \hrtext{} grammar. Examples of problems from this decidable
fragment include token-rings, tree-traversal used, in general, for
notification and binary consensus among the participants of general
(\hrtext-specified) architectures. We have studied the complexity of
the decidable fragment and found it to be doubly-exponential in the
unary size of the coverability property and in the maximal tree-width
the set of networks generated by the grammar. At the same time, we
found the problem to be \pspace-hard.

\vspace*{-.5\baselineskip}
\paragraph{Related Work}
Traditionally, verification of unbounded networks of parallel
processes considers known architectural patterns, typically cliques or
rings \cite{GermanSistla92,ClarkeGrumbergBrowne86}. Because the price
for decidability is drastic restriction on architecture styles
\cite{BloemJacobsKhalimovKonnovRubinVeithWidder15}, more recent works
propose practical semi-algorithms, \eg \emph{regular model checking}
\cite{KestenMalerMarcusPnueliShahar01} or \emph{automata learning}
\cite{ChenHongLinRummer17}. Here the architecture is implicitly
determined by the class of language recognizers: word automata encode
pipelines or rings, whereas tree automata describe trees.

Specifying parameterized concurrent systems by inductive definitions
is reminiscent of \emph{network grammars}
\cite{ShtadlerGrumberg89,LeMetayer,Hirsch}, that use inductive rules
to describe systems with linear (pipeline, token-ring) architectures
obtained by composition of an unbounded number of processes. In
contrast, our language is based on the \hrtext\ graph
algebra. Verification of network grammars against safety properties
(reachability of error configurations) requires the synthesis of
\emph{network invariants} \cite{WolperLovinfosse89}, computed by
rather costly fixpoint iterations \cite{LesensHalbwachsRaymond97} or
by abstracting (forgetting the particular values of indices in) the
composition of a small bounded number of instances
\cite{KestenPnueliShaharZuck02}. A more recent line of work considers
a lightweight invariant synthesis method based on the inference of
\emph{structural invariants}\footnote{Invariants that depend on the
  structure of a Petri net, which hold for any of its executions.} of
an infinite family of Petri nets, for parameterized systems whose
network architectures are specified using
logic~\cite{DBLP:conf/tacas/BozgaEISW20,DBLP:journals/tcs/BozgaIS23}. This
method is geared towards deadlock-freedom and mutual exclusion,
whereas our counting abstraction method works for coverability
properties, in general.

Other efficient methods to verify safety properties of parameterized
networks consist in changing the semantics of behaviors to obtain an
over-approximation having a decidable safety verification problem.  In
\cite{abdulla-monotonic-foundcs09}, the authors have implemented such
a scheme using a \emph{monotonic abstraction}, where the resulting
abstraction is a \emph{well-structured transition
  systems}~\cite{DBLP:conf/lics/AbdullaCJT96}.  They first consider
pipeline architectures, where communication is done by checking
existentially or universally the state of the other proceess.  In
\cite{abdulla-approximated.fmsd09}, a similar technique is applied to
networks with clique architectures, where processes can manipulate
shared boolean and natural variables. In contrast, both our methods
target network architectures defined by unrestricted \hrtext{} graph
grammars.

\section{Preliminaries}

We denote by $\nat$ the set of positive integers, including zero. For
$i, j \in \nat$, we denote by $\interv{i}{j}$ the set
$\set{i,\ldots,j}$, considered empty if $i > j$. The cardinality of a
finite set $A$ is denoted $\cardof{A}$. A singleton $\set{a}$ will be
denoted $a$. By $A \finsubseteq B$, we mean that $A$ is a finite
subset of $B$. The union of two disjoint sets $A$ and $B$ is denoted
as $A \uplus B$. The Cartesian product of two sets $A$ and $B$ is
denoted $A \times B$. As usual, we denote by $A^*$ the set of
(possibly empty) sequences of elements from $A$.

For a function $f : A \rightarrow B$ and $C \subseteq A$, we write
$f(C)\isdef \set{f(c) \mid c \in C}$ and $\proj{f}{C} \isdef
\set{(c,f(c)) \mid c \in C}$. The inverse of a function $f : A
\rightarrow B$ is the relation $f^{-1}\isdef\{(f(a),a) \mid a \in
A\}$. To alleviate notation, we write $f(x,y)$ instead of $f((x,y))$,
when no confusion arises. For two functions $f : A \rightarrow B$ and
$g : C \rightarrow D$, the function $f \times g : A \times B
\rightarrow C \times D$ maps each pair $(a,c) \in A \times C$ into the
pair $(f(a),g(c))\in B \times D$. A bijective function $f : A
\rightarrow A$ is a \emph{finite permutation} if the set $\set{a \in A
  \mid f(a)\neq a}$ is finite.  In particular, $a \leftrightarrow b$
denotes the finite permutation that switches $a$ with $b$ and leaves
the other elements of the domain unchanged. A finite partial function
is denoted as $f : A \finmap B$ and $\dom{f}$, $\img{f}$ denote its
domain and range, respectively.  When $f : A \to C$ and $g : B \to C$
coincide on their shared domain $A \cap B$, we write $f \cup g : A
\cup B \to C$ the function such that $\proj{(f \cup g)}{A} = f$ and
$\proj{(f \cup g)}{B} = g$.
\ifLongVersion\else
The proofs of the technical results from the paper are given in
\appref{app:proofs}.
\fi

\subsection{Petri Nets}

A \emph{net} is a tuple $\anet = (\places,\trans,\weight)$, where
$\places$ is a finite set of \emph{places}, $\trans$ is a finite set
of \emph{transitions} such that $\places\cap\trans=\emptyset$ and
$\weight : (\places \times \trans) \cup (\trans \times \places)
\rightarrow \nat$ is a \emph{weighted incidence relation} between
places and transitions. We denote by $\placeof{\anet}$,
$\transof{\anet}$ and $\weightof{\anet}$ the places, transitions and
incidence relation of $\anet$, respectively. For all
$x,y\in\places\cup\trans$ such that $\weight(x,y) > 0$, we say that
there is an \emph{edge of weight} $\weight(x,y)$ between $x$ and
$y$. For an element $x\in\places\cup\trans$, we define the set of
\emph{predecessors} $\pre{x} \isdef \set{y \in \places \cup \trans
  \mid \weight(y,x)>0}$, \emph{successors} $\post{x} \isdef \set{y \in
  \places \cup \trans \mid \weight(x,y)>0}$ and predecessor-successor
pair $\prepost{x} \isdef (\pre{x},\post{x})$.

A \emph{marking} of $\anet = (\places,\trans,\weight)$ is a function
$\amark : \places \rightarrow \nat$. A transition $t$ is
\emph{enabled} in the marking $\amark$ if $\amark(q) \geq
\weight(q,t)$, for each place $q \in \places$. For all markings
$\amark$, $\amark'$ and transitions $t \in \trans$, we write $\amark
\fire{t} \amark'$ whenever $t$ is enabled in $\amark$ and $\amark'(q)
= \amark(q) - \weight(q,t) + \weight(t,q)$, for all $q \in
\places$. Given two markings $\amark$ and $\amark'$, a finite sequence
of transitions $\vec{t} = (t_1, \ldots,t_n)$ is a \emph{firing
  sequence}, written $\amark \fire{\vec{t}} \amark'$, if and only if
either \begin{inparaenum}[(i)]
\item $n=0$ and $\amark=\amark'$, or
\item $n\geq1$ and there exist markings $\amark_1, \ldots,
  \amark_{n-1}$ such that $\amark \fire{t_1} \amark_1 \fire{t_2} \ldots
  \fire{t_{n-1}} \amark_{n-1} \fire{t_n} \amark'$.
\end{inparaenum}
A sequence $\vec{t}$ is \emph{fireable} from $\amark$ whenever there
exists a marking $\amark'$ such that $\amark \fire{\vec{t}} \amark'$.

A \emph{Petri net} (PN) is a pair $\amarkednet=(\anet,\amark_0)$,
where $\anet$ is a net and $\amark_0$ is the \emph{initial marking} of
$\anet$. For simplicity, we write $\placeof{\amarkednet} \isdef
\placeof{\anet}$, $\transof{\amarkednet} \isdef \transof{\anet}$,
$\weightof{\amarkednet} \isdef \weightof{\anet}$ and
$\initmarkof{\amarkednet}\isdef\amark_0$ for the elements of
$\amarkednet$. A marking $\amark$ is \emph{reachable} in $\amarkednet$
iff there exists a firing sequence $\vec{t}$ such that $\amark_0
\fire{\vec{t}} \amark$. We denote by $\reach{\amarkednet}$ the set of
reachable markings of $\amarkednet$. The \emph{reachability problem}
asks, given a PN $\amarkednet$ and a marking $\amark$, does
$\amark\in\reach{\amarkednet}$ ? The \emph{coverability problem} asks,
for a given PN $\amarkednet$ and marking $\amark$, does there exists a
marking $\amark' \in \reach{\amarkednet}$ such that $\amark \leq
\amark'$? Here the order of markings is the pointwise order on $\nat$,
\ie $\amark \leq \amark'$ iff $\amark(q) \leq \amark'(q)$ for all $q
\in \placeof{\amarkednet}$. The coverability problem is more concisely
stated using the set of covered markings $\cover{\amarkednet} \isdef
\{\amark \mid \exists \amark'\in\reach{\amarkednet} ~.~ \amark \leq
\amark'\}$, \ie given $\amarkednet$ and $\amark$, does
$\amark\in\cover{\amarkednet}$?



\subsection{Parameterized Systems}

We begin by defining parameterized communicating systems, \ie graphs
whose vertices model network nodes that run identical copies of one or
more process types.  Neighbouring processes synchronize their
transitions according to the observable edge labels of the network
graph. In most of the literature (see, \eg \cite{BloemBook} for a
survey) process types are represented by finite labeled transition
systems (LTS), \eg with disjoint observable and internal alphabets of
transition labels. For simplicity, here we use PNs whose transitions
mimick closely the transitions of a LTS, thus avoiding the formal
definition of the latter.

Let $\valpha$ and $\ealpha$ be finite disjoint alphabets of vertex and
edge labels, respectively. A (binary labeled) \emph{graph} is a tuple
$\graph=(\verts,\edges,\vlab)$, where $\verts$ is a finite set of
vertices, $\edges \subseteq \verts\times\ealpha\times\verts$ is a set
of labeled binary edges and $\vlab : \verts \rightarrow \valpha$ maps
each vertex to a vertex label. Edges $(v_1,a,v_2)$ are written $v_1
\arrow{a}{} v_2$. We denote by $\vertof{\graph}$, $\edgeof{\graph}$
and $\vlabof{\graph}$ the vertices, edges and vertex labeling of
$\graph$, respectively. We do not distinguish isomorphic graphs,
\ie graphs that differ only in the identities of their vertices.

\begin{definition}\label{def:process-type}
  A \emph{process type} $\ptype$ is a PN having all weights at most
  $1$ and exactly one marked place initially, whose transitions are
  partitioned into \emph{observable} $\obstransof{\ptype}$ and
  \emph{internal} $\inttransof{\ptype}$, \ie
  $\transof{\ptype}=\obstransof{\ptype} \uplus \inttransof{\ptype}$,
  and each transition has exactly one predecessor and one
  successor. Let $\ptypes = \set{\ptype_1, \ldots, \ptype_k}$ be a
  finite fixed set of process types such that $\placeof{\ptype_i} \neq
  \emptyset$, for all $i \in \interv{1}{k}$ and $\placeof{\ptype_i}
  \cap \placeof{\ptype_j} = \emptyset$, for all $1 \leq i < j \leq k$.
\end{definition}
Because a process type has exactly one initial token and all
transitions have one predecessor and one successor, every reachable
marking of a process type has exactly one initial token. A PN having
this property is said to be \emph{automata-like}. We denote by
$\placeof{\ptypes}$ and $\obstransof{\ptypes}$ the sets of places and
observable transitions from some $\ptype\in\ptypes$, respectively.

\input{figure-proctype}

\begin{example}
  \figref{fig:proc-typ1} (a) shows two process types $\mathit{Cont}$
  and $\mathit{Proc}$. They both represent entities that can hold a
  token and they can either grab a token, if they do not have it, or
  release the token, otherwise. The first one, which we identify as a
  controller, has only observable transitions, whereas the second one,
  which represents a worker process, has two internal trasitions
  $\mathit{start}$ and $\mathit{stop}$ depicted in yellow. These
  transitions are used to simulate the fact that when the worker has
  the token, it can move to a working state, from which he cannot
  release the token and when it stops working, it can move back to a
  state from which the token can be released.
\end{example}

\begin{definition}\label{def:system}
  A \emph{system} $\asys=(\verts,\edges,\vlab)$ is a graph whose
  vertices are labeled with process types from $\ptypes$
  ($\valpha=\ptypes$) and edges with pairs of observable transitions
  from $\obstransof{\ptypes}$ (\ie
  $\ealpha=\obstransof{\ptypes}\times\obstransof{\ptypes}$), such that
  $t_i \in \obstransof{\vlab(v_i)}$, for both $i=1,2$, for each edge
  $v_1\arrow{(t_1,t_2)}{} v_2 \in \edgeof{\asys}$.
\end{definition}

\begin{example}
  \figref{fig:proc-typ1} (b) shows two systems labeled with the process
  types given by \figref{fig:proc-typ1} (a). They both represent a
  network with four entities, one controller of type $\mathit{Cont}$
  and three working processes of type $\mathit{Proc}$. In $\asys_1$,
  the controller can pass a token to the first working process, which
  can pass the token to the second one which can pass the token to the
  third one. In $\asys_2$, the controller is at the center and it
  communicate with all the working processes that get and release the
  token.
\end{example}

The communication (\ie synchronization between processes) in a
system is formally captured by the following notion of behavior:

\begin{definition}\label{def:behavior}
  A \emph{behavior} is a PN $\amarkednet$ such that $1 \leq
  \cardof{\pre{t}}=\cardof{\post{t}} \leq 2$, for each
  $t\in\transof{\amarkednet}$. The \emph{behavior of a system}
   $\asys=(\verts,\edges,\vlab)$ is $\behof{\asys} \isdef
   (\anet,\amark_0)$, where: \begin{itemize}
   \item $\placeof{\anet} \isdef \set{(q,v) \mid q \in
     \placeof{\vlab(v)},~ v \in \verts}$, a place $(q,v)$ corresponds
     to the place $q$ of the process type $\vlab(v)$ that labels the
     vertex $v$;
   \item $\transof{\anet} \isdef \edges \cup \set{(t,v) \mid t \in
     \inttransof{\vlab(v)},~ v \in \verts}$, the transitions are
     either edges of the system (\ie modeling the synchronizations
     of two processes) or pairs $(t,v)$ corresponding to an internal
     transition $t$ of the process type $\vlab(v)$ that labels the
     vertex $v$;
   \item the weight function $\weightof{\anet}$ is defined below:
     \begin{align*}
       \weightof{\anet}((q,v),v_1 \arrow{(t_1,t_2)}{} v_2) \isdef & \left\{\begin{array}{ll}
       \weightof{\vlab(v)}(q,t_i) \text{, if } v=v_i, \text{, for } i=1,2 \\
       0 \text{, otherwise}
       \end{array}\right.  \\
       \weightof{\anet}(v_1 \arrow{(t_1,t_2)}{} v_2,(q,v)) \isdef & \left\{\begin{array}{ll}
       \weightof{\vlab(v)}(t_i,q) \text{, if } v=v_i \text{, for } i=1,2  \\
       0 \text{, otherwise}
       \end{array}\right.
     \end{align*}
     
     \vspace*{-.5\baselineskip}
     \begin{align*}
       \hspace*{-8mm}
       \weightof{\anet}((q,v),(t,v')) \isdef \left\{\begin{array}{ll}
       \weightof{\vlab(v)}(q,t) \text{, if } v=v' \\
       0 \text{, otherwise}
       \end{array}\right.
       \hspace*{4mm}
       \weightof{\anet}((t,v'),(q,v)) \isdef \left\{\begin{array}{ll}
       \weightof{\vlab(v)}(t,q) \text{, if } v=v' \\
       0 \text{, otherwise}
       \end{array}\right. 
     \end{align*}
   \item $\amark_0(q,v)\isdef\initmarkof{\vlab(v)}(q)$, for all $v \in \verts$ and $q \in \placeof{\vlab(v)}$. 
   \end{itemize}
\end{definition}
It is easy to check that $\behof{\asys}$ is a behavior, for any system
$\asys$. For example, \figref{fig:proc-typ1} (c) and (d) show the
behaviors of the two systems from \figref{fig:proc-typ1} (b).  By
construction, among places $\set{(q, v) \mid q\in\placeof{\vlab(v)}}$
for any $v$, there is exactly one token in all reachable markings
because $\vlab(v)$ is automata-like. By extension we say that such
behaviors are automata-like.

A \emph{parameterized system} $\parsys = \set{\asys_1,\asys_2,
  \ldots}$ is a possibly infinite set of systems, called
\emph{instances}. A parameterized system has an infinite set of
behaviors, denoted as $\behof{\parsys}$, \ie one for each
instance. \ifLongVersion
In the rest of this paper we are concerned with the following
parameterized verification problems:
\begin{definition}\label{def:parameterized-verif} 
  The \emph{parameterized reachability (\resp coverability) problem}
  takes in input a parameterized system $\parsys$, a set of places $\mathcal{Q}
  \subseteq \placeof{\ptypes}$ and a mapping 
  $\amark : \mathcal{Q} \rightarrow \nat$ and asks for the
  existence of an instance $\asys\in\parsys$ and of a marking
  $\overline{\amark} \in \reach{\behof{\asys}}$ (\resp
  $\overline{\amark} \in \cover{\behof{\asys}}$) where
  \(\sum_{\set{v \in \vertof{\asys} \mid q \in \placeof{\vlabof{\asys}(v)}}}
  \overline{\amark}(q,v)=\amark(q) \text{, for all } q \in
  \mathcal{Q}.\)
\end{definition}
For Petri nets, the reachability problem specified by the number of
token within a set of places is known as the \emph{submarking
reachability} \cite{DecidabilityPetriNets}. For example, a typical
correctness property that can be stated as a parameterized
coverability problem is mutual exclusion: a parameterized system
violates mutual exclusion if there is an instance that reaches a
marking in which two or more processes of the same type are in the
same given local place.
\begin{example}
  For the two systems depicted on \figref{fig:proc-typ1}, no two
  processes (\ie running on different nodes) work at the same time
  if and only if the parameterized coverability problem for
  $\places=\set{\mathit{work}}$ and $\amark = \set{(\mathit{work},2)}$
  has a negative answer, \ie some marking that agrees with $\amark$
  over $\mathit{work}$ cannot be covered by any marking reachable in
  the behavior of some instance.
\end{example}
\else The verification problems considered in this paper are, given a
parameterized system $\parsys$ and a marking $\amark$ for a subset
$\mathcal{Q}$ of the places in $\ptypes$, does there exist an instance
of $\parsys$ whose behavior reaches (covers) a marking that agrees
with $\amark$ over $\mathcal{Q}$? We shall define these problems
formally, once we have introduced the language for the specification
of parameterized systems. \fi

\subsection{Algebras}
\label{subsec:algebras}

We recall a few notions on algebras needed in the following. A
\emph{signature} is a set of function symbols $\signature=\set{\aop_1,
  \aop_2, \ldots}$.  An \emph{$\signature$-term} is a term built with
function symbols from $\signature$ and variables of arity zero.  An
$\signature$-term is \emph{ground} if it has no variables.  An
\emph{$\signature$-algebra} $\algof{A} =
(\universeOf{A},\aop^\algof{A}_1,\aop^\algof{A}_2,\ldots)$ interprets
the function symbols from $\signature$ as functions over the
\emph{domain} $\universeOf{A}$.  Given $\signature$-algebras
$\algof{A}$ and $\algof{B}$ having domains $\universeOf{A}$ and
$\universeOf{B}$, respectively, a \emph{homomorphism} is a function $h
: \universeOf{A} \rightarrow \universeOf{B}$ such that
$h(\aop^\algof{A}(a_1, \ldots, a_n)) = \aop^{\algof{B}}(h(a_1),
\ldots, h(a_n))$, for each function symbol $\aop\in\signature$ of
arity $n$ and $a_1, \ldots, a_n \in \universeOf{A}$. The \emph{kernel}
of a function $f : \universeOf{A} \rightarrow \universeOf{B}$ is the
equivalence relation $\kerof{f} \subseteq \universeOf{A} \times
\universeOf{A}$ defined as $a_1 \kerof{f} a_2 \iff f(a_1)=f(a_2)$. An
equivalence relation $\sim \subseteq \universeOf{A} \times
\universeOf{A}$ is an $\signature$-\emph{congruence} if and only if,
for each function symbol $\aop\in\signature$ of arity $n$ and $a_1,
a'_1, \ldots, a_n, a'_n \in \universeOf{A}$ such that $a_i\sim a'_i$,
for all $i \in \interv{1}{n}$, we have $\aop^\algof{A}(a_1,\ldots,a_n)
\sim \aop^\algof{A}(a'_1,\ldots,a'_n)$.

\begin{propositionE}[][category=proofs]\label{prop:cong-homo}
  Let $\algof{A}$ be an $\signature$-algebra having domain
  $\universeOf{A}$, and $f : \universeOf{A} \rightarrow
  \universeOf{B}$ be a function such that $\kerof{f}$ is an
  $\signature$-congruence. Then $f$ is a homomorphism between
  $\algof{A}$ and $\algof{B} \isdef
  (f(\universeOf{A}),\set{\aop^\algof{B}}_{\aop\in\signature})$, where
  $\aop^\algof{B}(b_1,\ldots,b_n)\isdef f(\aop^\algof{A}(f^{-1}(b_1),
  \ldots, f^{-1}(b_n)))$\footnote{The right-hand side of this
    definition is a singleton that we identify with its element.}, for
  each function symbol $\aop\in\signature$ of arity $n$ and all
  $b_1,\ldots,b_n\in f(\universeOf{A})$. Consequently,
  $f(\theta^\algof{A})=\theta^\algof{B}$, for each ground
  $\signature$-term $\theta$.
\end{propositionE}
\begin{proofE}
  First, we prove that $\aop^\algof{B}(b_1,\ldots,b_n)$ is
  well-defined and evaluates to a singleton, for each $\aop \in
  \signature$ and $b_1,\ldots,b_n \in f(\universeOf{A})$. Let $a_i \in
  f^{-1}(b_i)$ be elements, for all $i\in\interv{1}{n}$. The element
  $a_i$ exists, by the choice of $b_i \in f(\universeOf{A})$, for all
  $i \in \interv{1}{n}$. Then $f^{-1}(b_i)$ is the
  $\kerof{f}$-equivalence class of $a_i$, for all
  $i\in\interv{1}{n}$. Because $\kerof{f}$ is an
  $\signature$-congruence, $\aop^\algof{A}(f^{-1}(b_1), \ldots,
  f^{-1}(b_n))$ is the $\kerof{f}$-equivalence class of
  $\aop^\algof{A}(a_1,\ldots,a_n)$, hence
  $f(\aop^\algof{A}(f^{-1}(b_1), \ldots,
  f^{-1}(b_n)))=\set{f(\aop^\algof{A}(a_1), \ldots,
    \aop^\algof{A}(a_n))}$ is a singleton. Second, to prove that $f$
  is an homomorphism between $\algof{A}$ and $\algof{B}$, we compute,
  for each $\aop\in\signature$ and all
  $a_1,\ldots,a_n\in\universeOf{A}$:
  \begin{align*}
    \aop^\algof{B}(f(a_1), \ldots, f(a_n)) = & f(\aop^\algof{A}(f^{-1}(f(a_1)), \ldots, f^{-1}(f(a_n))))
    \text{, by the definition of } \aop^\algof{B} \\
    = & f(\aop^\algof{A}(a_1, \ldots, a_n)) \text{, because } \kerof{f} \text{ is an $\signature$-congruence}
  \end{align*}
  Finally, $f(\theta^\algof{A})=\theta^\algof{B}$ holds for each
  ground $\signature$-term $\theta$, because $f$ is a homomorphism
  between $\algof{A}$ and $\algof{B}$. \qed
\end{proofE}
\begin{proofSketch}
  The fact that $f(\aop^\algof{A}(f^{-1}(a_1),\ldots,f^{-1}(a_n)))$ evaluates to a singleton
  comes from the fact that $\kerof{f}$ is an $\signature$-congruence quite directly.
  In $f(\theta^\algof{A})$, use the now established definition of $\aop^\algof{B}$
  to push $f$ to the leaves.
  \qed
\end{proofSketch}

We introduce a standard signature of operations on graphs and define
two algebras, of open systems and behaviors. Let $\sourcelabels$ be a
countably infinite set of \emph{source labels}, fixed in the rest of
the paper. With no loss of generality, we assume that $\sourcelabels$
is partitioned into disjoint sets indexed by the process types
$\ptypes = \set{\ptype_1, \ldots, \ptype_k}$,
\ie $\sourcelabels = \sourcelabels_{\ptype_1} \uplus \ldots \uplus \sourcelabels_{\ptype_k}$.
A source label $\asrc\in\sourcelabels$
uniquely identifies a process type $\ptypeof{\asrc}\in\ptypes$ such
that $\asrc \in \sourcelabels_{\ptypeof{\asrc}}$.
A function $\alpha : \sourcelabels \rightarrow \sourcelabels$ is
\emph{$\ptypes$-preserving} if
$\ptypeof{\alpha(\asrc)}=\ptypeof{\asrc}$,
for all $\asrc\in\sourcelabels$.

The signature of the \emph{hyperedge-replacement} graph algebra
(\hrtext)~\cite{courcelle_engelfriet_2012} consists of the constants
$\sgraph{a}{\asrc_1}{\asrc_2}{}$, for all edge labels $a\in\ealpha$
and source labels $\asrc_1, \asrc_2 \in \sourcelabels$, the unary
symbols $\restrict{\slabs}{}$, for all $\slabs \finsubseteq
\sourcelabels$, $\rename{\alpha}{}$, for all $\ptypes$-preserving
finite permutations $\alpha: \sourcelabels \rightarrow \sourcelabels$
and the binary symbol $\pop{}$. By \hrtext{} congruence we mean an
equivalence relation that is a congruence for the \hrtext{} signature.

An \emph{open system} is a pair $\open{\asys} \isdef
(\asys,\sources)$, where $\asys=(\verts,\edges,\vlab)$ is a system and
$\sources : \sourcelabels \finmap \vertof{\asys}$ is an injective
partial function that assigns source labels to vertices of $\asys$
such that $\ptypeof{\asrc}=\vlab(\sources(\asrc))$, for each
$\asrc\in\dom{\sources}$, \ie the source label of a vertex has the
process type of that vertex. We say that a vertex $v\in\vertof{\asys}$
is the $\asrc$-source of $\open{\asys}$ if $\sources(\asrc)=v$. The
\emph{type} of $\open{\asys}$ is
$\typeof{\open{\asys}}\isdef\dom{\sources}$, \ie the set of source
labels that occurs in $\asys$. Since systems (\defref{def:system}) are
in fact open systems of empty type, we blur the distinction and refer
to open systems as systems from now on.

The algebra $\algof{S}$ of systems (\figref{fig:alg-zoo}) interprets
the \hrtext{} signature over the set $\systems$ of systems, as
follows: \begin{compactitem}[-]
\item $\sgraph{a}{\asrc_1}{\asrc_2}{\algof{S}}$ is the system having a
  single $a$-labeled edge between its two vertices labeled with
  process types $\ptypeof{\asrc_1}$ and $\ptypeof{\asrc_2}$, that are
  the $\asrc_1$- and $\asrc_2$-sources, respectively, for each edge
  label $a\in\ealpha$ and source labels $\asrc_1,\asrc_2\in\sourcelabels$.
\item $\restrict{\slabs}{\algof{S}}(\asys,\sources) \isdef
  (\asys,\proj{\sources}{\slabs})$ removes the source labels that are
  not in $\slabs \finsubseteq \sourcelabels$,
\item
  $\rename{\alpha}{\algof{S}}(\asys,\sources) \isdef
  (\asys,\sources\circ\alpha^{-1})$ relabels the sources according to
  $\alpha$; note that $\sources\circ\alpha^{-1}$ is an injective
  partial mapping, since $\alpha$ is a finite permutation of
  $\sourcelabels$,
\item $(\asys_1,\sources_1) \pop{\algof{S}} (\asys_2,\sources_2)$ is
  the disjoint union of $(\asys_1,\sources_1)$ and
  $(\asys_2,\sources_2)$, followed by: \begin{compactitem}[*]
  \item fusion of each pair of common $\asrc$-sources $v_i \in
    \vertof{\asys_i}$, for $\asrc\in\typeof{\asys_1} \cap
    \typeof{\asys_2}$ and $i=1,2$, into a $\asrc$-source labeled with
    $\ptypeof{\asrc}$,
  \item fusion of all the edges with the same label and endpoints into
    a single edge with the same label and endpoints.
  \end{compactitem}
\end{compactitem}
Every other \hrtext{} algebra considered in the rest of this paper
will be defined from $\algof{S}$ via an
homomorphism. \figref{fig:alg-zoo} shows the diagram of these algebras
and homomorphisms. Each such homomorphism $h$ (except for $\psi$,
which has a trivial definition) has a reference to a lemma proving
that $\kerof{h}$ is an \hrtext{} congruence.

\begin{figure}[t!]
  \vspace*{-\baselineskip}
  \begin{center}
    \begin{tikzpicture}
      \node[label=150:{\scriptsize systems}] (S) at (0,0) {$\algof{S}$};
      \node[label=90:{\scriptsize behaviors}] (B) at ($(S) + (3,0)$) {$\algof{B}$};
      \node[label=30:{\scriptsize folded behaviors}] (Bfold) at ($(B) + (3,0)$) {$\absof{\algof{B}}$};
      \node[draw=black,shape=circle,inner sep=1pt,label=-30:{\scriptsize folded nets}] (Bfin) at ($(Bfold) + (0,-1.5)$) {$\finabsof{\algof{B}}$};
      \node[draw=black,shape=circle,inner sep=1pt,label=-150:{\scriptsize flows}] (F) at ($(S) + (0,-1.5)$) {$\algof{F}$};

      \draw[->] (S) -- node[above]{$\abeh$} node[below]{\lemref{lemma:cong-sys-beh}} (B);
      \draw[->] (B) -- node[above]{$\afold$} node[below]{\lemref{lemma:cong-beh-fold}} (Bfold);
      \draw[->] (Bfold) -- node[right]{$\psi$} node[left]{\eqnref{eq:drop}} (Bfin);
      \draw[->] (S) -- node[left]{$\eta$} node[right]{\lemref{lemma:cong-sys-flow}} (F);
    \end{tikzpicture}
    \vspace*{-\baselineskip}
    \caption{Homomorphisms between the \hrtext{} algebras used in this
      paper. The circled algebras are the finite and effectively
      computable ones.}
    \label{fig:alg-zoo}
  \end{center}
  \vspace*{-2\baselineskip}
\end{figure}
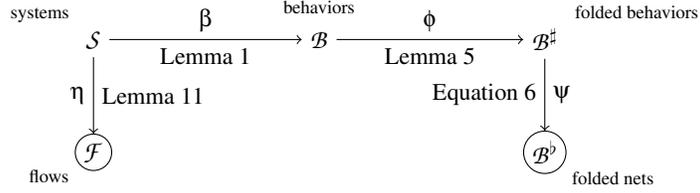

An \emph{open behavior} is a pair
$\open{\amarkednet}\isdef(\amarkednet,\sources)$, where $\amarkednet$
is a behavior and $\sources : \sourcelabels \times \placeof{\ptypes}
\finmap \placeof{\amarkednet}$ is an injective partial function
assigning pairs $(\asrc,q)$ to places of $\amarkednet$, where $\asrc$
is a source label and $q$ is a place from some process type in
$\ptypes$. A place $r\in\placeof{\amarkednet}$ is a $(\asrc,q)$-source
of $\open{\amarkednet}$ if $\sources(\asrc,q)=r$ and
$\typeof{\open{\amarkednet}} \isdef \dom{\sources}$ denotes the type
of $\open{\amarkednet}$. Since behaviors (\defref{def:behavior}) are
actually open behaviors of empty type, we blur the distinction and
refer to open behaviors as behaviors, when no confusion arises.

To define the algebra of behaviors, we extend the $\abeh$ function
introduced by \defref{def:behavior} to (open) systems, as follows. The
behavior of the system $\open{\asys}=(\asys,\sources)$ is
$\behof{\open{\asys}}\isdef(\behof{\asys},\overline{\sources})$, where
$\behof{\asys}$ is given in \defref{def:behavior} and the source
labeling $\overline{\sources}$ is $\overline{\sources}(\asrc,q) \isdef
(q,\sources(\asrc))$ if $\sources(\asrc)$ is defined and
$(q,\sources(\asrc)) \in \placeof{\amarkednet}$, and undefined
otherwise, for all $\asrc\in\sourcelabels$ and $q \in
\placeof{\ptypes}$.

\begin{lemmaE}[][category=proofs]\label{lemma:cong-sys-beh}
  $\kerof{\abeh}$ is a \hrtext{} congruence. 
\end{lemmaE}
\begin{proofE}
  We prove a stronger statement, namely that $\abeh$ is
  injective. Then, $\kerof{\abeh}$ becomes the equality, which is
  trivially a congruence for any signature. Let $\open{\asys}_i =
  (\asys_i,\sources_i)$ be open systems and
  $\behof{\open{\asys}_i}=((\anet_i,\amark_{0i}),\overline{\sources}_i)$
  be open behaviors, for $i=1,2$, such that $\anet_1=\anet_2$,
  $\amark_{01}=\amark_{02}$ and
  $\overline{\sources}_1=\overline{\sources}_2$. We show that
  $\asys_1=\asys_2$ below: \begin{itemize}
  \item $\vertof{\asys_1}=\vertof{\asys_2}$: Suppose, for a
    contradiction, that there exists a vertex $v \in \vertof{\asys_1}
    \setminus \vertof{\asys_2}$. Since
    $\placeof{\vlabof{\asys_1}(v)}\neq\emptyset$
    (\defref{def:process-type}), there exists a place
    $q\in\placeof{\vlabof{\asys_1}(v)}$. Since
    $\placeof{\anet_1}=\placeof{\anet_2}$, by 
    \defref{def:behavior}, we obtain $q\in\placeof{\vlabof{\asys_2}}(v)$
    and $v \in \vertof{\asys_2}$, contradiction. Hence
    $\vertof{\asys_1} \subseteq \vertof{\asys_2}$ and the proof in the
    other direction is symmetric.
  \item $\vlabof{\asys_1}=\vlabof{\asys_2}$: Let $v \in
    \vertof{\asys_1}$ ($= \vertof{\asys_2}$, by the previous point)
    and let $q \in \placeof{\vlabof{\asys_1}(v)}$ be a place. Since
    $\placeof{\anet_1}=\placeof{\anet_2}$, by definition
    \defref{def:behavior}, we obtain $q \in
    \placeof{\vlabof{\asys_2}(v)}$, hence
    $\vlabof{\asys_1}(v)=\vlabof{\asys_2}(v)$, by
    \defref{def:process-type}, \ie because the process types are
    assumed to have pairwise disjoint sets of places. Since the choice
    of $v$ was arbitrary, we obtain
    $\vlabof{\asys_1}=\vlabof{\asys_2}$. 
  \item $\edgeof{\asys_1} = \edgeof{\asys_2}$: since
    $\transof{\anet_1}=\transof{\anet_2}$,
    $\vertof{\asys_1}=\vertof{\asys_2}$ and
    $\vlabof{\asys_1}=\vlabof{\asys_2}$, by the previous points, we
    obtain $\edgeof{\asys_1} = \edgeof{\asys_2}$, by
    \defref{def:behavior}. 
  \end{itemize}
  To prove $\sources_1 = \sources_2$, let $\asrc \in \sourcelabels$ be
  a source label and $q \in \placeof{\ptypes}$ be a place from some
  process type, such that $\overline{\sources}_1(\asrc,q)$ is
  defined. Then, $\overline{\sources}_1(\asrc,q) =
  (q,\sources_1(\asrc))$. Since $\overline{\sources}_1 =
  \overline{\sources}_2$, we obtain that
  $\overline{\sources}_2(\asrc,q)$ is defined and
  $\sources_1(\asrc)=\sources_2(\asrc)$. Since the choice of $\asrc$
  was arbitrary, we obtain $\sources_1 = \sources_2$. \qed
\end{proofE}
\begin{proofSketch}
  In fact $\abeh$ is injective making $\kerof{\abeh}$ the equality
  (thus a congruence).
  \qed
\end{proofSketch}

We introduce an algebra $\algof{B}$ of behaviors
(\figref{fig:alg-zoo}), whose domain and interpretation of \hrtext{}
function symbols are defined as in \propref{prop:cong-homo}. Since
$\kerof{\abeh}$ is a \hrtext{} congruence
(\lemref{lemma:cong-sys-beh}), it follows that $\abeh$ is a
homomorphism between $\algof{S}$ and $\algof{B}$. As we discuss next,
a grammar that describes a parameterized system $\parsys=\set{\asys_1,
  \asys_2, \ldots}$ can be reused to describe its set of behaviors
$\behof{\parsys}$.

\subsection{Grammars}

A \emph{grammar} is a pair $\grammar=(\nonterm,\rules)$ consisting of
a finite set $\nonterm$ of \emph{nonterminals} and a finite set
$\rules$ of \emph{rules} of the form, either (1) $X \rightarrow
\rho[X_1,\ldots,X_n]$, where $X,X_1,\ldots,X_n \in \nonterm$ are
nonterminals and $\rho$ is a term whose only variables are
$X_1,\ldots,X_n$, or (2) $\rightarrow X$, where $X \in \nonterm$; the
rules of this form are called \emph{axioms}. Given terms $\theta$ and
$\eta$, a \emph{step} $\theta \step{\grammar} \eta$ obtains $\eta$
from $\theta$ by replacing an occurrence of a nonterminal $X$ with the
term $\rho$, for some rule $X \rightarrow \rho$ of $\grammar$. An
$X$-\emph{derivation} is a sequence of steps starting with a
nonterminal $X$. The derivation is \emph{complete} if it ends in a
ground term. Let $\alangof{X}{\algof{A}}{\grammar} \isdef
\{\theta^\algof{A} \mid X \step{\grammar}^* \theta \text{ is a
  complete derivation}\}$ and $\alangof{}{\algof{A}}{\grammar} \isdef
\bigcup_{\rightarrow X \in \rules} \alangof{X}{\algof{A}}{\grammar}$
be the \emph{language} of $\grammar$ in the algebra $\algof{A}$.

The following result, also known as the \emph{Filtering Theorem},
allows to build a grammar for the intersection between the language of
a grammar and a \emph{recognizable set}, \ie the image of a finite
set via an inverse homomorphism~\cite[Theorem 3.88]{courcelle_engelfriet_2012}:

\begin{theorem}\label{thm:filtering}
  Let $\signature=\set{\aop_1, \aop_2, \ldots}$ be a signature,
  $\algof{A}=(\universeOf{A}, \aop^\algof{A}_1, \aop^\algof{A}_2,
  \ldots)$ and $\algof{B}=(\universeOf{B}, \aop^\algof{B}_1,
  \aop^\algof{B}_2, \ldots)$ be $\signature$-algebras, such that
  $\universeOf{B}$ is finite, and $h$ be a homomorphism between
  $\algof{A}$ and $\algof{B}$. For each $\signature$-grammar
  $\grammar$ and set $C \subseteq \universeOf{B}$, one can build a
  grammar $\grammar_{h,C}$ such that
  $\alangof{}{\algof{A}}{\grammar_{h,C}} =
  \alangof{}{\algof{A}}{\grammar} \cap h^{-1}(C)$.
\end{theorem}
\noindent The construction of the filtered grammar $\grammar_{h,C}$ is
effective: for each rule $X \rightarrow \rho[X_1,\ldots,X_n]$ of
$\grammar$ and each sequence of elements $b, b_1, \ldots, b_n \in
\universeOf{B}$ such that $b = \rho^\algof{B}(b_1,\ldots,b_n)$, the
grammar $\grammar_{h,C}$ has a rule $X^b \rightarrow
\rho[X_1^{b_1},\ldots,X_n^{b_n}]$; the axioms of $\grammar_{h,C}$ are
$\rightarrow X^c$, for each axiom $\rightarrow X$ of $\grammar$ and
each element $c \in C$.

The grammars from the rest of the paper use the \hrtext{} signature to
describe parameterized systems. The following example shows two
grammars that specify parameterized systems with chain-like and
star-like network topologies, as in \figref{fig:proc-typ1} (b):

\begin{example}\label{ex:parameterized-grammars}
  The grammar $\grammar_{Chain}$ below defines systems with chain-like network topologies
  and at least three processes including the controller
  (\figref{fig:proc-typ1} (b) top):

  \vspace*{-.5\baselineskip}
  \begin{center}
  \begin{minipage}{0.6\textwidth}
  \begin{align*}
    & \to C \\
    C &\to \restrict{\set{\sigma_1}}{} ((relC,get)_{(\sigma_3,\sigma_2)} \pop{} (rel,get)_{(\sigma_2,\sigma_1)}) \\
    C &\to \restrict{\set{\sigma_1}}{}
      \rename{\sigma_1 \leftrightarrow \sigma_2}{}
      (C \pop{} (rel,get)_{(\sigma_1,\sigma_2)})
  \end{align*}
  \end{minipage}
  \begin{minipage}{0.3\textwidth}
    \vspace*{1.5\baselineskip}
    \scalebox{0.6}{
      \begin{tikzpicture}[node distance=1.5cm]
      \tikzstyle{every state}=[inner sep=3pt,minimum size=20pt]
      \node(v0)[gnode,draw=black,label=-100:{$Cont$}]{};
      \node(v1)[gnode,draw=black,label=-90:{$Proc$}] at ($(v0) + (2,0)$) {};
      \node(v2)[gnode,draw=black,label=-80:{$Proc$},label=30:{$\sigma_1$}] at ($(v1) + (2,0)$) {};
      \path (v0) edge [->,thick,line width=1pt] node[above]{$(\mathit{relC},\mathit{get})$} (v1);
      \path (v1) edge [->,thick,line width=1pt] node[above]{$(\mathit{relC},\mathit{get})$} (v2);
      \end{tikzpicture}
    } \\
    \scalebox{0.6}{
      \begin{tikzpicture}[node distance=1.5cm]
      \tikzstyle{every state}=[inner sep=3pt,minimum size=20pt]
      \node(v0)[shape=rectangle,draw=black,minimum width=1cm,minimum height=5mm] {C};
      \node(v1)[gnode,draw=black] at ($(v0) + (5mm,0)$) {};
      \node(v2)[gnode,draw=black,label=-80:{$Proc$},label=30:{$\sigma_1$}] at ($(v1) + (2,0)$) {};
      \path (v1) edge [->,thick,line width=1pt] node[above]{$(\mathit{rel},\mathit{get})$}(v2);
      \end{tikzpicture}
    }
  \end{minipage}
  \end{center}
  
  \noindent The right-hand sides of the last two rules correspond to
  the graphs on the right. Similarly, the $\grammar_{Star}$ grammar
  below defines systems with star-shaped network topology
  and at least two processes including the controller
  (\figref{fig:proc-typ1} (b) bottom): 

  \begin{center}
  \begin{minipage}{0.7\textwidth}
    \begin{align*}
      & \to Z \\
      Z &\to \restrict{\set{\sigma_1}}{}
        ((relC,get)_{(\sigma_1,\sigma_2)}
        \pop{\algof{S}} (getC,rel)_{(\sigma_1,\sigma_2)}) \\
      Z &\to \restrict{\set{\sigma_1}}{}
        (Z \pop{}
        (relC,get)_{(\sigma_1,\sigma_2)}
        \pop{\algof{S}} (getC,rel)_{(\sigma_1,\sigma_2)})
    \end{align*}
  \end{minipage}
  \vspace*{1.5\baselineskip}
  \begin{minipage}{0.2\textwidth}
    \scalebox{0.6}{
      \begin{tikzpicture}[node distance=1.5cm]
      \tikzstyle{every state}=[inner sep=3pt,minimum size=20pt]
      \node(v0)[gnode,draw=black,label=-100:{$Cont$},label=150:{$\sigma_1$}]{};
      \node(v1)[gnode,draw=black,label=-80:{$Proc$}] at ($(v0) + (2,0)$) {};
      \path (v0) edge [->,thick,line width=1pt,bend left=40] node[above]{$(\mathit{relC},\mathit{get})$}(v1);
      \path (v0) edge [->,thick,line width=1pt,bend right=40] node[below]{$(\mathit{getC},\mathit{rel})$}(v1);
      \end{tikzpicture}
    } \\
    \scalebox{0.6}{
      \begin{tikzpicture}[node distance=1.5cm]
      \tikzstyle{every state}=[inner sep=3pt,minimum size=20pt]
      \node(v0)[shape=rectangle,draw=black,minimum width=1cm,minimum height=5mm] {Z};
      \node(v1)[gnode,draw=black,label=0:{$\sigma_1$}] at ($(v0) + (5mm,0)$) {};
      \node(v2)[gnode,draw=black,label=-80:{$Proc$}] at ($(v1) + (2,0)$) {};
      \path (v1) edge [->,thick,line width=1pt,bend left=40] node[above]{$(\mathit{relC},\mathit{get})$}(v2);
      \path (v1) edge [->,thick,line width=1pt,bend right=40] node[below]{$(\mathit{getC},\mathit{rel})$}(v2);
      \end{tikzpicture}
    }
  \end{minipage}
  \end{center}
  \vspace*{-2\baselineskip}
\end{example}

The following argument will be repeated several times: if $\algof{A}$
is some \hrtext{} algebra related to $\algof{S}$ by a homomorphism
$h$, then $h(\alangof{}{\algof{S}}{\grammar}) =
\alangof{}{\algof{A}}{\grammar}$, for each grammar $\grammar$ written
using the \hrtext{} signature. Moreover, if the domain of $\algof{A}$
is finite, the language $\alangof{}{\algof{A}}{\grammar}$ is finite
and effectively computable by a finite Kleene iteration, provided that
the interpretation of each function symbol from the \hrtext{}
signature in $\algof{A}$ is effectively computable.  The finite and
effectively computable algebras in question are circled in
\figref{fig:alg-zoo}.

As a consequence of \lemref{lemma:cong-sys-beh}, a grammar that
specifies a parameterized system defines also its set of behaviors:

\begin{lemmaE}[][category=proofs]\label{prop:sys-beh}
  For each \hrtext{} grammar $\grammar$, we have
  $\behof{\alangof{}{\algof{S}}{\grammar}} =
  \alangof{}{\algof{B}}{\grammar}$.
\end{lemmaE}
\begin{proofE}
  By \lemref{lemma:cong-sys-beh}, $\abeh$ is a homomorphism between
  $\algof{S}$ and $\algof{B}$, hence
  $\behof{\theta^\algof{S}}=\theta^\algof{B}$, for each ground term
  $\theta$. The result follows, by the definition of the language of a
  grammar in a given algebra. \qed
\end{proofE}
\begin{proofSketch}
  Easy application of \lemref{lemma:cong-sys-beh}
  and \propref{prop:cong-homo}.
  \qed
\end{proofSketch}

We consider the problems of reachability and coverability for
parameterized systems specified by grammars and give the formal
definitions of the decision problems:

\begin{definition}\label{def:grammar-parameterized-verif}
  The $\paramreach{\grammar}{\mathcal{Q}}{\amark}$
  (\resp $\paramcover{\grammar}{\mathcal{Q}}{\amark}$) problem takes
  in input a grammar $\grammar$, a set of places $\mathcal{Q}
  \subseteq \placeof{\ptypes}$, a mapping $\amark : \mathcal{Q}
  \rightarrow\nat$ and asks for the existence of a system
  $\asys\in\alangof{}{\algof{S}}{\grammar}$ and marking
  $\overline{\amark}\in\reach{\behof{\asys}}$
  (\resp $\overline{\amark}\in\cover{\behof{\asys}}$) where 
  \(\sum_{v \in \vertof{\asys}:q \in \placeof{\vlabof{\asys}(v)}}
  \overline{\amark}(q,v)=\amark(q) \text{, for all } q \in
  \mathcal{Q}.\)
\end{definition}
Unsurprisingly, both problems are undecidable, even for simple
grammars defining parameterized systems with chain-like topologies
(see \exref{ex:parameterized-grammars}), when the structure of the
process types is unrestricted.

\begin{theoremE}[][category=proofs]\label{thm:undecidability}
  The $\paramreach{}{}{}$ and $\paramcover{}{}{}$ problems are
  undecidable.
\end{theoremE}
\begin{proofE}
  \newcommand{\Instrs}{\mathit{Instrs}}
\newcommand{\Count}{\mathit{Count}}
\newcommand{\Minsk}{\mathit{Minsk}}

To show that the two mentioned problems are undecidable,
we propose a reduction from the halting problem for a (two-counter) Minsky machine \cite{minsky-computation-67}.
A Minsky machine is a finite sequence of instructions $\Instrs$ manipulating two natural variables $c1$ and $c2$,
called counters, and each instruction has one of the following form:
\begin{enumerate}
\item $\ell:ci=ci+1;\mathtt{goto}~\ell';$
\item $\ell:\mathtt{if}~ci==0?~\mathtt{goto}~ \ell'; \mathtt{else}~ci=ci-1;\mathtt{goto}~\ell'';$
\end{enumerate}
where $\ell$.$\ell'$ and $\ell''$ are labels.
All labels are followed by an instruction except $\ell_f$, the final label.
The halting problem asks if the Minsky machine starting from an initial label $\ell_0$
with the two counters set to $0$ halts, \ie reaches the label $\ell_f$.

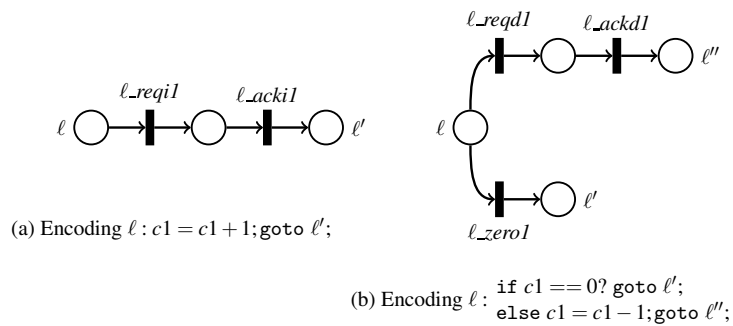
\begin{figure}[htbp]
\begin{center}
\scalebox{0.9}{
\begin{tikzpicture}[node distance=1.5cm]
  \tikzstyle{every state}=[inner sep=3pt,minimum size=20pt]

  \node(Q)[petri-p,draw=black,label=180:{$\ell$}]{};
  \node(ReqI1)[petri-t2,draw=black,fill=black,right of=Q,label=90:{$\ell\_\mathit{reqi1}$},xshift=-2em]{};
  \node(aux1)[petri-p,draw=black,right of=ReqI1,xshift=-2em]{};
  \node(AckI1)[petri-t2,draw=black,fill=black,right of=aux1,label=90:{$\ell\_\mathit{acki1}$},xshift=-2em]{};
  \node(R)[petri-p,draw=black,right of=AckI1,xshift=-2em,label=00:{$\ell'$}]{};

  \path (Q) edge [->,thick,line width=1pt](ReqI1);
  \path (ReqI1) edge [->,thick,line width=1pt](aux1);
  \path (aux1) edge [->,thick,line width=1pt](AckI1);
  \path (AckI1) edge [->,thick,line width=1pt](R);

 \node(Q2)[petri-p,draw=black,label=180:{$\ell$},right of=R,xshift=2em]{};
  \node(ReqD1)[petri-t2,draw=black,fill=black,above right of=Q2,label=90:{$\ell\_\mathit{reqd1}$},xshift=-2em]{};
  \node(aux2)[petri-p,draw=black,right of=ReqD1,xshift=-2em]{};
  \node(AckD1)[petri-t2,draw=black,fill=black,right of=aux2,label=90:{$\ell\_\mathit{ackd1}$},xshift=-2em]{};
  \node(R2)[petri-p,draw=black,right of=AckD1,xshift=-2em,label=00:{$\ell''$}]{};
  \node(Zero1)[petri-t2,draw=black,fill=black,below right of=Q2,label=-90:{$\ell\_\mathit{zero1}$},xshift=-2em]{};
   \node(S)[petri-p,draw=black,right of=Zero1,xshift=-2em,label=00:{$\ell'$}]{};

  \path (Q2) edge [->,thick,line width=1pt,out=90,in=180](ReqD1);
  \path (ReqD1) edge [->,thick,line width=1pt](aux2);
  \path (aux2) edge [->,thick,line width=1pt](AckD1);
  \path (AckD1) edge [->,thick,line width=1pt](R2);
  \path (Q2) edge [->,thick,line width=1pt,out=-90,in=180](Zero1);
  \path (Zero1) edge [->,thick,line width=1pt](S);

  \node(lab1)[below of=ReqI1,xshift=1em]{(a) Encoding $\ell:c1=c1+1;\mathtt{goto}~\ell';$};
   \node(lab1)[below of=Zero1,xshift=2em]{(b) Encoding $\ell: \begin{array}{l}\mathtt{if}~c1==0?~\mathtt{goto}~ \ell';\\ \mathtt{else}~c1=c1-1;\mathtt{goto}~\ell'';\end{array}$};

\end{tikzpicture}
}
\end{center}
\vspace*{-\baselineskip}
\caption{Simulating a Minsky machine with process type $\Minsk$}
\label{fig:proc-minsk}
\end{figure}

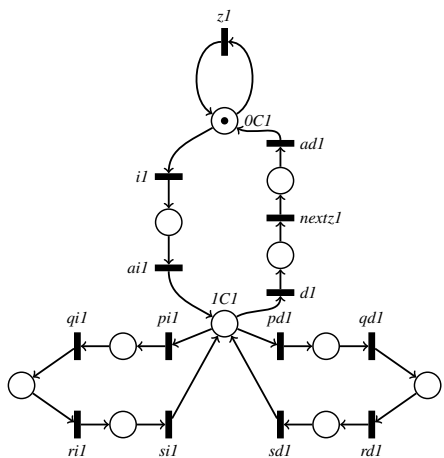
\begin{figure}[htbp]
\begin{center}
\scalebox{0.7}{
\begin{tikzpicture}[node distance=1.5cm]
  \tikzstyle{every state}=[inner sep=3pt,minimum size=20pt]
  
\node(0C1)[petri-p,draw=black,label=00:{$\mathit{0C1}$}]{};
\node(Z1)[petri-t2,draw=black,fill=black,above of=0C1,label=90:{$\mathit{z1}$}]{};
\node(I1)[petri-t,draw=black,fill=black,below left of=0C1,label=180:{$\mathit{i1}$}]{};
\node(D1)[petri-t,draw=black,fill=black,below right of=0C1,label=0:{$\mathit{ad1}$},yshift=2em]{};
\node(0to1)[petri-p,draw=black,below of=I1,yshift=2em]{};
\node(1to0)[petri-p,draw=black,below of=D1,yshift=2.5em]{};
\node(AI1)[petri-t,draw=black,fill=black,below of=0to1,label=180:{$\mathit{ai1}$},yshift=2em]{};
\node(NZ1)[petri-t,draw=black,fill=black,below of=1to0,label=0:{$\mathit{nextz1}$},yshift=2.5em]{};
\node(1to0b)[petri-p,draw=black,below of=NZ1,yshift=2.5em]{};
\node(AD1)[petri-t,draw=black,fill=black,below of=1to0b,label=0:{$\mathit{d1}$},yshift=2.5em]{};
\node(1C1)[petri-p,draw=black,below right of=AI1,label=90:{$\mathit{1C1}$}]{};
\node(PI1)[petri-t2,draw=black,fill=black,below left of=1C1,label=90:{$\mathit{pi1}$},yshift=2em]{};
\node(pi1toqi1)[petri-p,draw=black,left of=PI1,xshift=2em]{};
\node(QI1)[petri-t2,draw=black,fill=black,left of=pi1toqi1,label=90:{$\mathit{qi1}$},xshift=2em]{};
\node(qi1tori1)[petri-p,draw=black,below left of=QI1,yshift=1em]{};
\node(RI1)[petri-t2,draw=black,fill=black,below right of=qi1tori1,label=-90:{$\mathit{ri1}$},yshift=1em]{};
\node(ri1tosi1)[petri-p,draw=black,right of=RI1,xshift=-2em]{};
\node(SI1)[petri-t2,draw=black,fill=black,right of=ri1tosi1,label=-90:{$\mathit{si1}$},xshift=-2em]{};

\node(PD1)[petri-t2,draw=black,fill=black,below right of=1C1,label=90:{$\mathit{pd1}$},yshift=2em]{};
\node(pd1toqd1)[petri-p,draw=black,right of=PD1,xshift=-2em]{};
\node(QD1)[petri-t2,draw=black,fill=black,right of=pd1toqd1,label=90:{$\mathit{qd1}$},xshift=-2em]{};
\node(qd1tord1)[petri-p,draw=black,below right of=QD1,yshift=1em]{};
\node(RD1)[petri-t2,draw=black,fill=black,below left of=qd1tord1,label=-90:{$\mathit{rd1}$},yshift=1em]{};
\node(rd1tosd1)[petri-p,draw=black,left of=RD1,xshift=2em]{};
\node(SD1)[petri-t2,draw=black,fill=black,left of=rd1tosd1,label=-90:{$\mathit{sd1}$},xshift=2em]{};

\path (0C1) edge [->,thick,line width=1pt,out=30,in=0](Z1);
\path (Z1) edge [->,thick,line width=1pt,out=180,in=150](0C1);
\path (0C1) edge [->,thick,line width=1pt,out=-150,in=90](I1);
\path (D1) edge [->,thick,line width=1pt,out=90,in=-30](0C1);
\path (I1) edge [->,thick,line width=1pt](0to1);
\path (1to0b) edge [->,thick,line width=1pt](NZ1);
\path (NZ1) edge [->,thick,line width=1pt](1to0);
\path (1to0) edge [->,thick,line width=1pt](D1);
\path (0to1) edge [->,thick,line width=1pt](AI1);
\path (AD1) edge [->,thick,line width=1pt](1to0b);
\path (AI1) edge [->,thick,line width=1pt,out=-90,in=150](1C1);
\path (1C1) edge [->,thick,line width=1pt,out=30,in=-90](AD1);

\path (1C1) edge [->,thick,line width=1pt](PI1);
\path (PI1) edge [->,thick,line width=1pt](pi1toqi1);
\path (pi1toqi1) edge [->,thick,line width=1pt](QI1);
\path (QI1) edge [->,thick,line width=1pt](qi1tori1);
\path (qi1tori1) edge [->,thick,line width=1pt](RI1);
\path (RI1) edge [->,thick,line width=1pt](ri1tosi1);
\path (ri1tosi1) edge [->,thick,line width=1pt](SI1);
\path (SI1) edge [->,thick,line width=1pt](1C1);

\path (1C1) edge [->,thick,line width=1pt](PD1);
\path (PD1) edge [->,thick,line width=1pt](pd1toqd1);
\path (pd1toqd1) edge [->,thick,line width=1pt](QD1);
\path (QD1) edge [->,thick,line width=1pt](qd1tord1);
\path (qd1tord1) edge [->,thick,line width=1pt](RD1);
\path (RD1) edge [->,thick,line width=1pt](rd1tosd1);
\path (rd1tosd1) edge [->,thick,line width=1pt](SD1);
\path (SD1) edge [->,thick,line width=1pt](1C1);

 \node[petri-tok] at (0C1) {};

\end{tikzpicture}
}
\end{center}
\vspace*{-\baselineskip}
\caption{Process type $\Count_1$ to simulate the first counter}
\label{fig:proc-count1}
\end{figure}

We fix a Minsky machine given by its set of instructions $\Instrs$.
In order to simulate this model, we consider a family of systems which all have the following shape:
there is a central vertex which simulates the instructions of the machine,
on the right side of this vertex there is a sequence of vertices to encode the value of the counter $c1$
and on its left side a sequence of vertices to encode the value of the counter $c2$.
The central node has process type $\Minsk$, it is represented on \figref{fig:proc-minsk},
where we notice that there is a place for each instruction label
and furthermore the initial marking only puts a token in the place $\ell_0$.
The vertices to encode the value of counter $c1$ are labeled by the process type $\Count_1$
depicted on \figref{fig:proc-count1} and the vertices to encode the value of counter $c2$
are labeled the process type $\Count_2$ which is similar to $\Count_1$.
We remark that in the process types $\Count_1$, there are two specific places $0C1$ and $1C1$,
and that initially the place $0C1$ is marked.
During a simulation of the machine with a behavior,
the number of vertices of process type $\Count_1$ for which a token is in place $1C1$
corresponds to the current value of $c1$.
To ease the explanation, we shall call the central vertex,
the controller and the vertices labeled by $\Count_i$ the counting processes for $c_i$, for $i=1,2$.
The grammar $\grammar_\Minsk$ producing the desired family of systems is given on
\figref{fig:grammmar-minsk} and an example of a system belonging to
$\alangof{}{\algof{S}}{\grammar_{\Minsk}}$ is depicted on \figref{fig:undec4}
(we omit the label on the edges to ease presentation).

We now provide the main ingredients of the simulation.
At the beginning the controller is in state $\ell_0$
and all the counting processes are in states $0C1$ or $0C2$ according to their process type,
$\Count_1$ or $\Count_2$.
At any `big step' (\ie when forgetting intermediate steps needed for the correctness) of the simulation,
when the controller is in a state corresponding to a label of the Minsky machine,
the sequence of the counting processes of type $\Count_1$ will always have the following form:
it begins with a sequence of processes in state $1C1$ and ends with a sequence of processes in state $0C1$
(the same property holds for the counting processes of type $\Count_2$ considering the state $1C2$ and $0C2$).
Moreove, the number of processes in state $1C1$ (\resp $1C2$) represents the value of the first (\resp second)
counter at this stage of the simulation. We have then the following behavior:
\begin{itemize}
    \item when the controller simulates an instruction of the form $\ell:c1=c1+1;\mathtt{goto}~\ell';$
        it begins to request for an increment with the transition $\ell\_reqi1$ and it waits for its acknowledgment with $\ell\_acki1$.
        The request is then transmitted by all the processes in state $1C1$
        to the first counting process in state $0C1$ which changes its state and acknowledges it moving to $1C1$,
        the acknowledgment being transmitted back to the controller.
        If there is no process of type $\Count_1$ in state $0C1$, the simulation is stuck,
        it means that the system is not `big' enough to simulate correctly the Minsky machine.
    \item when the controller simulates an instruction of the form
        $\ell:\mathtt{if}~c1==0?~\mathtt{goto}~ \ell';\linebreak[0] \mathtt{else}~c1=c1-1;\mathtt{goto}~\ell'';$
        and the counter $c1$ is equal to $0$,
        this means that there is no counting process of type $\Count_1$ in state $1C1$,
        then the controller takes the transition $\ell\_zero1$ together with the first counting process of type
        $\Count_1$ which takes the transition $z1$.
    \item when the controller simulates an instruction of the form
        $\ell:\mathtt{if}~c1==0?~\mathtt{goto}~ \ell';\linebreak[0] \mathtt{else}~c1=c1-1;\mathtt{goto}~\ell'';$
        and the counter $c1$ has a strictly positive value, then, as for the increment,
        the controller requests a decrement with the transition $\ell\_decqi1$,
        this request is transmitted to the last counting process of type $\Count_1$ in state $1C1$
        (there is necessarily one) and if this last process has a right neighbor whose state is $0C1$,
        then it moves to $0C1$ acknowledging the decrement,
        this acknowledgment being transmitted back to the controller which goes in state $\ell''$.
        Note that each counting process in state $1C1$ can choose nondeterministically
        whether it is or not the last one in state $1C1$ by taking the transition $pd1$
        (it transmits the decrement request) or $d1$ (it chooses it is the last one),
        but if it makes a bad choice, the simulation will be stuck.
    \item Finally, we have that the machine halts if and only if there is a system and an execution
        in its associated behavior where the controller ends in $\ell_f$.
\end{itemize}

We now provide an example on how the transmission of the information is performed in a system.
We consider a Minsky machine with two instructions
$\ell_0:c1=c1+1;\mathtt{goto}~\ell';$ and $\ell':c1=c1+1;\mathtt{goto}~\ell'';$.
\figref{fig:behav} presents a partial representation of a behavior simulating these instructions
with two counter processes of type $\Count_1$
(to make the figure readable, we put a dotted line between two transitions when they are synchronized).
First, the  controller simulates the instruction $\ell_0:c1=c1+1;\mathtt{goto}~\ell';$,
it takes the transition $\ell_0\_reqi1$ which says that it requests an increment,
this transitions is `synchronized' with two transitions $i1$ and $pi1$ of the first counting process
on its right and since the state of this process is initially $0C1$
then this first process takes the transition $i1$.
Afterwards, the controller waits for an acknowledgment of its request for an increment
which will arrive when it takes the transition $\ell_0\_acki1$.
This latter transition is synchronized with the two transitions $ai_1$ and $si_1$
of the first counting process.
In our case, the pair $(\ell_0\_acki1,ai_1)$ takes place and the controller ends in state $\ell'$
and the first counting process in $1C1$.
The controller can now simulate $\ell':c1=c1+1;\mathtt{goto}~\ell'';$.
It takes the transition $\ell'\_reqi1$, but this time the first counting process (which is in state $1C1$)
takes the transition $pi1$ signifying it transmits the request to its neighbor (if there is one).
Afterwards, if the first counting process has a neighbor, it  fires the transition $qi1$
(which synchronizes with the transitions $i1$ and $pi1$ of its right neighbor),
and its neighbor which is in state $0C1$ takes the transition $i1$,
they then both move with the pair $(ri1,ai1)$ and the second counting process arrives in state $1C1$
whereas the first counting process is able to acknowledge the fact that the increment has been performed
to the controller process with the synchronization $(\ell'\_acki1,si_1)$.
We end in a configuration where the controller is in state $\ell''$
and the two first counting processes of type $\Count_1$ are in state $1C1$.

\begin{figure}
\begin{minipage}{1\textwidth}
\begin{center}
  $$
  \begin{array}{ccl}
     & \rightarrow &X_1\\
    X_1 & \rightarrow & X_2\\
    X_1 & \rightarrow & \restrict{\slabs}{\algof{S}}(\rename{\alpha_1}{\algof{S}}( X_1 \pop{\algof{S}}Y_1)) \\
    X_2 & \rightarrow & Init\\
    X_2 & \rightarrow & \restrict{\slabs}{\algof{S}}(\rename{\alpha_2}{\algof{S}}( Y_2 \pop{\algof{S}}X_2))\\
    
    Init & \rightarrow & \bigoplus^{\algof{S}}_{\ell \in L_{inc1}}\big (\sgraph{(\ell\_reqi1,i1)}{\asrc_0}{\asrc_1}{\algof{S}} \pop{\algof{S}} \sgraph{(\ell\_reqi1,pi1)}{\asrc_0}{\asrc_1}{\algof{S}} \pop{\algof{S}}\\
    & & \sgraph{(\ell\_acki1,ai1)}{\asrc_0}{\asrc_1}{\algof{S}} \pop{\algof{S}} \sgraph{(\ell\_acki1,si1)}{\asrc_0}{\asrc_1}{\algof{S}}\big) \pop{\algof{S}}\\
    & & \bigoplus^{\algof{S}}_{\ell \in L_{dec1}} \big( \sgraph{(\ell\_zero1,z1)}{\asrc_0}{\asrc_1}{\algof{S}} \pop{\algof{S}} \\
    & & \sgraph{(\ell\_reqd1,d1)}{\asrc_0}{\asrc_1}{\algof{S}} \pop{\algof{S}} \sgraph{(\ell\_reqd1,pd1)}{\asrc_0}{\asrc_1}{\algof{S}} \pop{\algof{S}}\\
    & & \sgraph{(\ell\_ackd1,ad1)}{\asrc_0}{\asrc_1}{\algof{S}} \pop{\algof{S}} \sgraph{(\ell\_ackd1,sd1)}{\asrc_0}{\asrc_1}{\algof{S}}\big) \pop{\algof{S}}\\
    & & \bigoplus^{\algof{S}}_{\ell \in L_{inc2}}\big( \sgraph{(i2,\ell\_reqi2)}{\asrc_2}{\asrc_0}{\algof{S}} \pop{\algof{S}} \sgraph{(pi2,\ell\_reqi2)}{\asrc_2}{\asrc_0}{\algof{S}} \pop{\algof{S}}\\
    & & \sgraph{(ai2,\ell\_acki2)}{\asrc_2}{\asrc_0}{\algof{S}} \pop{\algof{S}} \sgraph{(si2,\ell\_acki2)}{\asrc_2}{\asrc_0}{\algof{S}}\big) \pop{\algof{S}}\\
    & & \bigoplus^{\algof{S}}_{\ell \in L_{dec2}}\big( \sgraph{(z2,\ell\_zero2)}{\asrc_2}{\asrc_0}{\algof{S}} \pop{\algof{S}} \\
    & & \sgraph{(d2,\ell\_reqd2)}{\asrc_2}{\asrc_0}{\algof{S}} \pop{\algof{S}} \sgraph{(pd2,\ell\_reqd2)}{\asrc_2}{\asrc_0}{\algof{S}} \pop{\algof{S}}\\
    & & \sgraph{(ad2,\ell\_ackd2)}{\asrc_2}{\asrc_0}{\algof{S}} \pop{\algof{S}} \sgraph{(sd2,\ell\_ackd2)}{\asrc_2}{\asrc_0}{\algof{S}}\big)\\
    
    Y_1 & \rightarrow & \sgraph{(qi1,i1)}{\asrc_1}{\asrc_3}{\algof{S}} \pop{\algof{S}} \sgraph{(qi1,pi1)}{\asrc_1}{\asrc_3}{\algof{S}} \pop{\algof{S}}\\
    & & \sgraph{(ri1,ai1)}{\asrc_1}{\asrc_3}{\algof{S}} \pop{\algof{S}} \sgraph{(ri1,si1)}{\asrc_1}{\asrc_3}{\algof{S}} \pop{\algof{S}}\\
    & & \sgraph{(qd1,d1)}{\asrc_1}{\asrc_3}{\algof{S}} \pop{\algof{S}} \sgraph{(qd1,pd1)}{\asrc_1}{\asrc_3}{\algof{S}} \pop{\algof{S}}\\
    & & \sgraph{(rd1,ad1)}{\asrc_1}{\asrc_3}{\algof{S}} \pop{\algof{S}} \sgraph{(rd1,sd1)}{\asrc_1}{\asrc_3}{\algof{S}}\pop{\algof{S}}\\
 & & \sgraph{(nextz1,z1)}{\asrc_1}{\asrc_3}{\algof{S}}\\
    
    Y_2 & \rightarrow & \sgraph{(i2,qi2)}{\asrc_4}{\asrc_2}{\algof{S}} \pop{\algof{S}} \sgraph{(pi2,qi2)}{\asrc_4}{\asrc_2}{\algof{S}} \pop{\algof{S}}\\
    & & \sgraph{(ai2,ri2)}{\asrc_4}{\asrc_2}{\algof{S}} \pop{\algof{S}} \sgraph{(si2,ri2)}{\asrc_4}{\asrc_2}{\algof{S}} \pop{\algof{S}}\\
    & & \sgraph{(d2,qd2)}{\asrc_4}{\asrc_2}{\algof{S}} \pop{\algof{S}} \sgraph{(pd2,qd2)}{\asrc_4}{\asrc_2}{\algof{S}} \pop{\algof{S}}\\
    & & \sgraph{(ad2,rd2)}{\asrc_4}{\asrc_2}{\algof{S}} \pop{\algof{S}} \sgraph{(sd2,rd2)}{\asrc_4}{\asrc_2}{\algof{S}}\pop{\algof{S}}\\
 & & \sgraph{(z2,nextz2)}{\asrc_4}{\asrc_2}{\algof{S}}\\
  \end{array}
  $$
\end{center}
\end{minipage}
\vspace{3em}

\begin{minipage}{\textwidth}
  \begin{center}
    with $\ptypeof{\asrc_0}=\Minsk$, $\ptypeof{\asrc_1}=\ptypeof{\asrc_3}=\Count_1$, $\ptypeof{\asrc_2}=\ptypeof{\asrc_4}=\Count_2$,\\$\alpha_1=\{\asrc1 \leftrightarrow \asrc3\}$, \\$\alpha_2=\{\asrc2 \leftrightarrow \asrc4\}$,\\$\slabs=\{\asrc1,\asrc2\}$,\\
    $L_{inci}=\{\ell \mid \ell:ci=ci+1;\mathtt{goto}~\ell'; \in Instrs\}$ for $i \in \{1,2\}$\\ $L_{deci}=\{\ell \mid \ell: \mathtt{if}~ci==0?~\mathtt{goto}~ \ell'; \mathtt{else}~ci=ci-1;\mathtt{goto}~\ell''; \in Instrs\}$ for $i \in \{1,2\}$
  \end{center}
  \end{minipage}
\caption{Grammar $\grammar_\Minsk$ to produce systems for the simulation of a Minsky machine}
\label{fig:grammmar-minsk}
\end{figure}

We now present more formal arguments to explain why the reduction holds.
We first describe the systems belonging to the language of the  grammar $\grammar_{\Minsk}$
(presented on \figref{fig:grammmar-minsk}) in the algebras of systems.
Given two naturals $m,n \geq 1$, we define the system
$\asys^{m,n}_{\mathit{\Minsk}}=(\verts,\edges,\vlab)$ where:
\begin{itemize}
    \item $\verts=\{v^2_m,\ldots,v^2_1,v^c,v^1_1,\ldots,v^1_n\}$;
    \item $\edges=$\\$\{(v^2_i,a,v^2_{i-1}) \mid i \in [2,m] \mbox{ and } a \in \{(i2,qi2),(pi2,qi2),(ai2,ri2),(si2,ri2),\linebreak[0](d2,qd2),\linebreak[0](pd2,qd2),(ad2,rd2),(sd2,rd2),,(z2,nextz2)\}\} \cup $\\
    $\{(v^1_i,a,v^1_{i+1}) \mid i \in [1,n-1] \mbox{ and } a \in \{(qi1,i1),(qi1,pi1),(ri1,ai1),(ri1,si1),\linebreak[0](qd1,d1),\linebreak[0](qd1,pd1),(rd1,ad1),(rd1,sd1),(nextz1,z1)\}\} \cup$\\
    $\{(v^2_1,a,v^c) \mid a \in \{(i2,\ell\_reqi2), (pi2,\ell\_reqi2), (ai2,\ell\_acki2),(si2,\ell\_acki2)\mid \ell \in L_{inc2} \}\} \cup$\\
    $\{(v^2_1,a,v^c) \mid a \in \{(z2,\ell\_zero2),(d2,\ell\_reqd2),(pd2,\ell\_reqd2),(ad2,\ell\_ackd2),(sd2,\ell\_ackd2)\mid \ell \in L_{dec2} \}\} \cup$
    $\{(v^c,a,v^1_1) \mid a \in \{(\ell\_reqi1,i1), (\ell\_reqi1,pi1), (\ell\_acki1,ai1),(\ell\_acki1,si1)\mid \ell \in L_{inc1} \}\} \cup$\\
    $\{(v^c,a,v^1_1) \mid a \in \{(\ell\_zero1,z1),(\ell\_reqd1,d1),(\ell\_reqd1,pd1),(\ell\_ackd1,ad1),(\ell\_ackd1,sd1)\mid \ell \in L_{dec1} \}\}$ where $L_{inc1}$, $L_{dec1}$, $L_{inc2}$ and $L_{dec2}$ are defined on \figref{fig:grammmar-minsk};
    \item $\vlab(v^2_i)=\Count_2$ for all $i \in [1,m]$, $\vlab(v^1_i)=\Count_1$ for all $i \in [1,n]$ and $\vlab(v^c)=\Minsk$.

\end{itemize}

We can easily see that
$\alangof{}{\algof{S}}{\grammar_{\Minsk}}=\set{\asys^{m,n}_{\Minsk} \mid m \geq 1 \mbox{ and } n \geq 1}$.

We have given on Figures \ref{fig:proc-count1} and \ref{fig:proc-minsk},
the way process types $\Count_1$, $\Count_2$ and $\Minsk$ are built.
We notice that in the process type $\Minsk$,
there is a single vertex labelled by the halting label of the Minsky machine $\ell_f$. We let  $\amark_f:\mathcal{Q}
  \rightarrow\nat$ (where $\mathcal{Q}$ is the set of states of the different processes) be the marking such that
$\amark_f(\ell_f)=1$ and $\amark_f(q)=0$ for all other places $q$.
Our reduction claims that the Minsky machine halts iff the answer to
$\paramreach{\grammar_{\Minsk}}{\{\ell_f\}}{\amark_f}$
(\resp $\paramcover{\grammar_{\Minsk}}{\{\ell_f\}}{\amark_f}$) is positive.

Let us explain why the reduction holds for the two verification problems.
First note that if the answer to $\paramreach{\grammar_{\Minsk}}{\{\ell_f\}}{\amark_f}$
is positive then so is the answer to $\paramcover{\grammar_{\Minsk}}{\{\ell_f\}}{\amark_f}$
by definition of the problem.
Now assume the answer to $\paramcover{\grammar_{\Minsk}}{\{\ell_f\}}{\amark_f}$ is positive,
it means that there exists a system $\asys^{m,n}_{\Minsk}=(\verts,\edges,\vlab)$
with $\verts=\{v^2_m,\ldots,v^2_1,v^c,v^1_1,\ldots,v^1_n\}$
where $v^c$ is the only vertex such that $\vlab(v^c)=\Minsk$
and a marking $\overline{\amark}\in\cover{\behof{\asys^{m,n}_{\Minsk}}}$ verifying
$\overline{\amark}(\ell_f,v^c)=\amark(\ell_f)=1$.
By definition, there exists hence a marking $\overline{\amark}' \in \reach{\behof{\asys^{m,n}_\Minsk}}$
such that $\overline{\amark} \leq \overline{\amark}'$.
But since since $\behof{\asys^{m,n}_\Minsk}$ is an automata-like PN,
we have $\overline{\amark}'(\ell_f,v^c)=1$ and since $v^c$ is the only vertex such that $\vlab(v^c)=\Minsk$,
we conclude that the answer to $\paramreach{\grammar_{\Minsk}}{\{\ell_f\}}{\amark_f}$ is positive.

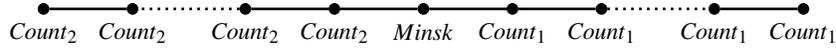
\begin{figure}[htbp]
\begin{center}
  \scalebox{1}{
    \begin{tikzpicture}[node distance=1.5cm]
      \tikzstyle{every state}=[inner sep=3pt,minimum size=20pt]
      \node(v0)[gnode,draw=black,label=-90:{$Minsk$}]{};
      \node(v1)[gnode,draw=black,label=-90:{$\Count_1$},right of=v0,xshift=-1em]{};
      \node(v2)[gnode,draw=black,label=-90:{$\Count_1$},right of=v1,xshift=-1em]{};
      \node(v3)[gnode,draw=black,label=-90:{$\Count_1$},right of=v2,xshift=0em]{};
      \node(v4)[gnode,draw=black,label=-90:{$\Count_1$},right of=v3,xshift=-1em]{};
      \node(v1b)[gnode,draw=black,label=-90:{$\Count_2$},left of=v0,xshift=1em]{};
      \node(v2b)[gnode,draw=black,label=-90:{$\Count_2$},left of=v1b,xshift=1em]{};
      \node(v3b)[gnode,draw=black,label=-90:{$\Count_2$},left of=v2b,xshift=0em]{};
      \node(v4b)[gnode,draw=black,label=-90:{$\Count_2$},left of=v3b,xshift=1em]{};

      \path (v0) edge [-,thick,line width=1pt] (v1);
      \path (v1) edge [-,thick,line width=1pt] (v2);
      \path (v2) edge [-,thick,line width=1pt,dotted] (v3);
      \path (v3) edge [-,thick,line width=1pt] (v4);
      \path (v0) edge [-,thick,line width=1pt] (v1b);
      \path (v1b) edge [-,thick,line width=1pt] (v2b);
      \path (v2b) edge [-,thick,line width=1pt,dotted] (v3b);
      \path (v3b) edge [-,thick,line width=1pt] (v4b);
\end{tikzpicture}
  }
\end{center}
\vspace*{-\baselineskip}
\caption{Shape of a sytem for the simulation of a Minsky machine}
\label{fig:undec4}
\end{figure}
\begin{figure}[htbp]
\begin{center}
\scalebox{0.7}{
\begin{tikzpicture}[node distance=1.5cm]
  \tikzstyle{every state}=[inner sep=3pt,minimum size=20pt]

\node(0C1)[petri-p,draw=black,label=00:{$\mathit{0C1}$}]{};
\node(Z1)[petri-t2,draw=black,fill=black,above of=0C1,label=90:{$\mathit{z1}$}]{};
\node(I1)[petri-t,draw=black,fill=black,below left of=0C1,label=180:{$\mathit{i1}$}]{};
\node(D1)[petri-t,draw=black,fill=black,below right of=0C1,label=0:{$\mathit{ad1}$},yshift=2em]{};
\node(0to1)[petri-p,draw=black,below of=I1,yshift=2em]{};
\node(1to0)[petri-p,draw=black,below of=D1,yshift=2.5em]{};
\node(AI1)[petri-t,draw=black,fill=black,below of=0to1,label=180:{$\mathit{ai1}$},yshift=2em]{};
\node(NZ1)[petri-t,draw=black,fill=black,below of=1to0,label=0:{$\mathit{nextz1}$},yshift=2.5em]{};
\node(1to0b)[petri-p,draw=black,below of=NZ1,yshift=2.5em]{};
\node(AD1)[petri-t,draw=black,fill=black,below of=1to0b,label=0:{$\mathit{d1}$},yshift=2.5em]{};
\node(1C1)[petri-p,draw=black,below right of=AI1,label=90:{$\mathit{1C1}$}]{};

\node(PI1)[petri-t,draw=black,fill=black,below left of=1C1,label=180:{$\mathit{pi1}$}]{};
\node(pi1toqi1)[petri-p,draw=black,below of=PI1,yshift=2em]{};
\node(QI1)[petri-t,draw=black,fill=black,below of=pi1toqi1,label=-110:{$\mathit{qi1}$},yshift=2em]{};

\node(qi1tori1)[petri-p,draw=black,below right of=QI1,xshift=0em]{};

\node(RI1)[petri-t,draw=black,fill=black,above right of=qi1tori1,label=-70:{$\mathit{ri1}$}]{};
\node(ri1tosi1)[petri-p,draw=black,above of=RI1,yshift=-2em]{};
\node(SI1)[petri-t,draw=black,fill=black,above of=ri1tosi1,label=0:{$\mathit{si1}$},yshift=-2em]{};

\path (0C1) edge [->,thick,line width=1pt,out=30,in=0](Z1);
\path (Z1) edge [->,thick,line width=1pt,out=180,in=150](0C1);
\path (0C1) edge [->,thick,line width=1pt,out=-150,in=90](I1);
\path (D1) edge [->,thick,line width=1pt,out=90,in=-30](0C1);
\path (I1) edge [->,thick,line width=1pt](0to1);
\path (1to0) edge [->,thick,line width=1pt](D1);
\path (1to0b) edge [->,thick,line width=1pt](NZ1);
\path (NZ1) edge [->,thick,line width=1pt](1to0);
\path (0to1) edge [->,thick,line width=1pt](AI1);
\path (AD1) edge [->,thick,line width=1pt](1to0b);
\path (AI1) edge [->,thick,line width=1pt,out=-90,in=150](1C1);
\path (1C1) edge [->,thick,line width=1pt,out=30,in=-90](AD1);

\path (1C1) edge [->,thick,line width=1pt](PI1);
\path (PI1) edge [->,thick,line width=1pt](pi1toqi1);
\path (pi1toqi1) edge [->,thick,line width=1pt](QI1);
\path (QI1) edge [->,thick,line width=1pt](qi1tori1);

\path (qi1tori1) edge [->,thick,line width=1pt](RI1);
\path (RI1) edge [->,thick,line width=1pt](ri1tosi1);

\path (ri1tosi1) edge [->,thick,line width=1pt](SI1);
\path (SI1) edge [->,thick,line width=1pt](1C1);

\node[petri-tok] at (0C1) {};

\node(0C1bis)[petri-p,draw=black,label=00:{$\mathit{0C1}$},right of=0C1,xshift=10em]{};
\node(Z1bis)[petri-t2,draw=black,fill=black,above of=0C1bis,label=90:{$\mathit{z1}$}]{};
\node(I1bis)[petri-t,draw=black,fill=black,below left of=0C1bis,label=180:{$\mathit{i1}$}]{};
\node(D1bis)[petri-t,draw=black,fill=black,below right of=0C1bis,label=0:{$\mathit{ad1}$},yshift=2em]{};
\node(0to1bis)[petri-p,draw=black,below of=I1bis,yshift=2em]{};
\node(1to0bis)[petri-p,draw=black,below of=D1bis,yshift=2.5em]{};
\node(AI1bis)[petri-t,draw=black,fill=black,below of=0to1bis,label=180:{$\mathit{ai1}$},yshift=2em]{};
\node(NZ1bis)[petri-t,draw=black,fill=black,below of=1to0bis,label=0:{$\mathit{nextz1}$},yshift=2.5em]{};
\node(1to0bbis)[petri-p,draw=black,below of=NZ1bis,yshift=2.5em]{};
\node(AD1bis)[petri-t,draw=black,fill=black,below of=1to0bbis,label=0:{$\mathit{d1}$},yshift=2.5em]{};
\node(1C1bis)[petri-p,draw=black,below right of=AI1bis,label=90:{$\mathit{1C1}$}]{};

\node(PI1bis)[petri-t,draw=black,fill=black,below left of=1C1bis,label=180:{$\mathit{pi1}$}]{};
\node(pi1toqi1bis)[petri-p,draw=black,below of=PI1bis,yshift=2em]{};
\node(QI1bis)[petri-t,draw=black,fill=black,below of=pi1toqi1bis,label=180:{$\mathit{qi1}$},yshift=2em]{};
\node(qi1tori1bis)[petri-p,draw=black,below right of=QI1bis]{};
\node(RI1bis)[petri-t,draw=black,fill=black,above right of=qi1tori1bis,label=0:{$\mathit{ri1}$}]{};
\node(ri1tosi1bis)[petri-p,draw=black,above of=RI1bis,yshift=-2em]{};
\node(SI1bis)[petri-t,draw=black,fill=black,above of=ri1tosi1bis,label=0:{$\mathit{si1}$},yshift=-2em]{};

\path (0C1bis) edge [->,thick,line width=1pt,out=30,in=0](Z1bis);
\path (Z1bis) edge [->,thick,line width=1pt,out=180,in=150](0C1bis);
\path (0C1bis) edge [->,thick,line width=1pt,out=-150,in=90](I1bis);
\path (D1bis) edge [->,thick,line width=1pt,out=90,in=-30](0C1bis);
\path (I1bis) edge [->,thick,line width=1pt](0to1bis);
\path (1to0bbis) edge [->,thick,line width=1pt](NZ1bis);
\path (NZ1bis) edge [->,thick,line width=1pt](1to0bis);
\path (1to0bis) edge [->,thick,line width=1pt](D1bis);
\path (0to1bis) edge [->,thick,line width=1pt](AI1bis);
\path (AD1bis) edge [->,thick,line width=1pt](1to0bbis);
\path (AI1bis) edge [->,thick,line width=1pt,out=-90,in=150](1C1bis);
\path (1C1bis) edge [->,thick,line width=1pt,out=30,in=-90](AD1bis);

\path (1C1bis) edge [->,thick,line width=1pt](PI1bis);
\path (PI1bis) edge [->,thick,line width=1pt](pi1toqi1bis);
\path (pi1toqi1bis) edge [->,thick,line width=1pt](QI1bis);
\path (QI1bis) edge [->,thick,line width=1pt](qi1tori1bis);
\path (qi1tori1bis) edge [->,thick,line width=1pt](RI1bis);
\path (RI1bis) edge [->,thick,line width=1pt](ri1tosi1bis);
\path (ri1tosi1bis) edge [->,thick,line width=1pt](SI1bis);
\path (SI1bis) edge [->,thick,line width=1pt](1C1bis);

\node[petri-tok] at (0C1bis) {};

\node(L)[petri-p,draw=black,label=180:{$\ell$},left of=1C1,xshift=-25em]{};
\node(ReqI1)[petri-t2,draw=black,fill=black,right of=L,label=90:{$\ell\_\mathit{reqi1}$},xshift=-2em]{};

\node(aux1)[petri-p,draw=black,right of=ReqI1,xshift=-2em]{};
\node(AckI1)[petri-t2,draw=black,fill=black,right of=aux1,label=90:{$\ell\_\mathit{acki1}$},xshift=-2em]{};

\node(L2)[petri-p,draw=black,right of=AckI1,label=90:{$\ell'$},xshift=-1.5em]{};
\node(ReqI1bis)[petri-t2,draw=black,fill=black,right of=L2,label=90:{$\ell'\_\mathit{reqi1}$},xshift=-2em]{};

\node(aux1bis)[petri-p,draw=black,right of=ReqI1bis,xshift=-2em]{};
\node(AckI1bis)[petri-t2,draw=black,fill=black,right of=aux1bis,label=90:{$\ell'\_\mathit{acki1}$},xshift=-2em]{};

\node(L3)[petri-p,draw=black,right of=AckI1bis,label=0:{$\ell''$},xshift=-2em]{};

\path (L) edge [->,thick,line width=1pt](ReqI1);
\path (ReqI1) edge [->,thick,line width=1pt](aux1);
 \path (aux1) edge [->,thick,line width=1pt](AckI1);
  \path (AckI1) edge [->,thick,line width=1pt](L2);
  \path (L2) edge [->,thick,line width=1pt](ReqI1bis);
  \path (ReqI1bis) edge [->,thick,line width=1pt](aux1bis);
  
  \path (aux1bis) edge [->,thick,line width=1pt](AckI1bis);
  \path (AckI1bis) edge [->,thick,line width=1pt](L3);

\node[petri-tok] at (L) {};

\path (ReqI1) edge [loosely dotted,line width=1pt, bend left] (I1);
\path (AckI1) edge [loosely dotted,line width=1pt, bend left] (AI1);
\path (ReqI1) edge [loosely dotted,line width=1pt, bend right] (PI1);
\path (AckI1) edge [loosely dotted,line width=1pt, bend right] (SI1);
\path (ReqI1bis) edge [loosely dotted,line width=1pt, bend left] (I1);
\path (AckI1bis) edge [loosely dotted,line width=1pt, bend left] (AI1);
\path (ReqI1bis) edge [loosely dotted,line width=1pt, bend right] (PI1);
\path (AckI1bis) edge [loosely dotted,line width=1pt, bend right] (SI1);

\path (QI1) edge [loosely dotted,line width=1pt] (I1bis);
\path (RI1) edge [loosely dotted,line width=1pt] (AI1bis);
\path (QI1) edge [loosely dotted,line width=1pt] (PI1bis);
\path (RI1) edge [loosely dotted,line width=1pt] (SI1bis);

\end{tikzpicture}
}
\end{center}
\vspace*{-\baselineskip}
\caption{Partial representation of a behavior for the simulation of a Minsky machine}
\label{fig:behav}
\end{figure}
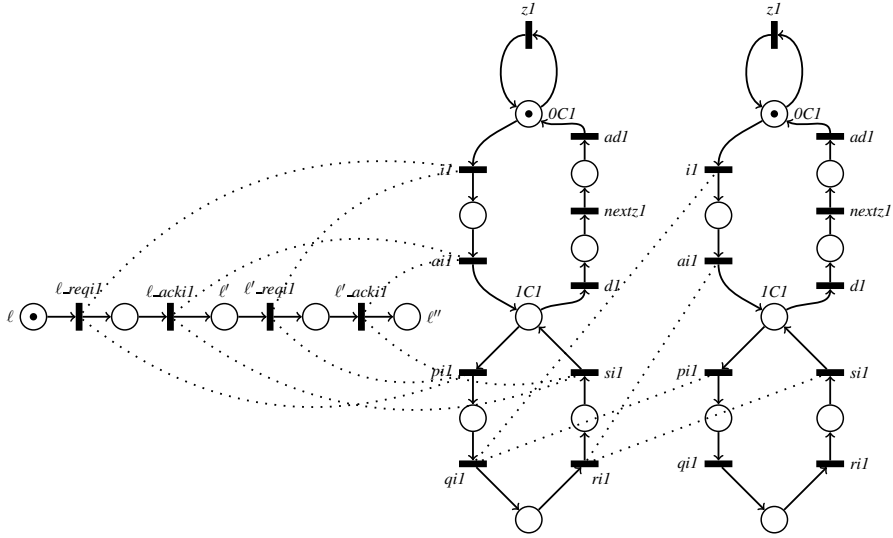

Assume now that the Minsky machine halts.
Let $m$ and $n$ be the maximum counter value taken by the counter $c2$ (\resp the counter $c1$)
during the execution of the machine starting at $\ell_0$ with $c1$ and $c2$ set at $0$ and ending in $\ell_f$.
We can simulate this execution, as we have explained before, in the behavior $\behof{\asys^{m+1,n+1}_\Minsk}$
to reach a marking with a token in the place $(\ell_f,v^c)$.
As a matter of fact, the answer to $\paramcover{\grammar_{\Minsk}}{\{\ell_f\}}{\amark_f}$ is positive.

On the other side if the answer to $\paramcover{\grammar_\Minsk}{\{\ell_f\}}{\amark_f}$ is positive,
there exists a system $\asys^{m,n}_\Minsk \in \alangof{}{\algof{S}}{\grammar_\Minsk}$
and a marking $\overline{\amark} \in \reach{\behof{\asys^{m,n}_\Minsk}}$
such that $\overline{\amark}(\ell_f,v^c)=1$.
The only way to reach such a marking from the initial marking of $\behof{\asys^{m,n}_\Minsk}$
is to simulate faithfully at each step an instruction of the Minsky machine
(if the behavior performs a wrong non-deterministic choice or if the vertices encoding the counters are not enough,
the simulation get stuck and there is no way to put a token in a place of the process type $\Minsk$ corresponding to a label of the machine).
We hence can rebuild an execution of the Minsky machine which ends in the label $\ell_f$. \qed

\end{proofE}
\begin{proofSketch}
  With 3 process types one can write a grammar whose language consists of
  nets that simulate executions of two-counter Minski machines,
  and the halting problem thus reduces to coverability.
  \qed
\end{proofSketch}

\section{Counting Abstraction}

We define the \emph{counting abstraction} of a set of behaviors as a
set of PNs having finitely many underlying nets with possibly
infinitely many initial markings each. The basic idea is to fold (i)
the copies of the same place from some process type into a single
place and (ii) the transitions having the same sets of predecessors
and successors into a single transition. We show that the counting
abstraction of the set of behaviors of a parameterized system
described by an \hrtext{} grammar can be computed by evaluating the
same grammar in a finite \hrtext{} algebra. Moreover, the infinite set
of initial markings of a folded PN can be finitely described by
another PN derived from the initial grammar describing the system.

\ifLongVersion\else
\begin{textAtEnd}[category=quotients]
\fi

We formalize our counting abstraction as quotienting of PNs with
respect to an equivalence relation on places, that has finitely many
classes. Let $\amarkednet = (\anet,\amark_0)$ be a PN with underlying
net $\anet = (\places,\trans,\weight)$ and let $\sim \subseteq \places
\times \places$ be an equivalence relation on the places of
$\anet$. We denote by $[q]_\sim$ the $\sim$-equivalence class of the
place $q \in \places$. The \emph{place-quotient} of $\anet$
w.r.t. $\sim$ is the net $(\places_{/\sim}, \trans, \weight_{/\sim})$,
where:
\begin{align*}
\places_{/\sim} \isdef \set{[q]_\sim \mid q \in \places} \hspace*{8mm}
\weight_{/\sim}([q]_\sim,t) \isdef \sum_{r \in [q]_\sim} \weight(r,t) \hspace*{8mm}
\weight_{\sim}(t,[q]_\sim) \isdef \sum_{r \in [q]_\sim} \weight(t,r) 
\end{align*}
for all $q \in \places$ and $t \in \trans$. The equivalence relation
$\sim$ on places induces the following equivalence relation $\approx
\subseteq \trans \times \trans$ on transitions:
\begin{align}\label{eq:rel-trans}
  t \approx t' \iffdef & \left\{\begin{array}{l}
    \weight_{/\sim}([q]_\sim,t) = \weight_{/\sim}([q]_\sim,t') \\
    \weight_{/\sim}(t,[q]_\sim) = \weight_{/\sim}(t',[q]_\sim) 
    \end{array}\right.
    \text{, for all } q \in \places 
\end{align}
We denote by $[t]_\approx$ the $\approx$-equivalence class of the
transition $t \in \trans$.  The \emph{quotient} of the net $\anet$
w.r.t. $\sim$ is the net $\anet_{/\sim} \isdef
({\places}_{/\sim},{\trans}_{/\approx},{\weight}_{/\sim})$, where
${\weight}_{/\sim}([t]_\approx) \isdef {\weight}_{/\sim}(t)$, for all
$t \in \trans$.  Note that ${\weight}_{/\sim}([t]_\approx)$ is
well-defined because all transitions in the equivalence class have the
same incoming and outgoing weighted edges. The quotient of the PN
$\amarkednet$ is the PN ${\amarkednet}_{~/\sim} \isdef
({\anet}_{/\sim},{\amark_0}_{\sim}) $, where
${\amark}_{/\sim}([q]_\sim) \isdef \sum_{r\in[q]_\sim} \amark(r)$, for
each $\amark:\places\rightarrow\nat$ and $q \in \places$. Quotienting
of markings is lifted to sets as usual ${\markset}_{/\sim} \isdef
\set{{\amark}_{/\sim} \mid \amark \in \markset}$.  The following lemma
relates the reachable (\resp coverable) markings in a PN with that of
its quotient:

\begin{lemma}\label{lemma:quotient-soundness}
  For each PN $\amarkednet$ and equivalence relation
  $\sim \subseteq \placeof{\amarkednet} \times \placeof{\amarkednet}$,
  ${\reach{\amarkednet}}_{/\sim} \subseteq \reach{{\amarkednet~}_{/\sim}}$
  and ${\cover{\amarkednet}}_{/\sim} \subseteq \cover{{\amarkednet~}_{/\sim}}$.
\end{lemma}
\begin{proof}
  Let $\amarkednet=(\anet,\amark_0)$,
  where $\anet=(\places,\trans,\weight)$.
  The proof for ${\reach{\amarkednet}}_{/\sim} \subseteq \reach{{\amarkednet}_{~/\sim}}$
  relies on the two facts below and
  ${\cover{\amarkednet}}_{/\sim} \subseteq \cover{{\amarkednet}_{~/\sim}}$
  is an immediate consequence. 

  \begin{fact}
    For all markings $\amark,\amark' : \places\rightarrow\nat$ and each transition $t \in \trans$,
    $\amark \fire{t} \amark'$ in $(\places,\trans,\weight)$ only if
    ${\amark}_{/\sim} \fire{t} {\amark'}_{/\sim}$
    in $({\places}_{/\sim}, \trans, {\weight}_{/\sim})$.
  \end{fact}
  \begin{proof}
    If $t$ is enabled in $\amark$ it means that
    $\amark(q) \geq \weight(q,t)$ for every $q \in \places$.
    Applying this inequality pointwise on
    ${\amark}_{/\sim}([q]_\sim)
      = \sum_{q' \in [q]_\sim} \amark(q')
      \geq \sum_{q' \in [q]_\sim} \weight(q',t)
      = {\weight}_{/\sim}([q]_\sim,t)$
    gives that $t$ is enabled in ${\amark}_{/\sim}$.
    By a similar manipulation,
    ${\amark}_{/\sim}([q]_\sim)
      - ({\weight}_{/\sim})([q]_\sim,t)
      + ({\weight}_{/\sim})(t,[q]_\sim)
    = \sum_{q' \in [q]_\sim} \amark(q')
      - \sum_{q' \in [q]_\sim} \weight(q',t)
      + \sum_{q' \in [q]_\sim} \weight(t,q')$
    can be rearranged into
    $\sum_{q' \in [q]_\sim} (\amark(q') - \weight(q',t) + \weight(t,q'))
      = \sum_{q' \in [q]_\sim} \amark'(q')
      = {\amark'}_{/\sim}([q]_\sim)$,
    which shows that the result of firing $t$ in ${\amark}_{/\sim}$
    is indeed ${\amark'}_{/\sim}$.
    \qed
  \end{proof}

  \begin{fact}
    For all markings $\amark,\amark' : \places\rightarrow\nat$ and each transition $t \in \trans$,
    ${\amark}_{/\sim} \fire{t} {\amark'}_{/\sim}$
    in $({\places}_{/\sim},\trans,{\weight}_{/\sim})$
    if and only if ${\amark}_{/\sim} \fire{[t]_\approx} {\amark'}_{/\sim}$
    in $({\places}_{/\sim},{\trans}_{/\approx},{\weight}_{/\sim})$.
  \end{fact}
  \begin{proof}
    $\weight(q,t) = {\weight}_{/\approx}(q,[t]_\approx)$
    for every $q \in \places$ by definition,
    therefore $t$ is enabled in ${\amark}_{/\sim}$ if and only if $[t]_\approx$
    is enabled in ${\amark}_{/\sim}$.
    Furthermore firing $t$ or $[t]_\approx$ in ${\amark}_{/\sim}$
    results in the same marking since
    ${\amark}_{/\sim}(q) - \weight(q,t) + \weight(t,q)
      = {\amark}_{/\sim}(q)
      - {\weight}_{/\approx}(q,[t]_\approx)
      + {\weight}_{/\approx}([t]_\approx,q)$.
    \qed
  \end{proof}
\end{proof}
\begin{proofSketch}
  Every transition $\amark \fire{t} \amark'$ in $(\places,\trans,\weight)$
  can be successively traced to
  some transition ${\amark}_{/\sim} \fire{t} {\amark'}_{/\sim}$
  in $({\places}_{/\sim}, \trans, {\weight}_{/\sim})$,
  then some ${\amark}_{/\sim} \fire{[t]_\approx} {\amark'}_{/\sim}$
  in $({\places}_{/\sim},{\trans}_{/\approx},{\weight}_{/\sim})$.
  Through this every firing sequence can be lifted from $(\places,\trans,\weight)$
  to $({\places}_{/\sim},{\trans}_{/\approx},{\weight}_{/\sim})$,
  thus preserving all reachability properties.
  \qed
\end{proofSketch}
\ifLongVersion\else
\end{textAtEnd}
\fi

\subsection{Folding}

We define a folding function on the domain $\universeOf{B}$ of system
behaviors. Let $\open{\asys} = (\asys,\sources)$ be a system and
$(\amarkednet,\overline{\sources}) = \behof{\open{\asys}}$ be its
behavior (\defref{def:behavior}). \ifLongVersion\else The folding is
defined as quotienting a behavior $\amarkednet$ via an equivalence
relation $\equiv$ on the places of $\amarkednet$. Note that quotienting is a
standard operation, meaning that equivalent places and transitions
having the same sets of predecessor and successor equivalence classes
$[q]_{\equiv}$, for $q \in \placeof{\amarkednet}$, are joined
together\footnote{We give the formal definition of the quotienting of
  a PN in \appref{app:quotients}.}, the result being denoted as
$\amarkednet_{~/\equiv}$. \fi

We recall that the places of $\amarkednet$ are pairs $(q,v)$, where $q
\in \placeof{\ptypes}$, $v \in \vertof{\asys}$ and the sources of the
behavior are labeled by the function $\overline{\sources}$. A function
$\eta : \sourcelabels \rightarrow \vertof{\asys}$, having $\dom{\eta}
\supseteq \dom{\sources}$, defines the following equivalence relation
$\foldrel{\eta} \subseteq \placeof{\amarkednet} \times
\placeof{\amarkednet}$:
\begin{align}\label{eq:foldrel}
  (q_1,v_1) \foldrel{\eta} (q_2,v_2) \iff
  q_1 = q_2 \text{ and } \set{v_1,v_2} \cap \img{\eta} \neq \emptyset \Rightarrow v_1 = v_2
\end{align}
\ie two places of $\amarkednet$ corresponding to the same place of a
process type (within two distinct instances thereof) are considered
equivalent, except for the sources with labels $\eta$, that are
equivalent only with themselves. The equivalence class of a place
$(q,v) \in \placeof{\amarkednet}$ is denoted
$[(q,v)]_{\foldrel{\eta}}$.
Because we have assumed that the set
of places corresponding to different process types are disjoint,
$(q_1,v_1) \foldrel{\sources} (q_2,v_2)$ implies that
$\vlabof{\asys}(v_1) = \vlabof{\asys}(v_2)$, \ie the two vertices are
instances of the same process type.

The \emph{folding} of $(\amarkednet,\overline{\sources})$ is defined
as $\foldof{\amarkednet,\overline{\sources}} \isdef
({\amarkednet}_{~/\foldrel{\sources}},{\overline{\sources}}_{/\foldrel{\sources}})$,
where, for each mapping $\eta : \sourcelabels \rightarrow
\vertof{\asys}$ having $\dom{\eta} \supseteq \dom{\sources}$, we
define $\overline{\sources}_{/\foldrel{\eta}}(\asrc,
[q]_{\foldrel{\eta}}) \isdef
[\overline{\sources}(\asrc,q)]_{\foldrel{\eta}}$ if $\asrc \in
\dom{\eta}$, else undefined, for each $\asrc \in \sourcelabels$. We
refer to \figref{fig:folding} for examples.


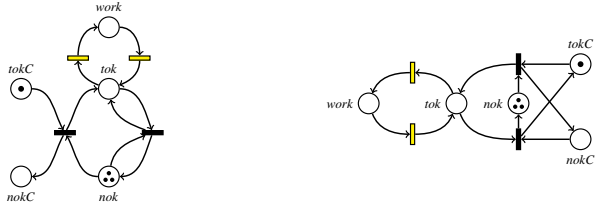
\begin{figure}[t!]
  \vspace*{-\baselineskip}
  \begin{center}
    
    \begin{minipage}{.3\textwidth}
      \vspace*{0\baselineskip}
      \begin{center}
      \scalebox{0.55}{
        \begin{tikzpicture}[node distance=1.5cm]
          \tikzstyle{every state}=[inner sep=3pt,minimum size=20pt]
          
      \node(p0)[petri-p,draw=black,label=90:{$\mathit{tokC}$}]{};
      \node(t1)[petri-t,draw=black,fill=black,below right of=p0,xshift=0em,label=0:{}]{}
      edge [<-,thick,line width=1pt,in=0,out=120](p0);
      \node(p1)[petri-p,draw=black,below left of=t1,xshift=0em,label=-90:{$\mathit{nokC}$}]{};
      
      \node[petri-tok] at (p0) {};
      
      \path (t1) edge [->,thick,line width=1pt,out=-120,in=0](p1);
      
      \node(p0b)[petri-p,draw=black,above right of=t1,xshift=0 em,label=90:{$\mathit{tok}$}]{};
      
      \path (t1) edge [->,thick,line width=1pt,out=60,in=180](p0b);
      
      \node(p1b)[petri-p,draw=black,below right of=t1,xshift=0em,label=-90:{$\mathit{nok}$}]{};
      
      \node[petri-small-tok,yshift=0.8mm] at (p1b) {};
      \node[petri-small-tok,yshift=-0.8mm,xshift=-0.8mm] at (p1b) {};
      \node[petri-small-tok,yshift=-0.8mm,xshift=0.8mm] at (p1b) {};
          
      \path (p1b) edge [->,thick,line width=1pt,in=-60,out=180](t1);
      \node(t1b)[petri-t,draw=black,fill=black,below right of=p0b,xshift=0em,label=0:{}]{}
      edge [->,thick,line width=1pt,out=-110,in=0](p1b)
      edge [<-,thick,line width=1pt,out=-160,in=80](p1b)
      edge [<-,thick,line width=1pt,in=0,out=110](p0b)
      edge [->,thick,line width=1pt,in=-80,out=160](p0b);
      
      \node(t2b)[petri-t,draw=black,fill=yellow,above left of=p0b,label=180:{},xshift=1em,yshift=-1em]{};
      \node(t3b)[petri-t,draw=black,fill=yellow,above right of=p0b,label=0:{},xshift=-1em,yshift=-1em]{};
      \node(p2b)[petri-p,draw=black,above right of=t2b,label=90:{$\mathit{work}$},xshift=-1em,yshift=-1em]{};
      \path (p0b) edge [->,thick,line width=1pt,in=-90,out=160](t2b);
      \path (p0b) edge [<-,thick,line width=1pt,in=-90,out=20](t3b);
      \path (p2b) edge [<-,thick,line width=1pt,in=90,out=180](t2b);
      \path (p2b) edge [->,thick,line width=1pt,in=90,out=0](t3b);
     
      \end{tikzpicture}}
      \end{center}
      
    \end{minipage}
    \begin{minipage}{.5\textwidth}
      \begin{center}
      \scalebox{0.55}{
        \begin{tikzpicture}[node distance=1.5cm]
          \tikzstyle{every state}=[inner sep=3pt,minimum size=20pt]
      
      \node(p0)[petri-p,draw=black,label=90:{$\mathit{tokC}$}]{};
      \node(p1)[petri-p,draw=black,below  of=p0,yshift=-1em,label=-90:{$\mathit{nokC}$}]{};
      
      \node[petri-tok] at (p0) {};
      
      \node(p1b)[petri-p,draw=black,left of=p0,yshift=-3em,label=180:{$\mathit{nok}$}]{};

      \node[petri-small-tok,yshift=0.8mm] at (p1b) {};
      \node[petri-small-tok,yshift=-0.8mm,xshift=-0.8mm] at (p1b) {};
      \node[petri-small-tok,yshift=-0.8mm,xshift=0.8mm] at (p1b) {};
      
      \node(t0b)[petri-t2,draw=black,fill=black,left of=p0,label=180:{}]{};
      \node(t1b)[petri-t2,draw=black,fill=black,left of=p1,label=180:{}]{};
      \node(p0b)[petri-p,draw=black,left of=p1b,label=180:{$\mathit{tok}$}]{};
      
      \node(t2b)[petri-t2,draw=black,fill=yellow,above left of=p0b,yshift=-1em,label=180:{}]{};
      \node(t3b)[petri-t2,draw=black,fill=yellow,below left of=p0b,label=180:{},yshift=1em]{};

      \node(p2b)[petri-p,draw=black,below left of=t2b,label=180:{$\mathit{work}$},yshift=1em]{};

      \path (p1b) edge [->,thick,line width=1pt] (t0b);
      \path (t1b) edge [->,thick,line width=1pt] (p1b);
      \path (p0) edge [->,thick,line width=1pt] (t0b);
      \path (t0b) edge [->,thick,line width=1pt] (p1);
      \path (p1) edge [->,thick,line width=1pt] (t1b);
      \path (t1b) edge [->,thick,line width=1pt] (p0);
      \path (t0b) edge [->,thick,line width=1pt,out=180,in=70] (p0b);
      \path (t1b) edge [<-,thick,line width=1pt,out=180,in=-70] (p0b);
      \path (t2b) edge [<-,thick,line width=1pt,out=0,in=110] (p0b);
      \path (t3b) edge [->,thick,line width=1pt,out=0,in=-110] (p0b);
      \path (t2b) edge [->,thick,line width=1pt,out=180,in=70] (p2b);
      \path (t3b) edge [<-,thick,line width=1pt,out=180,in=-70] (p2b);

      \end{tikzpicture}}
      \end{center}

\end{minipage}
\end{center}
\vspace*{-\baselineskip}
\caption{Foldings of the behaviors given in \figref{fig:proc-typ1} (c) and (d), respectively.}
\label{fig:folding}
\vspace*{-\baselineskip}
\end{figure}



\ifLongVersion\else
\begin{textAtEnd}[category=proofs]
\fi
\begin{lemma}\label{lemma:direct-char-beh}
  Let $(\amarkednet,\overline{\sources})$ and $(\amarkednet',\overline{\sources'})$ be open behaviors. Then, we have: 
  \begin{align}
    \label{eq:pop}
    (\amarkednet,\overline{\sources}) \pop{\algof{B}} (\amarkednet',\overline{\sources'}) = & ~{(\amarkednet \uplus \amarkednet', \overline{\sources} \cup \overline{\sources'})}_{/\srcrel{\sources}{\sources'}} \\
    \label{eq:restrict}
    \restrict{\slabs}{\algof{B}}(\amarkednet,\overline{\sources}) = & ~(\amarkednet,\proj{\overline{\sources}}{\slabs\times\placeof{\amarkednet}}) \text{, for all } \slabs \finsubseteq \sourcelabels \\
    \label{eq:rename}
    \rename{\alpha}{\algof{B}}(\amarkednet,\overline{\sources}) = & ~(\amarkednet,\overline{\sources}\circ(\alpha^{-1} \times \fnid)) \text{, for all finite permutations } \alpha \text{ of } \sourcelabels
  \end{align}
\end{lemma}
\begin{proof}
  (\ref{eq:pop}) We first explain why despite some abuse of notation,
  the right-hand side of the equation is well-defined. Recall that
  $\sources : \sourcelabels \rightarrow \placeof{\amarkednet}$ and
  $\sources' : \sourcelabels \rightarrow \placeof{\amarkednet'}$ are
  injective, and that their lifting to behaviors as
  $\overline{\sources}$ and $\overline{\sources'}$ have disjoint
  images, because the systems whose behaviors
  $(\amarkednet,\overline{\sources})$ and
  $(\amarkednet',\overline{\sources'})$ are, must have disjoint sets
  of vertices, by the definition of $\pop{\algof{S}}$. Then,
  $\overline{\sources}^{-1}$ and $\overline{\sources}^{-1}$ are
  well-defined partial functions, and they have disjoint domains, thus
  $\overline{\sources}^{-1}\cup\overline{\sources'}^{-1}$ is an
  unambiguously defined partial function. It is however not injective,
  and its kernel matches places labeled by the same source in both
  $(\amarkednet,\overline{\sources})$ and
  $(\amarkednet',\overline{\sources'})$.

    Strictly speaking $\overline{\sources}\cup\overline{\sources'}$ is
    not well-defined because if there are shared sources then
    $\sources$ and $\sources'$ have non-disjoint domains.  However by
    construction the functions
    ${\overline{\sources}}_{/\srcrel{\sources}{\sources'}}$ and
    ${\overline{\sources'}}_{/\srcrel{\sources}{\sources'}}$
    coincide on their shared domain.  Thus the function
    ${\overline{\sources}}_{/\srcrel{\sources}{\sources'}} \cup
    {\overline{\sources'}}_{/\srcrel{\sources}{\sources'}}$ is
    well-defined.  Rather than writing: 
    \[({(\amarkednet\uplus\amarkednet')}_{/\srcrel{\sources}{\sources'}},
    {\overline{\sources}}_{/\srcrel{\sources}{\sources'}} \cup
    {\overline{\sources'}}_{/\srcrel{\sources}{\sources'}})\] we stick
    to ${(\amarkednet \uplus \amarkednet', \overline{\sources} \cup
      \overline{\sources'})}_{/\srcrel{\sources}{\sources'}}$ which,
    despite having some intermediate objects not completely
    well-defined, remains intuitively unambiguous.

    The definition of $(\asys,\sources) \pop{\algof{S}}
    (\asys',\sources')$ takes the disjoint union $\asys \uplus \asys'$
    and fuses pairs of $\asrc$-sources and identical edges.  The
    equivalence relation that determines when to fuse $\asrc$-sources
    happens to be exactly the kernel of
    $\sources^{-1}\cup\sources'^{-1}$.  Recall that the
    $\asrc$-sources of a system become sets of $(\asrc,q)$-sources in
    the behavior and that each edge is translated into a transition.
    Thus the kernel of $\sources^{-1}\cup\sources'^{-1}$
    lifted to behaviors is exactly the kernel of
    $\overline{\sources}^{-1}\cup\overline{\sources}^{-1}$,
    from which we deduce that the quotient by $\srcrel{\sources}{\sources'}$ is
    compatible with the translation from systems to behaviors, which
    gives the desired property.  

    \vspace*{\baselineskip}
    \noindent
    (\ref{eq:restrict}) This fact follows from
    $\overline{\proj{\sources}{\slabs}} =
    \proj{\overline{\sources}}{\slabs\times\places}$, for any finite subset of
    source labels $\slabs\finsubseteq\sourcelabels$.

    \vspace*{\baselineskip}
    \noindent
    (\ref{eq:rename}) This fact follows from $\overline{\sources \circ
      \alpha^{-1}} = \overline{\sources}\circ(\alpha^{-1}\times\fnid)$,
    for any permutation $\alpha:\sourcelabels\rightarrow\sourcelabels$.  \qed
\end{proof}
\begin{proofSketch}
  It is not obvious that these objects are all well-defined,
  since $\overline{\sources}$ and $\overline{\sources'}$ do not coincide on their
  shared domain. They do once we quotient by $\srcrel{\sources}{\sources'}$.
  The rest comes from unfolding the definition of the $\overline{\bullet}$ operator.
  \qed
\end{proofSketch}
\ifLongVersion\else
\end{textAtEnd}
\fi

The following lemma allows to define an algebra of abstract behaviors
$\absof{\algof{B}}$ (\figref{fig:alg-zoo}) from the algebra
$\algof{B}$ of behaviors, using \autoref{prop:cong-homo}. The domain
of $\absof{\algof{B}}$ is the set $\absof{\universeOf{B}}$ of folded
behaviors, \ie quotients of behaviors
$(\amarkednet,\overline{\sources})$ w.r.t. the $\foldrel{\sources}$
relation.

\begin{lemmaE}[][category=proofs]\label{lemma:cong-beh-fold}
  $\kerof{\afold}$ is an \hrtext{} congruence. 
\end{lemmaE}
\begin{proofE}
  By a small abuse of notation we shall denote by
  $\foldrel{\sources}$ both the equivalence relation
  on places as well as the equivalence relation
  on transitions that it induces as defined in
  \autoref{eq:rel-trans}.

  The main proof relies on the following preliminary observations
  about manipulating quotients and markings, which hold for any
  permutation $\alpha : \sourcelabels\rightarrow\sourcelabels$ and any
  set $\slabs\finsubseteq \sourcelabels$:
  \begin{enumerate}[(i)]
    \item\label{it1:cong-beh-fold} $\sources$ and $\sources\circ\alpha^{-1}$
      have the same range, which leads to $\foldrel{\sources}$
      and $\foldrel{\sources\circ\alpha^{-1}}$ having the same
      equivalence classes and defining the same quotients;
    \item\label{it2:cong-beh-fold} the bijectivity of $\alpha$ and the fact that
      $\foldrel{\sources}$
      does not identify vertices with different sources mean that
      $\overline{\sources\circ\alpha^{-1}} = \overline{\sources}\circ\alpha^{-1}$,
      and ${(\overline{\sources} \circ\alpha^{-1})}_{/\foldrel{\sources}}
      = {\overline{\sources}}_{/\foldrel{\sources}} \circ\alpha^{-1}$;
    \item\label{it3:cong-beh-fold} because $\proj{\sources}{\tau}$
      is a restriction of $\sources$,
      the equivalence relation induced by the former is a superset of the latter,
      and equivalence classes of $\foldrel{\proj{\sources}{\tau}}$ are unions
      of equivalence classes of $\foldrel{\sources}$.
    \item\label{it4:cong-beh-fold}
      ${\overline{\sources_1}}_{/\foldrel{\sources_1}} =
      {\overline{\sources_2}}_{/\foldrel{\sources_2}}$ implies
      $\sources_1 = \sources_2$.  Pick some
      ${\overline{\sources_1}}_{/\foldrel{\sources_1}}(\sigma,[q]_\sim)
      =
      {\overline{\sources_2}}_{/\foldrel{\sources_2}}(\sigma,[q]_\sim)$.
      Unfolding the definitions gives that
      $[\overline{\sources_1}(\asrc,q)]_{\foldrel{\sources_1}} =
      [\overline{\sources_2}(\asrc,q)]_{\foldrel{\sources_2}}$, which
      are equal to $[(q,\sources_1(\asrc))]_{\foldrel{\sources_1}}$
      and $[(q,\sources_2(\asrc))]_{\foldrel{\sources_2}}$
      respectively.  The equivalence classes respect the labeling of
      sources, so both of the above places are alone in their
      equivalence class.  This proves $(q,\sources_1(\asrc)) =
      (q,\sources_2(\asrc))$ from which we easily get
      $\sources_1(\asrc) = \sources_2(\asrc)$. Since the choice of
      $\asrc\in\sourcelabels$ was arbitrary, we obtain $\sources_1 =
      \sources_2$.
  \end{enumerate}

  We now proceed separately for each \hrtext{} function symbol: 
  \begin{itemize}
    \item $\sgraph{a}{\asrc_1}{\asrc_2}{}$
      is of arity zero thus trivially 
      $\sgraph{a}{\asrc_1}{\asrc_2}{\algof{B}} \kerof{\afold} \sgraph{a}{\asrc_1}{\asrc_2}{\algof{B}}$.
    \item $\rename{\alpha}{}$: Let
      $(\amarkednet_{~1},\overline{\sources_1}) \kerof{\afold}
      (\amarkednet_{~2},\overline{\sources_2})$ be two open behaviors
      having the same folded image. Using observation
      (\ref{it1:cong-beh-fold}) above and given that
      $\rename{\alpha}{\algof{B}}$ does not modify the underlying net
      and the initial marking, we deduce that
      $\foldof{\rename{\alpha}{\algof{B}}(\amarkednet_{~1},\overline{\sources_1})}$
      and
      $\foldof{\rename{\alpha}{\algof{B}}(\amarkednet_{~2},\overline{\sources_2})}$
      have the same underlying net and initial marking. Then
      observation (\ref{it2:cong-beh-fold}) yields
      ${(\overline{\sources_1\circ\alpha^{-1}})}_{/\foldrel{\sources_1\circ\alpha^{-1}}}
      =
      \left({\overline{\sources_1}}_{/\foldrel{\sources_1}}\right) \circ\alpha^{-1}
      =
      \left({\overline{\sources_2}}_{/\foldrel{\sources_2}}\right) \circ\alpha^{-1}
      =
      {(\overline{\sources_2\circ\alpha^{-1}})}_{/\foldrel{\sources_2\circ\alpha^{-1}}}$
      and we conclude that
      $\rename{\alpha}{\algof{B}}(\amarkednet_{~1},\overline{\sources_1})
      \kerof{\afold}
      \rename{\alpha}{\algof{B}}(\amarkednet_{~2},\overline{\sources_2})$
    \item $\restrict{\slabs}{}$: Let
      $(\amarkednet_{~1},\overline{\sources_1}) \kerof{\afold}
      (\amarkednet_{~2},\overline{\sources_2})$ be two open behaviors
      having the same folded image. By observation
      (\ref{it3:cong-beh-fold}), we have that $\foldrel{\sources_i}
      \subseteq \foldrel{\proj{\sources_i}{\tau}}$ for $i = 1,2$.
      This means that
      ${\amarkednet_{~i}}_{/\foldrel{\proj{\sources_i}{\tau}}}
      =
      {\left({\amarkednet_{~i}}_{/\foldrel{\sources_i}}\right)}_{/\foldrel{\proj{\sources_i}{\tau}}}$
      and because we know that (a)
      ${\amarkednet_{~1}}_{/\foldrel{\sources_i}} =
      {\amarkednet_{~2}}_{/\foldrel{\sources_2}}$ as well as
      (b) $\foldrel{\proj{\sources_1}{\tau}} =
      \foldrel{\proj{\sources_2}{\tau}}$, we deduce that
      ${\amarkednet_{~1}}_{/\foldrel{\proj{\sources_1}{\tau}}}
      =
      {\amarkednet_{~2}}_{/\foldrel{\proj{\sources_2}{\tau}}}$.
      Since, moreover,
      ${\overline{\sources_1}}_{/\foldrel{\sources_1}} =
      {\overline{\sources_2}}_{/\foldrel{\sources_2}}$, we
      conclude that
      ${\overline{\sources_1}}_{/\foldrel{\proj{\sources_1}{\tau}}}
      =
      {\overline{\sources_2}}_{/\foldrel{\proj{\sources_2}{\tau}}}$
      and thus
      $\foldof{\restrict{\slabs}{\algof{B}}(\amarkednet_{~1},\overline{\sources_1})}
      =
      \foldof{\restrict{\slabs}{\algof{B}}(\amarkednet_{~2},\overline{\sources_2})}$. 
    \item $\aop = \pop{\algof{B}}$: Let
      $(\amarkednet_{~1},\overline{\sources_1}) \kerof{\afold}
      (\amarkednet_{~2},\overline{\sources_2})$ and
      $(\amarkednet'_{~1},\overline{\sources_1}) \kerof{\afold}
      (\amarkednet'_{~2},\overline{\sources_2})$ be pairwise
      equivalent behaviors.  \lemref{lemma:direct-char-beh} gives us
      $\afold((\amarkednet_{~i},\sources_i) \pop{\algof{B}}
      (\amarkednet'_{~i},\sources_i)) =
      ({\amarkednet''_{~i}}_{/\foldrel{\sources''_i}},
      {\overline{\sources''_i}}_{/\foldrel{\sources''_i}})$,
      where $\amarkednet''_{~i} = {(\amarkednet_{~i} \uplus
        \amarkednet'_{~i})}_{/\srcrel{\sources_i}{\sources'_i}}$ and
      $\sources''_i =
      {(\overline{\sources_i}\cup\overline{\sources'_i})}_{/\srcrel{\sources_i}{\sources'_i}}$.
      Since by the initial assumption we have
      ${\overline{\sources_1}}_{/\foldrel{\sources_1}} =
      {\overline{\sources_2}}_{/\foldrel{\sources_2}}$ we
      deduce by observation (\ref{it4:cong-beh-fold}) that $\sources_1
      = \sources_2$.  The same reasoning for $\sources'_1$ and
      $\sources'_2$ gives that the equivalence relations
      $\srcrel{\sources_1}{\sources'_1}$ and
      $\srcrel{\sources_2}{\sources'_2}$ are the same, hence
      $\sources''_1 = \sources''_2$.

      It remains to prove that
      ${\amarkednet''_{~1}}_{/\foldrel{\sources''_1}} =
      {\amarkednet''_{~2}}_{/\foldrel{\sources''_2}}$.  It helps that
      $\srcrel{\sources_i}{\sources'_i}$ fuses only places that are
      sources, while $\foldrel{\sources''_i}$ fuses only places that
      are not sources.  This means that whenever a place is part of
      some non-singleton equivalence class according to
      $\foldrel{\sources''_i}$, its equivalence class according to
      $\srcrel{\sources_i}{\sources''_i}$ is a singleton, and
      reciprocally.  Even though $\foldrel{\sources''_i}$ is
      technically defined on equivalence classes of places, this
      allows us to consider it as an equivalence relation between
      places, and to permute the two quotients to obtain
      ${\amarkednet''_{~i}}_{/\foldrel{\sources''_i}} =
      {(({\amarkednet_i \uplus
          \amarkednet'_i)}_{/\foldrel{\sources''_i}})}_{/\srcrel{\sources_i}{\sources'_i}}$
      It is now sufficient to prove
      ${(\amarkednet_{~1}\uplus\amarkednet'_{~1})}_{/\foldrel{\sources''_1}}
      =
      {(\amarkednet_{~2}\uplus\amarkednet'_{~2})}_{/\foldrel{\sources''_2}}$.
      Seeing that $\foldrel{\sources_i},\foldrel{\sources'_i}
      \subseteq \foldrel{\sources''_i}$ we can insert quotients by
      $\foldrel{\sources_i}$ and $\foldrel{\sources'_i}$ underneath a
      quotient by $\foldrel{\sources''_i}$ without affecting the final
      result:
      ${(\amarkednet_{~i}\uplus\amarkednet'_{~i})}_{/\foldrel{\sources''_i}}
      = {({\amarkednet_{~i}}_{/\foldrel{\sources_i}} \uplus
        {\amarkednet'_{~i}}_{/\foldrel{\sources'_i}})}_{/\foldrel{\sources''_i}}$.
      With this formulation it is finally clear that we can substitute
      the known equalities ${\amarkednet_{~1}}_{/\foldrel{\sources_1}}
      = {\amarkednet_{~2}}_{/\foldrel{\sources_2}}$ and
      ${\amarkednet'_{~1}}_{/\foldrel{\sources'_1}} =
      {\amarkednet'_{~2}}_{/\foldrel{\sources'_2}}$ to get the desired
      property.
  \end{itemize}
  Thus $\afold$ is an \hrtext{} congruence.
  \qed
\end{proofE}
\begin{proofSketch}
  For each function symbol $\aop \in \mathsf{HR}$, we express the effects
  of $\aop$ on the source labeling, essentially proving that if
  $\sources_1$ defines the same folding as $\sources_2$,
  then $\aop(\sources_1)$ defines the same folding as $\aop(\sources_2)$.
\end{proofSketch}

As an immediate consequence of \lemref{lemma:cong-beh-fold}, we obtain
that $\afold$ is a homomorphism between $\algof{B}$ and
$\absof{\algof{B}}$, hence the same \hrtext{} grammar can be used to
both specify an infinite set of behaviors and compute its folded
abstraction:

\begin{corollary}\label{prop:beh-fold}
  For each \hrtext{} grammar $\grammar$, we have
  $\foldof{\alangof{}{\algof{B}}{\grammar}}=\alangof{}{\absof{\algof{B}}}{\grammar}$.
\end{corollary}

Because $\ptypes$ is a finite set of process types, each having
finitely many places, the folded language
$\alangof{}{\absof{\algof{B}}}{\grammar}$ has a finite set of
underlying nets. This is because each transition $t$ in a behavior
$\amarkednet \in \alangof{}{\algof{B}}{\grammar}$ has at most two
incoming/outgoing edges. Since the same holds for each transition of
its quotient, \ie $\amarkednet_{~/\foldrel{\sources}}$, there are
finitely many places and transitions in each underlying net of some
$\amarkednet \in
\alangof{}{\absof{\algof{B}}}{\grammar}$. Nevertheless, the language
$\alangof{}{\absof{\algof{B}}}{\grammar}$ is unbounded, because the
set of initial markings is unbounded. In order to represent this set
in a finite way, we shall proceed in two
steps: \begin{compactenum}[1.]
\item\label{it1:initial-markings} Isolate the initial markings that
  correspond to a given net from
  $\alangof{}{\absof{\algof{B}}}{\grammar}$; we address this problem
  using the Filtering Theorem (\thmref{thm:filtering}).
\item\label{it2:initial-markings} Give a finite representation to the
  set of initial markings of each net; we tackle this problem using
  Esparza's idea \cite{Esparza95} of building PNs that simulate the derivations of the
  ``filtered'' grammars obtained in the previous step.
\end{compactenum}

To formally define the finite set of underlying nets from a language
$\alangof{}{\absof{\algof{B}}}{\grammar}$, we consider the function
$\psi$ on the domain $\absof{\universeOf{B}}$, that drops the initial
marking:
\begin{align}\label{eq:drop}
\psi((\anet,\amark_0),\sources) \isdef (\anet,\sources)
\end{align}
Then, $\psi(\alangof{}{\absof{\algof{B}}}{\grammar})$ is finite,
because the numbers of places and transitions in each net from this
set are bounded by $\cardof{\placeof{\ptypes}}$ and
$\cardof{\placeof{\ptypes}}^4$ (\ie each transition has at most $2$
incoming and $2$ outgoing edges of weight $1$), respectively.

Since none of the operations from $\absof{\algof{B}}$ modify the
initial markings, it is straightforward that $\kerof{\psi}$ is a
\hrtext{} congruence. We define the algebra $\finabsof{\algof{B}}$
(\figref{fig:alg-zoo}) having the finite domain
$\finabsof{\universeOf{B}} \isdef \psi(\absof{\universeOf{B}})$ and
the usual interpretations of the \hrtext{} function symbols given by
\propref{prop:cong-homo}. By standard arguments (similar to
\autoref{prop:sys-beh} and \autoref{prop:beh-fold}), we obtain that
$\psi$ is a homomorphism between $\absof{\algof{B}}$ and
$\finabsof{\algof{B}}$, \ie
$\psi(\alangof{}{\absof{\algof{B}}}{\grammar}) =
\alangof{}{\finabsof{\algof{B}}}{\grammar}$.

As previously discussed, $\alangof{}{\finabsof{\algof{B}}}{\grammar}$
is a finite set. However, to ensure that this set can be effectively
computed, the computation of the interpretation of the
\hrtext{} signature in $\finabsof{\algof{B}}$ needs to be effective:

\begin{propositionE}[][category=proofs]\label{prop:effective-finite-abstraction}
  For each function symbol $\aop$ from the \hrtext{} signature, the
  function $\aop^{\finabsof{\algof{B}}}$ is effectively computable.
\end{propositionE}
\begin{proofE}
  Both the folding $\phi : \behaviors \to \absof{\behaviors}$
  and the projection $\psi : \absof{\behaviors} \to \finabsof{\behaviors}$
  preserve the arity of transitions, so the elements of
  $\finabsof{\behaviors}$ obey $1 \leq \cardof{\pre{t}} = \cardof{\post{t}} \leq 2$
  for every transition $t$, and can thus be seen as behaviors with empty markings:
  we define $z : \finabsof{\behaviors} \to \behaviors$ by
  $z((\anet,\overline{\sources})) \isdef ((\anet,0), \overline{\sources})$.
  Because empty markings added together remain empty,
  $z(\finabsof{\behaviors})$ forms a stable subset of $\behaviors$
  by all operations of \hrtext{}.
  This homomorphism $z$ is a right inverse of the translation from
  $\behaviors$ to $\finabsof{\behaviors}$: $(\psi\circ\afold) \circ z = \fnid$.

  By definition $\aop^{\finabsof{\algof{B}}}(b_1, ..., b_n) =
  (\psi\circ\afold)(\aop^{\algof{B}}((\psi\circ\afold)^{-1}(b_1),...,(\psi\circ\afold)^{-1}(b_n))$.
  Since $\kerof{\psi \circ \afold}$ is an \hrtext{} congruence, any right
  inverse of $\psi\circ\afold$ can replace $(\psi\circ\afold)^{-1}$ in
  the expression above, and a good candidate is $z$ which is an
  effectively computable right inverse.  This gives
  $\aop^{\finabsof{\algof{B}}}(b_1, ..., b_n) =
  \psi(\afold(\aop^{\algof{B}}(z(b_1),...,z(b_n))))$.  All of $z$,
  $\afold$, and $\psi$ are effectively computable, so in order for
  $\aop^{\finabsof{\algof{B}}}$ to be effectively computable it
  suffices that $\aop^{\algof{B}}$ also be effectively computable,
  which is given by the alternative characterizations of \hrtext{} operating
  directly on behaviors in \lemref{lemma:direct-char-beh}.  \qed
\end{proofE}
\begin{proofSketch}
  since $\kerof{\phi}$ is \hrtext{} congruence (and so is $\kerof{\psi}$, obviously),
  it suffices to exhibit a right inverse of $\psi \circ \phi$.
  We find that the candidate $z : \finabsof{\behaviors} \to \behaviors$
  defined by $z((\anet,\overline{\sources})) \isdef ((\anet,0), \overline{\sources})$
  works and provides $\aop^{\finabsof{\algof{B}}} = \psi \circ \phi \circ \aop^{\absof{\algof{B}}} \circ z$
  effectively computable.
  \qed
\end{proofSketch}

By the previous arguments, the language
$\alangof{}{\finabsof{\algof{B}}}{\grammar}$ is finite and effectively
computable, hence it can be produced by a finite Kleene iteration of
the monotonic function that maps a tuple of sets indexed by the
nonterminals of $\grammar$ into their $\finabsof{\algof{B}}$
interpretations, given by the rules of $\grammar$.  Let
$\set{(\anet_1, \overline{\sources_1}), \ldots, (\anet_n,
  \overline{\sources_n})} \isdef
\alangof{}{\finabsof{\algof{B}}}{\grammar}$ be this set. Using the
Filtering Theorem (\thmref{thm:filtering}) one can effectively build
grammars $\grammar_1, \ldots, \grammar_n$ such that:
\begin{align}\label{eq:filtering}
  \psi(\alangof{}{\absof{\algof{B}}}{\grammar_i}) =
  \alangof{}{\finabsof{\algof{B}}}{\grammar_i} =
  \set{(\anet_i,\overline{\sources_i})} \text{, for each } i \in \interv{1}{n}
\end{align}
More precisely, the Filtering Theorem gives, for each $i \in
\interv{1}{n}$, a grammar $\grammar_i$ such that
$\alangof{}{\absof{\algof{B}}}{\grammar_i} =
\alangof{}{\absof{\algof{B}}}{\grammar} \cap
\psi^{-1}(\set{(\anet_i,\overline{\sources_i})})$. By applying $\psi$
to both sides of the equality, we obtain (\autoref{eq:filtering}),
thus taking care of the first step of the construction
(\ref{it1:initial-markings}).

\subsection{Initial Markings}

Let $\grammar=(\nonterm,\rules)$ be any of the grammars $\grammar_1,
\ldots, \grammar_n$ (\autoref{eq:filtering}). To simplify matters at
hand, we assume w.l.o.g that the right-hand side of each rule in
$\grammar$ has exactly one occurrence of an \hrtext{} function symbol,
\ie a constant $\sgraph{a}{\asrc_1}{\asrc_2}{}$, a unary function
symbol $\restrict{\slabs}{}$ or $\rename{\alpha}{}$, or the binary
function symbol $\pop{}$, applied to $0$, $1$ or $2$ variables,
respectively. Note that each grammar can be put in this form, at the
cost of adding polynomially many extra nonterminals.

First, we annotate each nonterminal $X \in \nonterm$ with sets of
sources $\slabs$ that are \emph{visible} (\ie have been introduced by
a constant $\sgraph{a}{\asrc_1}{\asrc_2}{}$ and have not been removed
by some application of $\restrict{\slabs'}{}$) in each complete
derivation starting in $\annot{X}{\slabs}$, where $\annot{X}{\slabs}$
is a shorthand for the pair $(X,\slabs)$.  The annotated grammar
$\widehat{\grammar} = (\widehat{\nonterm}, \widehat{\rules})$ can be
built from $\grammar$ by a standard worklist iteration \ifLongVersion
The relation between $\grammar$ and $\widehat{\grammar}$ is captured
below: \else The language of the annotated grammar
$\widehat{\grammar}$ is the same as the original grammar
$\grammar$.\else\footnote{For reasons of space, the annotation
  algorithm is given in \figref{fig:annotation} from Appendix
  \ref{app:proofs}.}\fi Next, we use the annotated grammar
$\widehat{\grammar} = (\widehat{\nonterm},\widehat{\rules})$, to build
a PN $\initof{\widehat{\grammar}}$ that generates the initial markings
of $\grammar$ in the set of folded PNs (\autoref{prop:beh-fold}).

\ifLongVersion\else
\begin{textAtEnd}[category=initial]
\fi
\begin{lemma}\label{lemma:annotation}
  Let $\grammar=(\nonterm,\rules)$ be a grammar and
  $\widehat{\grammar}=(\widehat{\nonterm},\widehat{\rules})$ be the
  corresponding annotated grammar. Then
  $\alangof{}{\algof{A}}{\grammar} =
  \alangof{}{\algof{A}}{\widehat{\grammar}}$, for each
  \hrtext{} algebra $\algof{A}$.
\end{lemma}
\begin{proof}
  ``$\subseteq$'' Let $X \step{\grammar}^* \theta$ be a complete
  derivation starting with an axiom $\rightarrow X \in \rules$. By
  reverse induction on the length of this derivation we build a
  complete derivation $X^\slabs \step{\widehat{\grammar}}^* \theta$
  starting with an axiom $\rightarrow X^\slabs \in
  \widehat{\rules}$. It is clear that the rules used in this
  derivation are inserted into $\widehat{\rules}$ at line
  \ref{line:insert-rule}, whereas the axiom is inserted at line
  \ref{line:insert-axiom}. We obtain $\alangof{}{\algof{A}}{\grammar}
  \subseteq \alangof{}{\algof{A}}{\widehat{\grammar}}$ by the
  definition of the language of a grammar in $\algof{A}$. 

  \noindent``$\supseteq$'' Let $X^\slabs \step{\widehat{\grammar}}^*
  \theta$ be a complete derivation starting with an axiom $\rightarrow
  X^\slabs \in \widehat{\rules}$. Then the complete derivation $X
  \step{\grammar}^* \theta$ is obtained by applying, for each rule
  $X_0^{\slabs_0} \rightarrow \rho[X_1^{\slabs_1}, \ldots,
    X_n^{\slabs_n}] \in \widehat{\rules}$ the rule $X_0 \rightarrow
  \rho[X_1,\ldots,X_n] \in \rules$. Then $\rightarrow X$ is an axiom
  of $\grammar$, since $\rightarrow X^\slabs$ is an axiom of
  $\widehat{\grammar}$ that has been introduced at line
  \ref{line:insert-axiom}. We obtain $\alangof{}{\algof{A}}{\grammar}
  \supseteq \alangof{}{\algof{A}}{\widehat{\grammar}}$ by the
  definition of the language of a grammar in $\algof{A}$. \qed
\end{proof}
\begin{proofSketch}
  By induction on a complete derivation.
  All rules used by the derivation
  are well-typed with regard to the annotations,
  and thus available in the annotated grammar.
  \qed
\end{proofSketch}

\begin{figure}
  {\small\begin{algorithmic}[0]
    \STATE \textbf{input}: $\grammar=(\nonterm,\rules)$
    \STATE \textbf{output}: $\widehat{\grammar}=(\widehat{\nonterm}, \widehat{\rules})$
  \end{algorithmic}}
  {\small\begin{algorithmic}[1]
    \STATE \textbf{initially} $\widehat{\nonterm} := \emptyset$, 
    $\widehat{\rules} := \emptyset$, 
    $\mathit{changed} := \mathrm{true}$
    \WHILE{$\mathit{changed}$} 
    \STATE $\mathit{changed} := \mathrm{false}$
    \FOR{each $X \rightarrow \rho[X_1,\ldots,X_n] \in\rules$ and $X_1^{\slabs_1}, \ldots, X_n^{\slabs_n} \in \widehat{\nonterm}$}
    \STATE \textbf{match} $\rho[X_1,\ldots,X_n]$ \textbf{with}
    \Statex \begin{minipage}{5cm}
      \vspace*{-\baselineskip}
      \[
      \begin{array}{rcl}
        \sgraph{a}{\asrc_1}{\asrc_2}{} & \Longrightarrow & \slabs := \set{\asrc_1,\asrc_2} \\
        \restrict{\slabs'}{}(X_1) & \Longrightarrow & \slabs := \slabs_1 \cap \slabs' \\
        \rename{\alpha}{}(X_1) & \Longrightarrow & \slabs := \alpha(\slabs_1) \\
        X_1 \pop{} X_2 & \Longrightarrow & \slabs := \slabs_1 \cup \slabs_2
      \end{array}
      \]
    \end{minipage}
    \IF{$X^\slabs \not\in \widehat{\nonterm}$}
    \STATE $\widehat{\rules} := \widehat{\rules} \cup \set{X^\slabs \rightarrow \rho[X_1^{\slabs_1}, \ldots, X_n^{\slabs_n}]}$ \label{line:insert-rule}
    \STATE $\widehat{\nonterm} := \widehat{\nonterm} \cup \set{X^\slabs}$
    \STATE $\mathit{changed} := \mathrm{true}$
    \ENDIF
    \ENDFOR
    \STATE $\widehat{\rules} := \widehat{\rules} \cup \set{\rightarrow X^\slabs \mid \rightarrow X \in \rules,~ X^\slabs \in \widehat{\nonterm}}$ \label{line:insert-axiom}
    \ENDWHILE
  \end{algorithmic}}
  \caption{Bottom-up annotation of a grammar with visible sources}
  \label{fig:annotation}
\end{figure}

Formally we define $\initof{\widehat{\grammar}} \isdef (\anet,\amark_0)$ as:
\begin{align*}
  \placeof{\anet} \isdef & ~\set{S} \uplus \widehat{\nonterm} \uplus \placeof{\ptypes} \hspace*{5mm}
  \transof{\anet} \isdef ~\widehat{\rules} \\
  \weightof{\anet}(S,\rightarrow \annot{X}{\slabs}) \isdef & ~1 \hspace*{5mm}
  \weightof{\anet}(\rightarrow \annot{X}{\slabs}, \annot{X}{\slabs}) \isdef 1
  \text{, for all } \rightarrow \annot{X}{\slabs} \in \widehat{\rules}
  \\
  \weightof{\anet}(\rightarrow \annot{X}{\slabs}, q) \isdef & ~\sum_{\asrc\in\slabs} \initmarkof{\ptypeof{\asrc}}(q)
  \text{, for all } q \in \placeof{\ptypes}
  \\
  \weightof{\anet}(\annot{X}{\slabs},\annot{X_0}{\slabs_0} \rightarrow \rho[\annot{X_1}{\slabs_1},\ldots,\annot{X_n}{\slabs_n}]) \isdef & ~\left\{\begin{array}{ll}
  1 \text{, if } \annot{X}{\slabs} = \annot{X_0}{\slabs_0} \\
  0 \text{, otherwise} 
  \end{array}\right.
  \\
  \weightof{\anet}(\annot{X_0}{\slabs_0} \rightarrow \rho[\annot{X_1}{\slabs_1},\ldots,\annot{X_n}{\slabs_n}], \annot{X}{\slabs}) \isdef & ~\left\{\begin{array}{ll}
  1 \text{, if } \annot{X}{\slabs} \in \set{\annot{X_1}{\slabs_1}, \ldots, \annot{X_n}{\slabs_n}} \\
  0 \text{, otherwise} 
  \end{array}\right.
  \\
  & \text{for all } \annot{X}{\slabs} \in \widehat{\nonterm} \text{ and } \annot{X_0}{\slabs_0} \rightarrow \rho[\annot{X_1}{\slabs_1}, \ldots, \annot{X_n}{\slabs_n}] \in \widehat{\rules}
  \\
  \weightof{\anet}(\annot{X}{\slabs} \rightarrow \restrict{\slabs'}{}(\annot{X_1}{\slabs_1}), q) \isdef & \sum_{\asrc\in\slabs_1\setminus\slabs'} \initmarkof{\ptypeof{\asrc}}(q)
  \text{, for all } q \in \placeof{\ptypes}
  \\
  \amark_{0}(x) \isdef & \left\{\begin{array}{ll}
  1 \text{, if } x = S \\
  0 \text{, otherwise} 
  \end{array}\right. \text{, for all } x \in \placeof{\anet}
\end{align*}
\ifLongVersion\else
\end{textAtEnd}
\fi
This construction is quite intuitive: by reinterpreting each rule of
$\widehat{\grammar}$ as a transition of $\initof{\widehat{\grammar}}$,
we get a PN whose firing sequences mimick derivations of
$\widehat{\grammar}$ in which the rules/transitions of
$\widehat{\grammar}$ and $\initof{\widehat{\grammar}}$ are applied in
the same order. More precisely, a partial derivation having $k$
occurrences of the nonterminal variable $X$, when reinterpreted as a
firing sequence, will lead to a marking in which there are $k$ tokens
in the place represented by $X$. Assume that $\annot{X}{\slabs}
\step{\widehat{\grammar}}^* \theta$ is a complete derivation, \ie
$\theta$ is a ground \hrtext{} term. Then, every instance of some
process type $\ptype$ that occurs in the system $(\asys,\sources)
\isdef \theta^\algof{S}$ starts in a vertex labeled by a source label
$\asrc \in \sourcelabels$ such that, either: \begin{enumerate}[(a)]
  \item\label{it1:sources-instances} $\asrc$ is removed by an
    application of $\restrict{\slabs'}{}$ such that $\asrc \not\in
    \slabs'$, then $\asrc$ occurs in the subtree of $\theta$ rooted at
    the particular occurrence of $\restrict{\slabs'}{}$ that removed
    it, or
  \item\label{it2:sources-instances} $\asrc$ is visible at the root of
    $\theta$, \ie $\asrc\in\slabs$.
\end{enumerate}

\begin{figure}[t!]
  \vspace*{-\baselineskip}
  \begin{center}
    \scalebox{0.7}{
    \begin{tikzpicture}
      \node (tokC) [petri-p,draw=black,label=-90:{$tokC$}] {};
      \node (nokC) [petri-p,draw=black,label=-90:{$nokC$}] at ($(tokC) + (1,-0.5)$) {};

      \node (tok) [petri-p,draw=black,label=-90:{$tok$}] at ($(tokC) + (-3,0)$) {};
      \node (nok) [petri-p,draw=black,label=-90:{$nok$}] at ($(tok) + (1,-0.5)$) {};
      \node (work) [petri-p,draw=black,label=-90:{$work$}] at ($(nok) + (1,-0.5)$) {};

      \node (S) [petri-p,draw=black,label=90:{$S$}] at ($(tok) + (0,2.5)$) {};
      \node (C) [petri-p,draw=black,label=90:{$C$}] at ($(S) + (2,0)$) {};
      \node[petri-tok] at (S) {};

      \node[petri-t2,draw=black,fill=black](S/C) at ($(S)!0.5!(C)$) {};
      \node[petri-t,draw=black,fill=black](C/Co) at ($(C) + (1,-1.5)$) {};
      \node[petri-t,draw=black,fill=black](C/C+Pr) at ($(C) + (0,-1.5)$) {};

      \draw[->,thick,line width=1pt] (S) edge (S/C);
      \draw[->,thick,line width=1pt] (S/C) edge (C);
      \draw[->,thick,line width=1pt] (S/C) edge[out=-30,in=90] (nok);
      \draw[->,thick,line width=1pt] (C) edge[bend right=30] (C/C+Pr);
      \draw[->,thick,line width=1pt] (C/C+Pr) edge[bend right=30] (C);
      \draw[->,thick,line width=1pt] (C/C+Pr) edge[out=-90,in=60] (nok);
      \draw[->,thick,line width=1pt] (C) edge[out=-30,in=90] (C/Co);
      \draw[->,thick,line width=1pt] (C/Co) edge (tokC);
      \draw[->,thick,line width=1pt] (C/Co) edge[out=-90,in=30] (nok);
    \end{tikzpicture}}
    \hspace*{12mm}
    \scalebox{0.7}{
    \begin{tikzpicture}
      \node (tok) [petri-p,draw=black,label=-90:{$tok$}] {};
      \node (nok) [petri-p,draw=black,label=-90:{$nok$}] at ($(tok) + (1,-0.5)$) {};
      \node (work) [petri-p,draw=black,label=-90:{$work$}] at ($(nok) + (1,-0.5)$) {};

      \node (tokC) [petri-p,draw=black,label=-90:{$tokC$}] at ($(tok) + (-2,0)$) {};
      \node (nokC) [petri-p,draw=black,label=-90:{$nokC$}] at ($(tokC) + (1,-0.5)$) {};

      \node (S) [petri-p,draw=black,label=90:{$S$}] at ($(tokC) + (0,2.5)$) {};
      \node (Z) [petri-p,draw=black,label=90:{$Z$}] at ($(S) + (2.5,0)$) {};
      \node[petri-tok] at (S) {};

      \node[petri-t2,draw=black,fill=black](S/Z) at ($(S)!0.5!(Z)$) {};
      \node[petri-t,draw=black,fill=black](Z/Co) at ($(Z) + (1,-1.5)$) {};
      \node[petri-t,draw=black,fill=black](Z/Z+Pr) at ($(Z) + (0,-1.5)$) {};

      \draw[->,thick,line width=1pt] (S) edge (S/Z);
      \draw[->,thick,line width=1pt] (S/Z) edge (Z);
      \draw[->,thick,line width=1pt] (S/Z) edge[out=-130,in=90] (tokC);
      \draw[->,thick,line width=1pt] (Z) edge[bend right=30] (Z/Z+Pr);
      \draw[->,thick,line width=1pt] (Z/Z+Pr) edge[bend right=30] (Z);
      \draw[->,thick,line width=1pt] (Z/Z+Pr) edge[out=-90,in=120] (nok);
      \draw[->,thick,line width=1pt] (Z) edge[out=-30,in=90] (Z/Co);
      \draw[->,thick,line width=1pt] (Z/Co) edge[out=-90,in=60] (nok);
    \end{tikzpicture}}
  \end{center}
  \vspace*{-\baselineskip}
  \caption{ $\initof{\widehat{\grammar}_{Chain}}$ (left) and
    $\initof{\widehat{\grammar}_{Star}}$ (right), showing how to
    initialize the marking of stars and chains generated by the
    grammars presented in \exref{ex:parameterized-grammars}. }
  \label{fig:initof-example}
  \vspace*{-\baselineskip}
\end{figure}

When processing the rules of the form $\annot{X}{\slabs} \rightarrow
\restrict{\slabs'}{}(\annot{X_1}{\slabs_1})$ and $\rightarrow
\annot{X}{\slabs}$ in the construction of
$\initof{\widehat{\grammar}}$, we add an outgoing edge to each place
$q$ that is initially marked in $\ptypeof{\asrc}$, for some $\asrc \in
\slabs_1 \setminus \slabs'$ or $\asrc \in \slabs$. Then, for every
instance of the process type that occurs in the system, its initial
marking will be added to $q$ by a firing sequence of
$\initof{\widehat{\grammar}}$. \ifLongVersion
\begin{example}
  \figref{fig:initof-example} shows the PNs obtained by this
  construction applied to $\widehat{\grammar}_{Chain}$ and
  $\widehat{\grammar}_{Star}$, \ie the annotated versions of the two
  grammars from \exref{ex:parameterized-grammars}. Intuitively, a
  chain system generated by the grammar $\grammar_\mathit{Chain}$, see
  \figref{fig:proc-typ1} (b) top, has $1$ instance of the process type
  $\mathit{Cont}$ and $n \geq 2$ instances of the process type
  $\mathit{Proc}$. This means that, initially, there is $1$ token on
  the place $\mathit{tokC}$ and $n\geq2$ tokens in $\mathit{nok}$, as
  in \figref{fig:initof-example} (a) Similarly, the initial marking of
  a star system, see \figref{fig:proc-typ1} (b) bottom, generated by
  the grammar $\grammar_\mathit{Star}$ has $1$ token in
  $\mathit{tokC}$ and $n\geq1$ tokens in $\mathit{nok}$, as in
  \figref{fig:initof-example} (b).
  
  We illustrate the construction of
  $\initof{\widehat{\grammar}_\mathit{Chain}}$ by explaining how the
  transitions labeled (1) and (2) in \figref{fig:initof-example} have
  been introduced: \begin{compactenum}[(1)]
    \item corresponds to the rule $\to C^{\set{\sigma_1}}$: the final
      system has one visible source label $\sigma_1$, whose process
      type is $\ptypeof{\sigma_1} = Proc$. When this rule is applied,
      it creates one instance of $Proc$, hence the outgoing edge to
      $nok$. Since this rule is, its corresponding transition has an
      incoming edge from the start place $S$ of
      $\initof{\widehat{\grammar}_\mathit{Chain}}$ and produces an
      instance of $C^\set{\sigma_1}$, hence the outgoing edge to $C$.
    \item corresponds to the rule $C^\set{\sigma_1} \to
      \restrict{\set{\sigma_1}}{} \rename{\sigma_1 \leftrightarrow
        \sigma_2}{} (C^\set{\sigma_1} \pop{}
      (rel,get)_{(\sigma_1,\sigma_2)})$: it consumes one occurrence of
      $C^\set{\sigma_1}$ and produces another, hence having both an
      incoming and outgoing edge to $C$.  The $\restrict{}{}$
      operation hides a $\sigma_2$-source, where
      $\ptypeof{\sigma_2}=\mathit{Proc}$, thus each application of
      this rule is responsible for one additional instance of the
      process type $Proc$.
  \end{compactenum}
\end{example}
\else We refer to \figref{fig:initof-example} for examples of
initialization PNs obtained by this construction applied to
$\widehat{\grammar}_{Chain}$ and $\widehat{\grammar}_{Star}$, \ie the
annotated versions of the two grammars from
\exref{ex:parameterized-grammars}. \fi For a PN $\amarkednet$ and a
set of places $\mathcal{Q} \subseteq \placeof{\amarkednet}$, we denote
by:
\begin{align}\label{eq:zero-reach}
\zreach{\mathcal{Q}}{\amarkednet} \isdef
\{\amark\in\reach{\amarkednet} \mid \amark(q)=0
\text{, for all } q \in \mathcal{Q}\}
\end{align}
the set of reachable markings having zero tokens in a place from
$\mathcal{Q}$. For a set $\mathcal{M}$ of markings and a set
$\mathcal{Q}$ of places, we denote by
$\proj{\mathcal{M}}{\mathcal{Q}}$ the set of restrictions of each
$\amark\in\mathcal{M}$ to the places in $\mathcal{Q}$. The relation
between the set of folded PNs described by an annotated grammar
$\widehat{\grammar}$ and the PN $\initof{\widehat{\grammar}}$ is
formally captured below:

\begin{lemmaE}[][category=proofs]\label{lemma:initial-markings}
  Let $\widehat{\grammar}=(\widehat{\nonterm},\widehat{\rules})$ be an
  annotated grammar. Then, we have:
  \begin{align*}
  \set{\initmarkof{\amarkednet} \mid (\amarkednet,\sources) \in \alangof{}{\absof{\algof{B}}}{\widehat{\grammar}}}
  = \proj{\zreach{\set{S} \uplus \widehat{\nonterm}}{\initof{\widehat{\grammar}}}}{\placeof{\ptypes}}
  \end{align*}
\end{lemmaE}
\begin{proofE}
  Let $\initof{\widehat{\grammar}} \isdef (\anet,\amark_0)$ in the
  rest of this proof. 
  
  \noindent ``$\subseteq$'' Let $(\amarkednet,\sources) \in
  \alangof{}{\absof{\algof{B}}}{\widehat{\grammar}}$ be a folded PN
  and $\amark\isdef\initmarkof{\amarkednet}$ be its inital marking. By
  the definition of the $\absof{\algof{B}}$ algebra, there exists an
  open behavior $(\overline{\amarkednet},\overline{\sources}) \in
  \alangof{}{\algof{B}}{\widehat{\grammar}}$ such that
  $\foldof{\overline{\amarkednet},\overline{\sources}} =
  (\amarkednet,\sources)$. Let $\annot{X}{\slabs}
  \step{\widehat{\grammar}}^* \theta$ be a complete derivation,
  starting with an axiom $\rightarrow \annot{X}{\slabs} \in
  \widehat{\rules}$ and ending with a ground term $\theta$, such that
  $\theta^\algof{B} = (\overline{\amarkednet},\overline{\sources})$.
  Let $\overline{\amark} \isdef \initmarkof{\overline{\amarkednet}}$
  be the initial marking of $\overline{\amarkednet}$. Note that
  $\overline{\amark}(q,v) = 1$ if $q$ is initially marked in its
  process type and $\overline{\amark}(q,v) = 0$ otherwise, for all
  $(q,v) \in \placeof{\overline{\amarkednet}}$. Then $\amark(q) =
  \sum_{(q,v)\in\placeof{\overline{\amarkednet}}}
  \overline{\amark}(q,v)$, for all $q \in \placeof{\ptypes}$ (we
  recall that the $\foldrel{\overline{\sources}}$-equivalence classes
  are denoted by places from the process types). Let $\amark_{0}
  \fire{\vec{t}} \amark'$ be the firing sequence of
  $\initof{\widehat{\grammar}}$ that mimicks the derivation
  $\annot{X}{\slabs} \step{\widehat{\grammar}}^* \theta$. By the
  definition of $\initof{\widehat{\grammar}}$ we have $\amark' \in
  \reach{\initof{\widehat{\grammar}}}$ and $\amark'(x)=0$, for all $x
  \in \set{S} \uplus \widehat{\nonterm}$. It remains to prove that
  $\amark'$ and $\amark$ agree over $\placeof{\ptypes}$. Let $q \in
  \placeof{\ptypes}$ be a place and distinguish the following
  cases: \begin{itemize}
  \item If $q$ is not initially marked in its process type then
    $\amark(q)=0$. But $\amark'(q)=0$ follows from the definition of
    $\initof{\widehat{\grammar}}$, because the only places from
    $\placeof{\ptypes}$ having an incoming edge in
    $\initof{\widehat{\grammar}}$ are the places from
    $\placeof{\ptypes}$ that are initially marked in their respective
    process types.
  \item Else, $q$ is initially marked in its process type and
    $\amark(q)$ is the number of instances of that process type in the
    system $\theta^\algof{S}$. Then, $\amark'(q)=\amark(q)$ by the argument made at
    points (\ref{it1:sources-instances}) and
    (\ref{it2:sources-instances}) above.
  \end{itemize}

  \noindent ``$\supseteq$'' Let $\amark\in
  \reach{\initof{\widehat{\grammar}}}$ be a marking such that
  $\amark(x)=0$, for all $x \in \set{S} \uplus \widehat{\nonterm}$.
  Since each firing sequence of $\initof{\widehat{\grammar}}$ starting
  from $\amark_0$ corresponds to a derivation of $\widehat{\grammar}$
  starting from an axiom $\rightarrow \annot{X}{\slabs} \in
  \widehat{\rules}$, the firing sequence $\amark_0 \fire{\vec{t}}
  \amark$ corresponds to a complete derivation $\annot{X}{\slabs}
  \step{\widehat{\grammar}}^* \theta$, because $\amark(x)=0$, for all
  $x \in \set{S}\uplus\widehat{\nonterm}$, by the definition of
  $\initof{\widehat{\grammar}}$. Let $(\overline{\amarkednet},
  \overline{\sources}) = \theta^\algof{B}$ be the behavior of the
  system built by this derivation and $\overline{\amark} \isdef
  \initmarkof{\overline{\amarkednet}}$ be its initial marking.  Let
  $(\amarkednet,\sources) \isdef \foldof{\overline{\amarkednet},
    \overline{\sources}} = \theta^{\absof{\algof{B}}}$ be the folded
  PN of this behavior and $\amark' \isdef \initmarkof{\amarkednet}$ be
  its initial marking,
  \ie $\amark'(q)=\sum_{(q,v)\in\placeof{\overline{\amarkednet}}} \overline{\amark}(q,v)$,
  for all $q \in \placeof{\ptypes}$.
  We prove that $\amark$ and $\amark'$ agree over $\placeof{\ptypes}$,
  by distinguishing the cases below: \begin{itemize}
    \item If $q$ is not initially marked in its process type then
      $\amark'(q)=0$. In this case $\amark(q)=0$ because there are no
      incoming edges to $q$ in $\initmarkof{\widehat{\grammar}}$, by the
      definition of the latter. 
    \item Else, $q$ is initially marked in its process types and
      $\amark'(q)$ is the number of instances of that process type in
      the system $\theta^\algof{S}$. Then, $\amark(q)=\amark'(q')$, by
      the argument made at points (\ref{it1:sources-instances}) and
      (\ref{it2:sources-instances}) above. \qed
  \end{itemize}
  \qed
\end{proofE}
\begin{proofSketch}
  Read the above as: the initial markings of nets of $\alangof{}{\absof{\algof{B}}}{\widehat{\grammar}}$
  are exactly the reachable markings of $\initof{\widehat{\grammar}}$ once there are no
  more tokens in places representing nonterminals.
  The proof stems from the earlier observation that if partial derivations with
  $k$ instances of $X \in \widehat{\nonterm}$ match with firing sequences
  that put $k$ tokens in place $X$,
  then complete derivations must have zero tokens in every place that represents
  a nonterminal.
  \qed
\end{proofSketch}

\subsection{Soundness}

We now have all the elements to describe our counting abstraction
method and prove its soundness, \ie if the reachability
(resp. coverability) problem has a negative answer for the
abstraction, then the concrete reachability (resp. coverability)
problem has a negative answer. 

For each grammar $\grammar_i=(\nonterm_i,\rules_i)$ defined at
(\autoref{eq:filtering}), that corresponds to the (open) net
$(\anet_i,\sources_i)$, for $i\in\interv{1}{n}$, we define the PN
$\foldpn{\grammar_i} \isdef (\amarkednet_{~i}, \sources_i)$, where:
\begin{align*}
  \placeof{\amarkednet_{~i}} \isdef & ~\placeof{\initof{\widehat{\grammar}_i}} \cup \placeof{\anet_i} \hspace*{3mm}
  \transof{\amarkednet_{~i}} \isdef ~\transof{\initof{\widehat{\grammar}_i}} \uplus \transof{\anet_i} \hspace*{3mm}
  \weightof{\amarkednet_{~i}} \isdef ~\weightof{\initof{\widehat{\grammar}_i}} \cup \weightof{\anet_i} \hspace*{3mm}
  \initmarkof{\amarkednet_{~i}} \isdef ~\initmarkof{\initof{\widehat{\grammar}_i}}
\end{align*}
Note that $\zreach{\placeof{\initof{\widehat{\grammar}_i}} \setminus
  \placeof{\anet_i}}{\foldpn{\grammar_i}}$ is the set of markings of
$\foldpn{\grammar_i}$ that can be reached \emph{after} the full
generation of its initial marking, \ie from those markings that have
no more tokens in any of the places of
$\initof{\widehat{\grammar}_i}$, excepted the initially marked places
of $\placeof{\ptypes}$.

Our verification method for the grammar-based parameterized
reachability and coverability problems
(\defref{def:grammar-parameterized-verif}) relies on the construction
of a finite number of PNs $\foldpn{\grammar_1}, \ldots,
\foldpn{\grammar_n}$ from a given grammar $\grammar$ that describes a
set of systems. \ifLongVersion The relation between the reachability (resp. cover)
set of the original parameterized system and its finitary abstraction
is captured by the following lemma: \fi

\ifLongVersion\else
\begin{textAtEnd}[category=proofs]
\fi
\begin{lemma}\label{lemma:soundness}
  For each grammar $\grammar$ such that
  $\alangof{}{\finabsof{\algof{B}}}{\grammar} =
  \set{(\anet_1,\sources_1), \ldots, (\anet_n, \sources_n)}$, we have:
  \begin{align}
    \label{eq:reach}
    \bigcup_{(\amarkednet,\sources)\in\alangof{}{\algof{B}}{\grammar}} \hspace*{-5mm} \reach{\amarkednet}_{/\foldrel{\sources}} \subseteq &
    ~\reach{\alangof{}{\absof{\algof{B}}}{\grammar}} = \bigcup_{i=1}^n \proj{\zreach{\placeof{\foldpn{\grammar_i}} \setminus \placeof{\ptypes}}{\foldpn{\grammar_i}}}{\placeof{\ptypes}} \\[-1mm]
    \label{eq:cover}
    \bigcup_{(\amarkednet,\sources)\in\alangof{}{\algof{B}}{\grammar}} \hspace*{-5mm} \cover{\amarkednet}_{/\foldrel{\sources}} \subseteq &
    ~\cover{\alangof{}{\absof{\algof{B}}}{\grammar}} = \bigcup_{i=1}^n \proj{\cover{\foldpn{\grammar_i}}}{\placeof{\ptypes}}
  \end{align}
  where $\grammar_i$ is a grammar such that
  $\alangof{}{\finabsof{\algof{B}}}{\grammar_i} =
  \set{(\anet_i,\sources_i)}$, for all $i \in \interv{1}{n}$.
\end{lemma}
\begin{proof}
  (\ref{eq:reach}) For the first inclusion, we compute:
  \begin{align*}
    \bigcup_{(\amarkednet,\overline{\sources})\in\alangof{}{\algof{B}}{\grammar}}\reach{\amarkednet}_{/\foldrel{\sources}} \subseteq
    & ~\bigcup_{(\amarkednet,\overline{\sources})\in\alangof{}{\algof{B}}{\grammar}}\reach{\amarkednet_{~/\foldrel{\sources}}} \text{, by \lemref{lemma:quotient-soundness}} \\
    = & ~\bigcup_{(\amarkednet,\overline{\sources})\in\alangof{}{\algof{B}}{\grammar}}\reach{\foldof{\amarkednet,\overline{\sources}}} \\
    = & ~\reach{\foldof{\alangof{}{\algof{B}}{\grammar}}} \\
    = & ~\reach{\alangof{}{\absof{\algof{B}}}{\grammar}} \text{, by the definition of $\absof{\algof{B}}$}
  \end{align*}
  For the second equality, we compute:
  \begin{align*}
    \alangof{}{\absof{\algof{B}}}{\grammar} = & ~\bigcup_{i=1}^n \alangof{}{\absof{\algof{B}}}{\grammar_i} =
    ~\bigcup_{i=1}^n \alangof{}{\absof{\algof{B}}}{\widehat{\grammar}_i} \text{, by \lemref{lemma:annotation}} \\
    = & \bigcup_{i=1}^n \set{((\anet_i,\amark_0),\sources_i) \mid \amark_0 \in \proj{\big(\zreach{\placeof{\foldpn{\grammar_i}} \setminus \placeof{\ptypes}}{\initof{\widehat{\grammar}_i}}\big)}{\placeof{\ptypes}}}
    \text{, by \lemref{lemma:initial-markings}}
  \end{align*}
  The second equality follows from the above, by the definition of
  $\foldpn{\grammar_i}$.

  \vspace*{\baselineskip}
  \noindent(\ref{eq:cover}) The first inclusion is proved along the
  same lines as the first inclusion of point (\ref{eq:reach}). We
  prove the second equality:
  \[\cover{\alangof{}{\absof{\algof{B}}}{\grammar}} = \bigcup_{i=1}^n \proj{\left(\cover{\foldpn{\grammar_i}}\right)}{\placeof{\ptypes}}\]
  \noindent``$\subseteq$'' Let $\amark \in
  \cover{\alangof{}{\absof{\algof{B}}}{\grammar}}$ be a marking. Then
  there exists a PN $((\anet_i,\amark_0),\sources_i) \in
  \alangof{}{\absof{\algof{B}}}{\grammar}$ having a firing sequence
  $\amark_0 \fire{\vec{t}} \amark'$ such that $\amark \leq \amark'$,
  for some $i \in \interv{1}{n}$. By \lemref{lemma:initial-markings},
  we have $\amark_0 \in
  \proj{\big(\zreach{\placeof{\foldpn{\grammar_i}} \setminus
      \placeof{\ptypes}}{\initof{\widehat{\grammar}_i}}\big)}{\placeof{\ptypes}}$,
  hence $\initof{\widehat{\grammar}_i}$ has a firing sequence
  $\initmarkof{\initof{\widehat{\grammar}_i}} \fire{\vec{t}_0}
  \amark'_0$, for some marking $\amark'_0$ of
  $\initof{\widehat{\grammar}_i}$ that agrees with $\amark_0$ over
  $\placeof{\ptypes}$. Then,
  $\initmarkof{\foldpn{\grammar_i}} = \initmarkof{\initof{\widehat{\grammar}_i}} \fire{\vec{t}_0}
  \amark'_0 \fire{\vec{t}} \amark''$ is a firing sequence of
  $\foldpn{\grammar_i}$, for some marking $\amark''$ of
  $\foldpn{\grammar_i}$ that agrees with $\amark$ over
  $\placeof{\ptypes}$, thus
  $\amark\in\proj{\left(\cover{\foldpn{\grammar_i}}\right)}{\placeof{\ptypes}}$.

  \vspace*{\baselineskip}
  \noindent``$\supseteq$'' Let
  $\amark\in\proj{\left(\cover{\foldpn{\grammar_i}}\right)}{\placeof{\ptypes}}$
  be a marking and $\amark' \in \cover{\foldpn{\grammar_i}}$ be the
  extension of $\amark$ to $\placeof{\foldpn{\grammar_i}}$ that
  witnesses the membership, for some $i \in \interv{1}{n}$. Then there
  exists a firing sequence $\initmarkof{\foldpn{\grammar_i}}
  \fire{\vec{t}} \amark''$ of $\foldpn{\grammar_i}$ such that $\amark'
  \leq \amark''$. W.l.o.g. we consider that $\vec{t} =
  \vec{t}';\vec{t}''$, where $\vec{t}'$ and $\vec{t}''$ contain
  transitions from $\initof{\widehat{\grammar}_i}$ and $\anet_i$,
  respectively. Note that, since these sets of transitions are
  disjoint and independent, any firing sequence of
  $\foldpn{\grammar_i}$ can be rearranged in this way. Hence
  $\initmarkof{\foldpn{\grammar_i}} \fire{\vec{t}'} \amark_0
  \fire{\vec{t}''} \amark''$ is a firing sequence of
  $\foldpn{\grammar_i}$, for some marking $\amark_0 \in
  \reach{\foldpn{\grammar_i}}$. By the definition of
  $\initmarkof{\widehat{\grammar}_i}$, there exists a marking
  $\amark'_0\in\zreach{\placeof{\foldpn{\grammar_i}}
    \setminus \placeof{\ptypes}}{\initof{\widehat{\grammar}_i}}$ such
  that $\amark_0 \leq \amark'_0$. Indeed, $\amark'_0$ can be obtained
  from $\amark_0$ by moving the tokens from $\set{S} \uplus
  \widehat{\nonterm}_i$ into $\placeof{\ptypes}$, thus obtaining an
  initial marking for $\anet_i$, extended with the places $\set{S}
  \uplus \widehat{\nonterm}_i$ all having zero tokens.  By
  \lemref{lemma:initial-markings}, we obtain
  $((\anet_i,\proj{\amark'_0}{\placeof{\ptypes}}),\sources_i) \in
  \alangof{}{\absof{\algof{B}}}{\widehat{\grammar}_i}$, hence
  $((\anet_i,\proj{\amark'_0}{\placeof{\ptypes}}),\sources_i) \in
  \alangof{}{\absof{\algof{B}}}{\grammar_i}$, by
  \lemref{lemma:annotation}. By the monotonicity of firing sequences,
  we obtain that $\amark'_0 \fire{\vec{t}''} \amark'''$ is a firing
  sequence of $\foldpn{\grammar_i}$, for some marking $\amark''' \in
  \reach{\foldpn{\grammar_i}}$ such that $\amark'' \leq \amark'''$.
  By the definition of $\foldpn{\grammar_i}$, we have that
  $\proj{\amark'_0}{\placeof{\ptypes}} \fire{\vec{t}''}
  \proj{\amark'''}{\placeof{\ptypes}}$ is a firing sequence of
  $(\anet_i,\proj{\amark'_0}{\placeof{\ptypes}})$.  Then, we obtain
  $\amark' \leq \amark'''$ by transitivity,
  \ie $\proj{\amark'}{\placeof{\ptypes}} \in \cover{\anet_i,\amark'_0}$
  for some $i\in\interv{1}{n}$, thus
  $\amark\in\cover{\alangof{}{\absof{\algof{B}}}{\grammar}}$. \qed
\end{proof}
\begin{proofSketch}
  Due to the transitivity of reachability,
  we straightforwardly apply in succession the equalities and inclusions provided by
  \lemref{lemma:quotient-soundness}, \lemref{lemma:annotation}, \lemref{lemma:initial-markings}. 
  \qed
\end{proofSketch}
\ifLongVersion\else
\end{textAtEnd}
\fi

The soundness of the method is formally captured below:

\begin{theoremE}[][category=proofs]\label{thm:soundness}
  Let $\grammar$ be an \hrtext{} grammar such that
  $\alangof{}{\finabsof{\algof{B}}}{\grammar}=\set{(\anet_1,\sources_1),
    \ldots, (\anet_n,\sources_n)}$, $\mathcal{Q} \subseteq
  \placeof{\ptypes}$ a set of places, $\amark : \mathcal{Q}
  \rightarrow \nat$ a mapping. Then, one can effectively build
  grammars $\grammar_1,\ldots,\grammar_n$ such that
  $\alangof{}{\finabsof{\algof{B}}}{\grammar_i}=\set{(\anet_i,\sources_i)}$,
  for all $i\in\interv{1}{n}$ and:
  \begin{enumerate}
  \item\label{it1:thm:soundness}
    $\paramreach{\grammar}{\mathcal{Q}}{\amark}$ has a negative answer
    if
    $\amark\not\in\bigcup_{i=1}^n\proj{\zreach{\placeof{\foldpn{\grammar_i}}
        \setminus
        \placeof{\ptypes}}{\foldpn{\grammar_i}}}{\mathcal{Q}}$.
  \item\label{it2:thm:soundness}
    $\paramcover{\grammar}{\mathcal{Q}}{\amark}$ has a negative answer
    if
    $\amark\not\in\bigcup_{i=1}^n\proj{\cover{\foldpn{\grammar_i}}}{\mathcal{Q}}$.
  \end{enumerate}
\end{theoremE}
\begin{proofE}
  The effective construction of $\grammar_1,\ldots,\grammar_n$ follows
  from \propref{prop:effective-finite-abstraction} and the
  effectiveness of the Fitering Theorem (\thmref{thm:filtering}).

  \vspace*{\baselineskip}
  \noindent(\ref{it1:thm:soundness}) We prove the contrapositive: $\paramreach{\grammar}{\mathcal{Q}}{\amark}$ has a positive answer
  \begin{align*}
    \iff & ~\exists \asys \in \alangof{}{\algof{S}}{\grammar} ~\exists \overline{\amark} \in \reach{\behof{\asys}} ~\forall q \in \mathcal{Q} ~.~
    \hspace*{-5mm} \sum_{v \in \vertof{\asys} \mid q \in \placeof{\vlabof{\asys}(v)}} \hspace*{-5mm} \overline{\amark}(q) = \amark(q)
    \text{, by definition of $\mathsf{Reach}$} \\
    \iff & ~\exists (\amarkednet,\sources) \in \alangof{}{\algof{B}}{\grammar} ~\exists \overline{\amark} \in \reach{\amarkednet} ~\forall q \in \mathcal{Q} ~.~
    \hspace*{-5mm} \sum_{v \in \vertof{\asys} \mid q \in \placeof{\vlabof{\asys}(v)}} \hspace*{-5mm} \overline{\amark}(q) = \amark(q)
    \text{, by \propref{prop:sys-beh}} \\
    \iff & ~\exists (\amarkednet,\sources) \in \alangof{}{\algof{B}}{\grammar} ~\exists \amark' \in \reach{\amarkednet}_{/\foldrel{\sources}} ~.~ \proj{\amark'}{\mathcal{Q}} = \amark \\
    \Longrightarrow & ~\exists (\amarkednet,\sources) \in \alangof{}{\algof{B}}{\grammar} ~\exists \amark' \in \reach{\amarkednet_{~/\foldrel{\sources}}} ~.~ \proj{\amark'}{\mathcal{Q}} = \amark
    \text{, by \lemref{lemma:soundness}, \autoref{eq:reach}} \\
    \iff & \amark \in \proj{\reach{\alangof{}{\absof{\algof{B}}}{\grammar}}}{\mathcal{Q}} \text{, by \propref{prop:beh-fold}} \\
    \iff & \amark \in \bigcup_{i=1}^n \proj{\left(\zreach{\placeof{\foldpn{\grammar_i}} \setminus \placeof{\ptypes}}{\foldpn{\grammar_i}}\right)}{\placeof{\ptypes}}
    \text{, by \lemref{lemma:soundness}, \autoref{eq:reach}}
  \end{align*}
  
  \vspace*{\baselineskip}
  \noindent(\ref{it2:thm:soundness}) Along the same lines as
  (\ref{it1:thm:soundness}), using \lemref{lemma:soundness},
  \autoref{eq:cover} instead of \autoref{eq:reach}. \qed
\end{proofE}
\noindent Note that, if $\paramreach{\grammar}{\mathcal{Q}}{\amark}$
(\resp $\paramcover{\grammar}{\mathcal{Q}}{\amark}$) has a negative
answer, then each instance the parameterized system described by
$\grammar$ (\ie the set of systems $\alangof{}{\algof{S}}{\grammar}$)
is safe with respect to the property encoded by the marking $\amark$,
\ie does not reach (\resp cover) the marking $\amark$ over the set of
places $\mathcal{Q}$. For instance, mutual exclusion (only one process
of a certain type in a certain state) is naturally encoded as a
coverability problem.


\newcommand{\tcount}[2]{\#_{#1}(#2)}
\newcommand{\fromcount}[2]{\tcount{{#1}\to}{#2}}
\newcommand{\tocount}[2]{\tcount{\to{#1}}{#2}}

\section{A Decidable Fragment}
\label{sec:pebble-passing}

We have shown that the parameterized coverability problem is
undecidable, for systems with fairly simple network topologies
(chains) and unrestricted process types
(\thmref{thm:undecidability}). We refine this result by proving that
only restricting the process types, but not the topology of the
network, suffices to recover decidability. Albeit based on a simple
communication pattern (\ie passing a pebble from one node to a
neighbour having no pebble), our decidable fragment is non-trivial: we
found it to be in \twoexptime, with a \pspace-hard lower bound.
 

\subsection{Pebble-Passing Systems}

The class of pebble-passing systems (\ppstext) is defined by restricting
the process types and interactions of a system, as in
\figref{fig:send-recv}. We give the formal definition below:

\begin{definition}\label{def:pebble-passing-systems}
  Let $\ptypes_\pps$ be a set of process types, where
  $\placeof{\ptype} = \set{q^\ptype_\bot, q^\ptype_\top}$ and
  $\transof{\ptype} = \obstransof{\ptype} = \set{\send, \recv}$, such
  that $\prepost{\send} = (q^\ptype_\top, q^\ptype_\bot)$ and
  $\prepost{\recv} = (q^\ptype_\bot, q^\ptype_\top)$, for each $\ptype
  \in \ptypes_\pps$. Let
  $\ealpha_\pps\isdef\set{(\send,\recv),(\recv,\send)}$ be a set of
  edge labels. A system $\asys=(\verts,\edges,\vlab)$ over
  $\ptypes_\pps$ and $\ealpha_\pps$ is said to be
  \emph{pebble-passing}.
\end{definition}
Intuitively, a token in $q^\ptype_\top$ (\ie $\amark(q^\ptype_\top) =
1$) represents the ownership of a ressource, called \emph{pebble}, and
a token in $q^\ptype_\bot$ is the absence of a pebble, called
\emph{hole}. Since each process type is automata-like (\ie has a token
in exactly one place), each node of the system can have either a
pebble or a hole, in all the reachable markings of its behavior
(\autoref{def:behavior}). An edge $(v, (\send, \recv), v')$
(\resp $(v, (\recv, \send), v')$) will be denoted $v \to v'$
(\resp $v' \to v$). Intuitively, firing an interaction $v \to v'$
moves a pebble from $v$ to $v'$ and simultaneously moves a hole from
$v'$ to $v$. Thus each transition preserves the total numbers of
pebbles and holes in the system, respectively. 


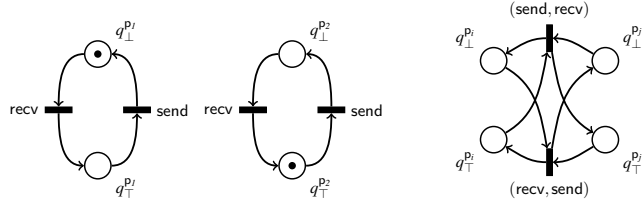
\begin{figure}[t!]
  \vspace*{-\baselineskip}
    \begin{center}
      \scalebox{0.7}{
        \begin{tikzpicture}[node distance=1.5cm]
            \tikzstyle{every state}=[inner sep=3pt,minimum size=20pt]
            \node (q0)[petri-p,draw=black,label=30:{$\mathit{q^{\ptype_1}_\bot}$}]{};
            \node (send)[petri-t,draw=black,fill=black,below right of=q0,xshift=-1em,label=0:{$\send$}]{};
            \node (recv)[petri-t,draw=black,fill=black,below left of=q0,xshift=1em,label=180:{$\recv$}]{};
            \node (q1)[petri-p,draw=black,below left of=send,xshift=1em,label=-30:{$\mathit{q^{\ptype_1}_\top}$}]{};
            \node (tok)[petri-tok] at (q0) {};

            \path (q0) edge [->,thick,line width=1pt,out=180,in=90] (recv);
            \path (recv) edge [->,thick,line width=1pt,out=-90,in=180] (q1);
            \path (q1) edge [->,thick,line width=1pt,out=0,in=-90] (send);
            \path (send) edge [->,thick,line width=1pt,out=90,in=0] (q0);
        \end{tikzpicture}
        \begin{tikzpicture}[node distance=1.5cm]
            \tikzstyle{every state}=[inner sep=3pt,minimum size=20pt]
            \node (q0)[petri-p,draw=black,label=30:{$\mathit{q^{\ptype_2}_\bot}$}]{};
            \node (send)[petri-t,draw=black,fill=black,below right of=q0,xshift=-1em,label=0:{$\send$}]{};
            \node (recv)[petri-t,draw=black,fill=black,below left of=q0,xshift=1em,label=180:{$\recv$}]{};
            \node (q1)[petri-p,draw=black,below left of=send,xshift=1em,label=-30:{$\mathit{q^{\ptype_2}_\top}$}]{};
            \node (tok)[petri-tok] at (q1) {};

            \path (q0) edge [->,thick,line width=1pt,out=180,in=90] (recv);
            \path (recv) edge [->,thick,line width=1pt,out=-90,in=180] (q1);
            \path (q1) edge [->,thick,line width=1pt,out=0,in=-90] (send);
            \path (send) edge [->,thick,line width=1pt,out=90,in=0] (q0);
        \end{tikzpicture}
      }
        \qquad
      \scalebox{0.7}{
        \begin{tikzpicture}[node distance=1.5cm]
            \tikzstyle{every state}=[inner sep=3pt,minimum size=20pt]
            \node (q0p1)[petri-p,draw=black,label=150:{$\mathit{q^{\ptype_i}_\bot}$}]{};
            \node (q1p1)[petri-p,draw=black,below of=q0p1,label=-150:{$\mathit{q^{\ptype_i}_\top}$}]{};
            \node (move12)[petri-t2,draw=black,fill=black,above right of=q0p1,yshift=-2em,label=90:{$(\send,\recv)$}]{};
            \node (move21)[petri-t2,draw=black,fill=black,below right of=q1p1,yshift=2em,label=-90:{$(\recv,\send)$}]{};

            \node (q0p2)[petri-p,draw=black,below right of=move12,yshift=2em,label=30:{$\mathit{q^{\ptype_j}_\bot}$}]{};
            \node (q1p2)[petri-p,draw=black,below of=q0p2,label=-30:{$\mathit{q^{\ptype_j}_\top}$}]{};

            \path (q0p1) edge [->,thick,line width=1pt,out=-30,in=100] (move21);
            \path (q1p2) edge [->,thick,line width=1pt,out=-150,in=0] (move21);
            \path (move21) edge [->,thick,line width=1pt,out=80,in=-150] (q0p2);
            \path (move21) edge [->,thick,line width=1pt,out=180,in=-30] (q1p1);

            \path (q1p1) edge [->,thick,line width=1pt,out=30,in=-100] (move12);
            \path (q0p2) edge [->,thick,line width=1pt,out=150,in=0] (move12);
            \path (move12) edge [->,thick,line width=1pt,out=-80,in=150] (q1p2);
            \path (move12) edge [->,thick,line width=1pt,out=180,in=30] (q0p1);
        \end{tikzpicture}
      }
    \end{center}
    \vspace{-2em}
    \caption{Two process types (left), and two kinds of interactions (right),
    which all other process types and interactions in our restriction have the same shape as.}
    \label{fig:send-recv}
    \vspace*{-\baselineskip}   
\end{figure}


Pebble-passing systems have very strict constraints on their process
types and interactions, but no constraints on the set of network
topologies, other than that it is definable by an \hrtext{} grammar
written with constants of the form $(\send,\recv)_{\asrc_1,\asrc_2}$
or $(\recv,\send)_{\asrc_1,\asrc_2}$, for some source labels
$\asrc_1,\asrc_2 \in \sourcelabels$, where $\sourcelabels$ is the set
of source labels that may occur in a grammar. We denote by $\hr_\pps$
the signature of these grammars. The rest of this section is concerned
with the proof of the following theorem:

\begin{theorem}\label{thm:pebble-passing}
  The $\paramcover{}{}{}$ problem for grammars written using the
  $\hr_\pps$ signature is in \twoexptime\ and \pspace-hard.
\end{theorem}
The proof of \thmref{thm:pebble-passing} is organized as
follows. The double exponential upper bound relies on a result showing
that, in order to cover a target marking $\mtarget$ it is sufficient
to consider only firing sequences that cross (\ie move a pebble to and
from) each place at most $K$ times, where $K$ is the size of the unary
encoding of $\mtarget$. Based on this result (\lemref{lem:cover-soft-cap}),
we define a finite algebra $\algof{F}$
(\figref{fig:alg-zoo}), having the property that the coverability
problem reduces to a membership test on the language of the input
grammar in $\algof{F}$. The upper bound follows from the fact that
$\alangof{}{\algof{F}}{\grammar}$ is computable in double exponential
time (see \propref{prop:alg-flows} for a precise estimation). The
lower bound uses a polynomial reduction from the emptiness problem for
$2$-way nondeterministic automata~\cite{2nfa}.


\subsection{Firing Sequences}
\label{sec:pebble-fire-chara}

For the rest of this section, let $\open{\asys}= (\asys,\sources)$ be
a fixed open pebble-passing system, having an underlying system
$\asys=(\verts,\edges,\vlab)$ whose behavior is $\behof{\asys} \isdef
(\anet,\amark_0)$. Below, we introduce an equivalent characterization
of the firing sequences of $\behof{\asys}$.

The \emph{footprint} of a marking $\amark$ is a mapping $\fpof{\amark}
: \verts \to \set{0,1}$ defined as $\fpof{\amark}(v) \isdef
\amark(q^{\vlab(v)}_\top,v)$, for each vertex $v \in \verts$. Note
that $\fpof{\amark}$ evaluates to $1$ on pebble and to $0$ on hole
vertices. We say that a marking footprint $\pi$ is \emph{valid over
  $\mathcal{V}\subseteq\verts$} iff $0 \leq \pi(v) \leq 1$ for each
vertex $v \in \mathcal{V}$, \resp \emph{valid}, when
$\mathcal{V}=\verts$ follows from the context.

Given a subset of states $\mathcal{Q} \subseteq
\placeof{\ptypes}$ and a marking to cover $\mtarget: \mathcal{Q}
\rightarrow\nat$, the coverability problem asks for the existence of a
reachable marking $\amark:\placeof{\anet} \rightarrow\set{0,1}$ such
that, for each process type $\ptype\in\ptypes$ and each place $q \in
\mathcal{Q}$: 
\begin{align}
  \cardof{\set{v\in\vlab^{-1}(\ptype) \mid \fpof{\amark}(v) = 0}} \geq & \mtarget(q^\ptype_\bot) \label{eq1:cover} \\
  \cardof{\set{v\in\vlab^{-1}(\ptype) \mid \fpof{\amark}(v) = 1}} \geq & \mtarget(q^\ptype_\top) \label{eq2:cover}
\end{align}
The footprint $\fpof{v \to v'} : \verts \to \ints$ of an edge $v\to v'
\in \edges$ is defined as $\fpof{v \to v'}(u) \isdef 1$ if $u = v$,
$-1$ if $u = v'$ and $0$, otherwise.
We extend footprints to sequences of edges $\vec{e} \in \edges^*$ as
in $\fpof{\vec{e}} \isdef \sum_{e \in \vec{e}} \tau_e$. Because each
edge $v \to v' \in \edges$ corresponds to the transition that moves a
token from $(q^{\vlab(v)}_\top,v)$ to $(q^{\vlab(v)}_\bot,v)$ (\resp
from $(q^{\vlab(v')}_\bot,v')$ to $(q^{\vlab(v')}_\top,v')$) in the PN
$\behof{\asys}$, we shall abuse notation and write $\behof{v \to v'}$
for the transition corresponding to $v \to v'$ in $\behof{\asys}$ and
$\behof{\vec{e}}$ for the sequence of transitions corresponding to
$\vec{e} \in \edges^*$.


We remark that, for each firing sequence $\amark \fire{\abeh(\vec{e})}
\amark'$, we have $\fpof{\amark'} - \fpof{\amark} =
\fpof{\vec{e}}$. Intuitively, $\fpof{\vec{e}}$ is a witness of the fact
that the effect of firing the transitions $\abeh(\vec{e})$ is to move
pebbles from $\set{v \mid \fpof{\vec{e}}(v) = -1}$
to $\set{v \mid \fpof{\vec{e}}(v) = 1}$;
the vertices from $\set{v \mid \fpof{\vec{e}}(v) = 0}$ may store pebbles in between,
but are ultimately restored to their initial state.

We denote by $\tcount{e}{\vec{e}}$ the number of times the edge $e$
occurs in the sequence $\vec{e}$. We use the shorthands
$\fromcount{u}{\vec{e}} \isdef \sum_{u'\in\verts}\tcount{u \to
  u'}{\vec{e}}$ and $\tocount{u}{\vec{e}} \isdef
\sum_{u'\in\verts}\tcount{u' \to u}{\vec{e}}$. We define the following
partial orders between sequences of edges: $\vec{e'} \preceq \vec{e}
\iffdef \tcount{e}{\vec{e'}} \leq \tcount{e}{\vec{e}}$, for each $e
\in \edges$, and $\vec{e'} \sqsubseteq \vec{e} \iffdef \vec{e'}
\preceq \vec{e} \text{ and } \fpof{\vec{e'}} = \fpof{\vec{e}}$.
The following lemma characterizes the existence of fireable
sub-sequences:

\begin{lemmaE}[][category=proofs] 
  \label{lem:fireable-subsequence}
  For each marking $\amark$ and each sequence of edges
  $\vec{e}$, the following are equivalent: \begin{compactenum}[(i)]
  \item $\fpof{\amark} + \fpof{\vec{e}}$ is a valid marking footprint,
  \item there exists a sequence of edges $\vec{e'} \sqsubseteq
    \vec{e}$ such that $\abeh(\vec{e'})$ is fireable from $\amark$.
  \end{compactenum}
\end{lemmaE}
\begin{proofE}
  $(ii) \Rightarrow (i)$ is easy: if $\abeh(\vec{e'})$ is fireable from $\amark$,
  then there exists some $\amark \fire{\abeh(\vec{e'})} \amark'$.
  We write $\fpof{\amark'} - \fpof{\amark} = \fpof{\vec{e'}}$,
  and the knowledge that $\fpof{\vec{e'}} = \fpof{\vec{e}}$
  gives $\fpof{\amark} + \fpof{\vec{e}} = \fpof{\amark'}$
  which is a marking footprint because $\amark'$ is an automata-like marking.

  For $(i) \Rightarrow (ii)$, we proceed by induction on the length of $\vec{e}$.
  When $\vec{e}$ is empty, an obvious (and in fact the only) candidate for $\vec{e'}$
  is the empty sequence which has signature $\fpof{\epsilon} = 0$ everywhere,
  and is obviously fireable.

  In general when $\vec{e}$ is nonempty, we reinterpret $\vec{e}$ as a
  directed multigraph, where each edge $u \to v$ has multiplicity
  $\tcount{(u \to v)}{\vec{e}}$.  As does any nonempty directed
  multigraph, either it contains an elementary cycle, or it is acyclic and it
  contains a maximal nonempty path.
  \begin{itemize}
    \item In the case of a cycle $\vec{e} \supseteq \vec{c} = u_0 \to
      u_1 \to u_2 \cdots u_{k-1} \to u_0$, we make the observation
      that $\fpof{\vec{c}} = 0$ everywhere. Note that an elementary
      cycle is characterized by every vertex involved having incoming
      and outgoing degree exactly 1 each, so $\fpof{\vec{c}}(u_i) =
      \fpof{\to u_i}(u_i) + \fpof{u_i \to}(u_i) = 1 - 1 = 0$.  We
      deduce that $\vec{e} \setminus \vec{c} \sqsubseteq \vec{e}$,
      where $\vec{e} \setminus \vec{c}$ is naturally the sequence
      obtained by removing from $\vec{e}$ an occurrence of each
      transition from $\vec{c}$. Applying the inductive hypothesis to
      $\vec{e} \setminus \vec{c}$ which is strictly shorter than
      $\vec{e}$ gives $\vec{e'}$ fireable and $\vec{e'} \sqsubseteq
      \vec{e} \setminus \vec{c} \sqsubseteq \vec{e}$ where we conclude
      by transitivity of $\sqsubseteq$.
    \item Otherwise we have a directed acyclic graph,
      and a maximal path $\vec{e} \supseteq \vec{p} = u_0 \to u_1 \to u_2 \cdots u_{k-1} \to u_k$
      of that graph.
      First observe that a path from $u_0$ to $u_k$ has the same footprint
      as the single edge $u_0 \to u_k$ (evaluating to $-1$ on $u_0$, $1$ on $u_k$,
      and all intermediate vertices cancel out).
      Secondly we have assumed by $(i)$ that
      $0 \leq \fpof{\amark_0}(u_0) + \fpof{\vec{p}}(u_0)$ which means $\fpof{\amark_0}(u_0) = 1$,
      and $\fpof{\amark_0}(u_k) + \fpof{\vec{p}}(u_k) \leq 1$ which means $\fpof{\amark_0}(u_k) = 0$.
      The discrete quantity $\fpof{\amark_0}(u_i)$,
      going from 1 at $i=0$ to 0 at $i=k$,
      must for some value $i_0$ in between satisfy simultaneously
      $\fpof{\amark_0}(u_{i_0}) = 1$ and $\fpof{\amark_0}(u_{i_0+1}) = 0$.
      Thus $e_{i_0} = u_{i_0} \to u_{i_0+1}$ occurs in $\vec{e}$ and $\abeh(e_{i_0})$ is fireable.
      This yields $\amark_1$ a marking such that $\fpof{\amark_1} = \fpof{\amark} + \fpof{e_{i_0}}$,
      on which we apply the inductive hypothesis for $\vec{e} \setminus e_{i_0}$
      (at this point we need to show that $\fpof{\amark_1} + \fpof{\vec{e} \setminus e_{i_0}}$ is a marking footprint,
      this comes from $\fpof{\amark_1} + \fpof{\vec{e} \setminus e_{i_0}} = (\fpof{\amark} + \fpof{e_{i_0}}) + (\fpof{\vec{e}} - \fpof{e_{i_0}}) = \fpof{\amark} + \fpof{\vec{e}}$).
      Having thus obtained $\vec{e''} \sqsubseteq \vec{e} \setminus e_{i_0}$
      fireable from $\amark_1$,
      we construct $\vec{e'} \isdef e_{i_0}; \vec{e''} \sqsubseteq \vec{e}$
      which is fireable from $\amark$.
  \end{itemize}
  We thus conclude that $\fpof{\amark} + \fpof{\vec{e}}$ is a valid
  marking signature if and only if there exists some equivalent subsequence of $\vec{e}$
  that is actually fireable from $\amark$.
  \qed
\end{proofE}
\begin{proofSketch}
  By induction on the length of a firing sequence.
  After ensuring that $\vec{e}$ is acyclic, pick any maximal path and use combinatorial arguments
  to show that some transition along that path is fireable.
  \qed
\end{proofSketch}

Using the previous lemma, we prove that, in order to cover a given
marking $\mtarget$, it suffices to consider only those firing
sequences that cross each vertex a bounded number of times, where the
bound is the size of the unary encoding of $\mtarget$.
To that end, we define the \emph{degree} of a sequence $\vec{e} \in \edges^*$
as the maximum number of occurrences of one vertex in the sequence, \ie
$\deg(\vec{e}) \isdef \max\{\fromcount{u}{\vec{e}}, \tocount{u}{\vec{e}} \mid u\in\verts\}$.
From now on, the domain of $\mtarget$ will implicitly be the set $\mathcal{Q}
\subseteq \placeof{\ptypes}$. To simplify the following statement, we
say that $\amark : \placeof{\anet} \rightarrow \nat$ \emph{covers}
$\mtarget$ iff $\sum_{(q,v)\in\placeof{\anet}} \amark(q,v) \geq
\mtarget(q)$, for each $q \in \mathcal{Q}$ and that $\mtarget$ is
\emph{coverable} by $\asys$ iff there exists
$\amark\in\reach{\behof{\asys}}$ that covers $\mtarget$. Since, in a
pebble-passing system, markings can be equated to their footprints, we
say that $\fpof{\amark}$ covers $\mtarget$ whenever $\amark$ and
$\mtarget$ satisfy the conditions (\ref{eq1:cover}) and (\ref{eq2:cover}) above.


\begin{lemmaE}[][category=proofs] \label{lem:cover-soft-cap}
  A marking $\mtarget$ is coverable by $\asys$ iff there exists
  $\vec{e} \in \edges^{*\leq K}$ such that $\fpof{\amark_0} +
  \fpof{\vec{e}}$ covers $\mtarget$, where $K \isdef \sum_{q \in
    \mathcal{Q}} \mtarget(q)$, and $\edges^{*\leq K} \isdef
  \set{\vec{e} \in \edges^* \mid \deg(\vec{e}) \leq K}$.
\end{lemmaE}
\begin{proofE}
  The $(\Leftarrow)$ direction is easy, since the condition that
  $\fpof{\amark_0} + \fpof{\vec{e}}$ covers $\mtarget$
  alone implies that $\mtarget$ is coverable (through \lemref{lem:fireable-subsequence}
  and what it means for a footprint to cover a marking).

  We focus on the proof $(\Rightarrow)$, and assume that $\mtarget$ is coverable by $\asys$.
  In order to prove the existence of a sequence that satisfies the requirement,
  we proceed by as such: we assume that we have a minimal (in terms of length)
  firing sequence $\behof{\vec{e}}$ that covers $\mtarget$ in $\asys$,
  deduce independently that it must satisfy be such that
  $\fpof{\amark_0} + \fpof{\vec{e}}$ covers $\mtarget$
  and assume by contradiction that this firing sequence must
  have degree at least $K+1$, in other words it must
  fire at least $K+1$ times an incoming or outgoing
  (assume without loss of generality that it is incoming, the other proof is identical)
  transition for some vertex $u$.
  We then construct a strictly shorter sequence that also covers
  the marking, thereby contradicting the hypothesis of minimality.

  We first ensure that $\vec{e}$ is acyclic:
  a cycle $\vec{c}$ is characterized by $\fpof{\vec{c}} = 0$ everywhere,
  since in a cycle every vertex involved has as many incoming as outgoing edges occurring.
  Thus $\fpof{\vec{e} \setminus \vec{c}} = \fpof{\vec{e}} - \fpof{\vec{c}} = \fpof{\vec{e}}$:
  if $\fpof{\amark_0} + \fpof{\vec{e}}$ is a marking signature,
  then so is $\fpof{\amark_0} + \fpof{\vec{e} \setminus \vec{c}}$.
  Necessarily if $\behof{\vec{e}}$ is fireable and $\vec{e}$ contains a cycle $\vec{c}$,
  then there is also an equivalent subsequence of $\vec{e} \setminus \vec{c}$
  that is fireable. This would contradict the minimality of $\vec{e}$.

  Now that $\vec{e}$ is acyclic, if there still exists a vertex $u$
  with in-degree at least $K+1$,
  denote by $i_1 < \cdots < i_{K+1}$ a set of indices of $K+1$ occurrences
  in $\vec{e}$ of an edge leading to $u$.
  In that situation, \figref{algo:path-decomposition} will compute a partial assignment
  $p$ of edges to disjoint paths from 1 to $K+1$, such that
  all of these paths go through $u$

  \begin{figure}[h!]
    {\small\begin{algorithmic}[0]
      \STATE \textbf{input}: $\vec{e} = (w_0 \to v_0); (w_1 \to v_1); ...$ a fireable sequence of edges (\ie $\behof{\vec{e}}$ is a firing sequence)
      \STATE \textbf{input}: $u$ a vertex that occurs in $\vec{e}$
      \STATE \textbf{output}: $p : [0, |\vec{e}|[ \to \nat^?$ a decomposition of $\vec{e}$ in paths with distinct endpoints
    \end{algorithmic}
    \begin{algorithmic}[1]
      \STATE $n \gets 0$
      \STATE $p \gets (\_ \mapsto \bot)$
      \COMMENT{how to interpret $p$: if $p(i) = n$ it means that $\vec{e}_i$ is assigned to the $n$'th path}
      \FOR{$|\vec{e}| > i \geq 0$ \DECR}
        \STATE \textbf{invariant}: every $p^{-1}(n')$ for $n' \leq n$ is a path
        \IF{$v_i = u$} \COMMENT{This edge is incoming for $u$, we want one path that goes through it}
          \STATE $p(i) \gets n$
          \STATE $u^- \gets w_i$
          \STATE $u^+ \gets v_i$
          \FOR{$i < j < |\vec{e}|$ \INCR}
            \STATE \textbf{invariant}: $p^{-1}(n)$ is a path from $u^-$ to $u^+$
            \IF{$p(j) = \bot$} \COMMENT{$\vec{e}_i$ is not yet part of a path}
              \IF{$w_j = u^+$} \COMMENT{$u^+$ is the end of the path $p^{-1}(n)$, so $u^+ \to v_j$ can extend it}
                \STATE $p(j) \gets n$
                \STATE $u^+ \gets v_j$ \COMMENT{And of course $v_j$ is the new end of $p^{-1}(n)$}
              \ENDIF
              \IF{$v_j = u^-$} \COMMENT{$u^-$ is the start of the path $p^{-1}(n)$ so $w_j \to u^-$ can extend it}
                \STATE $p(j) \gets n$
                \STATE $u^- \gets w_j$ \COMMENT{And of conrse $w_j$ is the new start of $p^{-1}(n)$}
              \ENDIF
            \ENDIF
          \ENDFOR
          \STATE $n \gets n + 1$ \COMMENT{No other edges can be added to this path, so we handle the next}
        \ENDIF
      \ENDFOR
    \end{algorithmic}}
    \caption{Decomposing a sequence into paths.
      The fact that these paths have distinct endpoints is enforced by the specific shape that they have,
      as interleavings of one positive and one negative path that each independently appear in that order in $\vec{e}$.
    }
    \label{algo:path-decomposition}
  \end{figure}

  We thus have $K+1$ paths.
  Though these paths are not completely ordered,
  they do have a specific shape that guarantees that their endpoints are distinct.
  We call \emph{positive path} an ordered path $(u_0 \to u_1); (u_1 \to u_2); (u_2 \to u_3); \ldots$
  that occurs in that order as a subsequence of $\vec{e}$.
  We call \emph{negative path} an ordered path $(u_1 \to u_0); (u_2 \to u_1); (u_3 \to u_2); \ldots$
  that occurs in that order as a subsequence of $\vec{e}$.
  The paths we are interested in are interleavings of one positive path starting from $u$,
  and one negative path ending in $u$.
  For example, $(w_1 \to u); (u \to v_1); (w_2 \to w_1); (w_3 \to w_2); (v_1 \to v_2); (v_2 \to v_3); (w_4 \to w_3); (v_3 \to v_4); (v_4 \to v_5); (w_5 \to w_4)$
  satisfies this criterion as it is an interleaving
  of the positive path
  $(u \to v_1); (v_1 \to v_2); (v_2 \to v_3); (v_3 \to v_4); (v_4 \to v_5)$
  starting from $u$
  and the negative path
  $(w_1 \to u); (w_2 \to w_1); (w_3 \to w_2); (w_4 \to w_3); (w_5 \to w_4)$
  ending in $u$.

  The reason these kinds of paths are relevant is that in the specific context where they occur
  as an ordered subsequence of a firing sequence of a pebble-passing system,
  the end of a maximal positive path has a pebble in the final marking,
  and the start of a maximal negative path has a hole in the final marking.
  Since it is impossible to move a pebble (\resp hole) to a vertex that already contains one,
  two maximal positive paths that use disjoint sets of edges cannot have the same end,
  and two maximal negative paths that use disjoint sets of edges cannot have the same start.
  By being interleavings of each a maximal positive path and a maximal negative path,
  the $K + 1$ paths we have just created must have pairwise distinct starts and ends.

  Put into equations this means
  each path $\vec{p}_j$ from $u_0^j$ to $u_{k_j}^j$ contains the edge $e_{i_j}$,
  whose endpoints satisfy $u_0^j \ne u_0^{j'}$ and $u_{k_j}^j \ne u_{k_{j'}}^{j'}$ whenever $j \ne j'$,
  and because the paths are maximal we have $\fpof{\amark'}(u_0^j) = 0$
  and $\fpof{\amark'}(u_{k_j}^j) = 1$ for every $j$.
  This last point implies that each sequence $\vec{e} \setminus \vec{p_j}$ is fireable from $\amark_0$:
  $\fpof{\vec{p}_j}$ evaluates to $-1$ on $u_0^j$, $1$ on $u_{k_j}^j$, and 0 everywhere else,
  and thus $0 \leq \fpof{\amark'} - \fpof{\vec{p}_j} \leq 1$.

  A coverability query, as we recall from Equations \eqref{eq1:cover} and \eqref{eq2:cover},
  is simply a requirement of cardinality for sets of the form $\{v \mid \fpof{\amark'}(v) = 1\}$
  or $\{v \mid \fpof{\amark'}(v) = 0\}$ for specific process types.
  A marking $\mtarget$ can be covered by looking at $K$ such vertices.
  The previous construction has given us $K+1$ pairs of one vertex of each of
  these two sets, so there is at least one path where both endpoints are
  irrelevant to the marking,
  in other words one path $p_{j_0}$ from $u_0$ to $u_k$
  where in fact the inequalities
  $\cardof{\set{v \kerof{\vlab} u_0 \mid (\fpof{\amark_0} + \fpof{\vec{e}})(v) = 0}} > \mtarget(q^{\vlab(u_0)}_\bot)$
  and
  $\cardof{\set{v \kerof{\vlab} u_k \mid (\pi_{\amark_0} + \fpof{\vec{e}})(v) = 1}} > \mtarget(q^{\vlab(u_k)}_\top)$
  are strict.
  It follows that a final marking with one fewer pebble in $\vlab(u_k)$ and one fewer hole in $\vlab(u_0)$
  would still cover $\mtarget$.
  One such marking is given by the application of \lemref{lem:fireable-subsequence}
  to $\vec{e} \setminus \vec{p}_{j_0}$ as mentioned above.
  We have thus reached a different marking than $\amark'$ but that nevertheless
  still covers $\mtarget$, and done so in strictly fewer steps.
  This is a contradiction.

  Thus the proof by contradiction ends and we deduce that there must exist
  a sequence that in addition to covering $\mtarget$ also has degree at most $K$.
\end{proofE}

\subsection{Flows}

To check coverability, we consider a finite algebra $\algof{F}$ whose
elements represent the sequences of edges that cross each vertex at
most $K$ times. The domain of $\algof{F}$ is $\universeOf{F} \subseteq
\pow{[0,K]^\sourcelabels \times [0,K]^\sourcelabels \times
  [0,K]^\mathcal{Q}}$. The elements $(f^+,f^-,n) \in \universeOf{F}$
represent sequences of edges, such that $f^+(\asrc)$ (\resp
$f^-(\asrc)$) is the in-degree (\resp out-degree) of the
$\asrc$-source of $\asys$ and $n(q)$ is the number of tokens that end
in $q \in \mathcal{Q}$. Intuitively, a tuple $(f^+,f^-,n)$ witnesses
the existence of a sequence that covers $n$, that can later be
combined with other sequences which compensate its surplus $f^+$ and
deficit $f^-$ on the sources of $\asys$.

Formally, we define the following mappings, where
$\universeOf{S}$ denotes the set of open systems, \ie systems with
sources:
\begin{align*}
  \omega : & ~\edges^{*\leq K} \times
  \verts^\sourcelabels \rightarrow \interv{0}{K}^\sourcelabels \times
  \interv{0}{K}^\sourcelabels \times \interv{0}{K}^\mathcal{Q} \\
  \omega(\vec{e},\sources) = & (f^+,f^-,n) \iffdef
  \left\{\begin{array}{l} f^+(\asrc) =
  \fromcount{\sources(\asrc)}{\vec{e}} \hspace*{8mm} f^-(\asrc) =
  \tocount{\sources(\asrc)}{\vec{e}} \\
  n(q^\ptype_\bot) = \min(\mtarget(q^\ptype_\bot), \cardof{\set{v\in\vlab^{-1}(\ptype)
      \setminus \img{\sources} \mid
      (\fpof{\initmarkof{\ptype}}+\fpof{\vec{e}})(v) = 0}}) \\
  n(q^\ptype_\top) = \min(\mtarget(q^\ptype_\top),
  \cardof{\set{v\in\vlab^{-1}(\ptype) \setminus \img{\sources} \mid
      (\fpof{\initmarkof{\ptype}}+\fpof{\vec{e}})(v) = 1}})
  \end{array}\right.
  \\[1mm]
  \eta : & ~\universeOf{S} \rightarrow \pow{\interv{0}{K}^\sourcelabels \times \interv{0}{K}^\sourcelabels \times \interv{0}{K}^\mathcal{Q}} \\
  \eta(\asys,\sources) \isdef & \set{\omega(\vec{e},\sources) \mid \vec{e} \in \edgeof{\asys}^{*\leq K},~ \fpof{\initmarkof{\behof{\asys}}}+\fpof{\vec{e}} \text{ is a valid marking footprint over } \vertof{\asys}\setminus\img{\sources}}
\end{align*}


We define the finite algebra of \emph{flows} $\algof{F}$ using
\propref{prop:cong-homo}, where $\eta$ is taken to be the
homomorphism between $\algof{S}$ and $\algof{F}$:

\begin{lemmaE}[][category=proofs]\label{lemma:cong-sys-flow}
  $\kerof{\eta}$ is a \hrtext{} congruence. 
\end{lemmaE}
\begin{proofE} Excluding the trivial case of an edge,
  we perform a case analysis.
  \begin{itemize}
    \item \underline{$\rename{\alpha}{\algof{S}}$}: we assume $(\asys,\sources)
      \kerof{\eta} (\asys',\sources')$.  We pick any
      $\omega(\vec{e},\sources\circ\alpha^{-1})$ in
      $\eta(\rename{\alpha}{\algof{S}}(\asys,\sources))$, and prove
      that it also occurs in
      $\eta(\rename{\alpha}{\algof{S}}(\asys',\sources'))$.  First we
      notice that $\rename{\alpha}{\algof{S}}(\asys,\sources)$ has the
      same edges as $(\asys,\sources)$, and thus the same set of
      sequences of transitions.  This means that
      $\omega(\vec{e},\sources) \in \eta(\asys,\sources)$.  The
      initial equivalence gives the existence of a matching
      $\omega(\vec{e},\sources) = \omega(\vec{e'},\sources') \in
      \eta(\asys',\sources')$, and thus a candidate is
      $\omega(\vec{e'},\sources'\circ\alpha^{-1})$.  Since
      $\img{\sources}$ and $\img{\sources'}$ are unchanged by a
      composition on the right by $\alpha^{-1}$, we maintain
      $\omega(\vec{e'},\sources'\circ\alpha^{-1}) =
      \omega(\vec{e},\sources\circ\alpha^{-1})$ and thus
      $\eta(\rename{\alpha}{\algof{S}}(\asys,\sources)) \subseteq
      \eta(\rename{\alpha}{\algof{S}}(\asys',\sources'))$.  The
      symmetry of the definition makes this actually an equality.
    \item \underline{$\restrict{\tau}{\algof{S}}$}:
      we assume $(\asys,\sources) \kerof{\eta} (\asys',\sources')$.
      Once again $\restrict{\tau}{\algof{S}}(\asys,\sources)$ has the same edges as $(\asys,\sources)$,
      so this time the difficulty comes from the fact that $\img{\sources}$ shrinks,
      not from the set of sequences to consider.
      The key observation to make is that we always have
      $\fpof{\vec{e}}(v) = \tocount{v}{\vec{e}} - \fromcount{v}{\vec{e}}$,
      and this holds in particular for $v = \sources(\asrc)$,
      where additionally $\fpof{\vec{e}}(\sources(\asrc)) = f^+(\asrc) - f^-(\asrc)$.
      What this means is that when the equivalence $\kerof{\eta}$
      imposes $f^+ = f'^+$ and $f^- = f'^-$, it also guarantees that
      $\fpof{\initmarkof{\behof{\asys}}} + \fpof{\vec{e}}$
      is a valid marking footprint on $\sources(\asrc)$ if and only if
      $\fpof{\initmarkof{\behof{\asys'}}} + \fpof{\vec{e'}}$
      is a valid marking footprint on $\sources'(\asrc)$.
      The same reasoning goes for the more specific criterion that
      $(\fpof{\initmarkof{\ptype}} + \fpof{\vec{e}})(\sources(v)) = 0 \text{ or } 1$,
      from which we get that after restriction the two tuples are still in accordance
      and thus $\eta(\restrict{\tau}{\algof{S}}(\asys,\sources))
        = \eta(\restrict{\tau}{\algof{S}}(\asys',\sources'))$.
    \item \underline{$\pop{\algof{S}}$}: take $(\asys_1,\sources_1)
      \kerof{\eta} (\asys'_1,\sources'_1)$ and $(\asys_2,\sources_2)
      \kerof{\eta} (\asys'_2,\sources'_2)$.  We write
      $(\asys,\sources) = (\asys_1,\sources_1) \pop{\algof{S}}
      (\asys_2,\sources_2)$, and similarly for $(\asys',\sources')$.
      We consider a sequence $\vec{e} \in \edgeof{\asys}^{*\leq K}$
      and aim to prove that one can find $\vec{e'} \in
      \edgeof{\asys'}^{*\leq K}$ that results in the same tuple
      through $\omega$. By definition of $\pop{\algof{S}}$ we have
      $\edgeof{\asys} = \edgeof{\asys_1} \uplus \edgeof{\asys_2}$.  We
      can thus partition $\vec{e}$ into two subsequences: $\vec{e}_i$
      holds all edges from $\asys_i$ ($i = 1,2$).  To these two
      subsequences naturally correspond through the equivalence
      $\kerof{\eta}$ two sequences of $\edgeof{\asys'_1}$ and
      $\edgeof{\asys'_2}$ respectively, which we write $\vec{e}'_1$
      and $\vec{e}'_2$.  Since $\vec{e}' \isdef \vec{e}'_1 ;
      \vec{e}'_2$ is a good candidate for our solution, we take a
      moment to write the relationships between the various tuples
      involved:
      \begin{itemize}
        \item $f^+ = f^+_1 + f^+_2$, $f'^+ = f'^+_1 + f'^+_2$,
          $f^- = f^-_1 + f^-_2$, and $f'^- = f'^-_1 + f'^-_2$:
          these hold because they are defined from elementary
          $\tcount{e}{\vec{e}}$ which due to the partition of edges
          satisfy additive properties;
        \item $n(q) = n_1(q) + n_2(q)$, and $n'(q) = n'_1(q) + n'_2(q)$:
          these hold because during a composition the sets of vertices that
          are not sources are disjoint in the result;
        \item and naturally $\fpof{\vec{e}} = \fpof{\vec{e}_1} + \fpof{\vec{e}_2}$,
          and $\fpof{\vec{e}'} = \fpof{\vec{e}'_1} + \fpof{\vec{e}'_2}$
          because once again the sets of edges are disjoint each time.
      \end{itemize}
      In the output of a composition, the set $\vertof{\asys}\setminus\img{\sources}$
      of vertices that are not sources is a disjoint union of the vertices that were not sources
      from the two graphs that were composed, thus the condition
      ``is a valid marking footprint over $\vertof{\asys}\setminus\img{\sources}$''
      is easily preserved.
      Thus $\eta(\asys,\sources) = \eta(\asys',\sources')$.
  \end{itemize}
  We conclude that $\kerof{\eta}$ is a \hrtext{} congruence.
  \qed
\end{proofE}

The remainder of the proof for the upper bound from
\thmref{thm:pebble-passing} relies on the fact that the interpretation of
the \hrtext{} signature in $\algof{F}$ is effectively computable. The
size of the grammar $\grammar$ is the total number of occurrences of a
nonterminal or function symbol in a rule from $\grammar$, denoted as
$\sizeof{\grammar}$.

\newcommand{\closedstp}[3]{\mathsf{closed}^{#1}_{#2}({#3})}

\begin{propositionE}[][category=proofs]\label{prop:alg-flows}
  The size of each element $f \in \universeOf{F}$ is
  $2^{\bigO((\cardof{\sourcelabels}+\cardof{\ptypes}) \cdot \log K)}$
  and the function $\aop^\algof{F}(f_1,\ldots,f_n)$ can be computed in
  time $2^{\bigO((\cardof{\sourcelabels}+\cardof{\ptypes}) \cdot \log
    K)}$, for each \hrtext-function symbol $\aop$ of arity $n\geq0$
  and all elements $f_1,\ldots,f_n \in \universeOf{F}$. Moreover, for
  each grammar $\grammar$ using source labels from $\sourcelabels$,
  the language $\alangof{}{\algof{F}}{\grammar}$ is computable in time
  $2^{\sizeof{\grammar} \cdot
    2^{\bigO((\cardof{\sourcelabels}+\cardof{\ptypes}) \cdot \log
      K)}}$.
\end{propositionE}
\begin{proofE}
  We denote by $\closedstp{\sources}{\tau}{\ptype}$ the set of source
  labels such that $\asrc \in \closedstp{\sources}{\tau}{\ptype}$ iff
  $\asrc\in\dom{\sources}$ and $\asrc\not\in\tau$ and the
  $\ptypeof{\asrc} = \ptype$.  We define $\fpof{\initmark}(\ptype)
  \isdef \amark_0(q^\ptype_\top)$. The function $\delta_x$ outputs 1
  on $x$ and $0$ everywhere else (Kronecker delta). For a function $f$
  with domain $\sourcelabels$, we denote by $\proj{f}{\tau}$ the
  restriction of $f$ to the source labels in $\tau \subseteq
  \sourcelabels$. The inference rules below define the operations of
  $\algof{F}$: 
  \begin{prooftree}
    \AxiomC{$
      0 \leq k \leq K
    $}
    \RightLabel{Edge}
    \UnaryInfC{$
      (k \cdot \delta_{\sigma_2}, k \cdot \delta_{\sigma_1}, 0) \in \sgraph{(\send,\recv)}{\asrc_1}{\asrc_2}{\algof{F}}
    $}
  \end{prooftree}
  \begin{prooftree}
    \AxiomC{$
      (f^+,f^-,n) \in F
    $}
    \RightLabel{Rename}
    \UnaryInfC{$
      (f^+\circ\alpha^{-1}, f^-\circ\alpha^{-1}, n)
      \in \rename{\alpha}{\algof{F}} (F)
    $}
  \end{prooftree}
  \begin{prooftree}
    \AxiomC{$
      \begin{array}{c}
        (f^+_1,f^-_1,n_1) \in F_1 \qquad
        (f^+_2,f^-_2,n_2) \in F_2 \\
        f^+_1 + f^+_2 \leq K \qquad
        f^-_1 + f^-_2 \leq K
      \end{array}
    $}
    \RightLabel{Compose}
    \UnaryInfC{$
      (f^+_1 + f^+_2, f^-_1 + f^-_2, \min(\mtarget, n_1 + n_2)) \in F_1 \pop{\algof{F}} F_2
    $}
  \end{prooftree}
  \begin{prooftree}
    \AxiomC{$
      \begin{array}{c}
        (f^+, f^-, n) \in F \\
        \forall\sigma\in\closedstp{\sources}{\tau}{\ptype}.\ f^+(\asrc) - f^-(\asrc) + \fpof{\initmark(\ptypeof{\asrc})} \in \set{0,1} \\
        \forall\ptype.\ d_n(q^\ptype_\bot) \isdef \cardof{\set{\asrc\in\closedstp{\sources}{\tau}{\ptype} \mid f^+(\asrc) - f^-(\asrc) + \pi_{init}(\ptype) = 0}} \\
        \forall\ptype.\ d_n(q^\ptype_\top) \isdef \cardof{\set{\asrc\in\closedstp{\sources}{\tau}{\ptype} \mid f^+(\asrc) - f^-(\asrc) + \pi_{init}(\ptype) = 1}} \\
      \end{array}
    $}
    \RightLabel{Restrict}
    \UnaryInfC{$
      (\proj{f^+}{\tau}, \proj{f^-}{\tau}, \min(\mtarget, n + d_n))
      \in \restrict{\tau}{\algof{F}}(F)
    $}
  \end{prooftree}
  First we prove that these inference rules are compatible with $\eta$
  defined earlier, \ie for any ground \hrtext{} term $\theta$, we have
  $\eta(\theta^\algof{S}) = \theta^\algof{F}$.
  \begin{itemize}
    \item \underline{$(\send,\recv)^\algof{F}_{\asrc_1,\asrc_2}$}:
      the firing sequences of $(\send,\recv)^\algof{S}_{\asrc_1,\asrc_2}$
      are a single edge repeated arbitrarily many times.
      Of these the firing sequences that use each vertex at most $K$ times
      are exactly the $\vec{e}_k$ of footprint
      $\fpof{\vec{e}_k}$ that equals $-k$ on $\sources(\asrc_1)$,
      $+k$ on $\sources(\asrc_2)$, and $0$ everywhere else.
      The set $\verts\setminus\img{\sources}$ is empty and thus $n$ is 0 everywhere.
      This generates exactly the set that the rule Edge produces.
    \item \underline{$\rename{\alpha}{\algof{F}}$}: by applying the
      same renaming to $f^+$ and $f^-$ as we do to the graph, we
      obviously preserve the relationship between the two.
    \item \underline{$\pop{\algof{F}}$}:
      without the constraints of finiteness,
      we would simply take $f^+_1 + f^+_2$, $f^-_1 + f^-_2$, and $n_1 + n_2$.
      The additional constraints are translated exactly from their
      equivalent in the definition of $\eta$:
      we require $f^+_1 + f^+_2 \leq K$ and $f^-_1 + f^-_2 \leq K$
      to conform to the rule that the sequence we consider must belong
      to $\edges^{*\leq K}$ (recall that $f^+$ and $f^-$ are incoming
      and outgoing degrees, so if one of them exceeds $K$ then the overall
      sequence has degree greater than $K$),
      and apply $\min(\mtarget, ...)$ so that $n$ remains bounded by $\mtarget$
      just like how it is defined in $\omega$.
      Since we apply exactly the same constraints to $(f^+,f^-,n)$ here
      as we did in the original definition of $\omega$,
      we obtain the same final set of tuples.
    \item \underline{$\restrict{\tau}{\algof{F}}$}: the key property
      is that $\closedstp{\sources}{\tau}{\ptype} =
      \text{ptype}^{-1}(\ptype) \cap (\dom{\sources} \setminus
      (\dom{\proj{\sources}{\tau}}))$.  This is relevant because it
      implies $(\vlab^{-1}(\ptype) \setminus
      \img{\proj{\sources}{\tau}}) = (\vlab^{-1}(\ptype) \setminus
      \img{\sources}) \uplus
      \sources(\closedstp{\sources}{\tau}{\ptype})$, in which the
      first two terms have the same shape as one that occurs in the
      definition of $n$ for $\omega$, and the third gives the
      definition of $d_n$ in rule Restrict above.  This means that
      given $n$ and $n'$ from tuples in $\eta(\asys,\sources)$ and
      $\eta(\restrict{\tau}{\algof{S}}(\asys,\sources))$, we indeed
      have $n' = n + d_n$.  The constraint $0 \leq f^+(\asrc) -
      f^-(\asrc) + \fpof{\initmarkof{\ptypeof{\asrc}}} \leq 1$ from
      the premiss of the Restrict rule enforces the requirement
      ``$\fpof{\initmarkof{\behof{\asys}}} + \fpof{\vec{e}}$ is a
      valid marking footprint'' for all vertices that were added to
      $\vertof{\asys}\setminus\img{\sources}$ by the restriction.  The
      rest of the definition follows the same principles as for the
      other cases: $f^+(\asrc) - f^-(\asrc)$ has already been
      justified to be a synonym for $\fpof{\vec{e}}(\sources(\asrc))$,
      $f^+$ and $f^-$ must be constrained to their new domain, and $n
      + d_n$ must not exceed $\mtarget$.
  \end{itemize}
  Therefore $\algof{F}$ defined explicitly here is compatible
  with its implicit definition in \lemref{lemma:cong-sys-flow}.

  Since each tuple in an element of $\universeOf{F}$ is of size
  \[\cardof{[0,K]^\sourcelabels \times [0,K]^\sourcelabels \times
    [0,K]^\mathcal{Q}} = (K+1)^{2\cardof{\sourcelabels} +
    2\cardof{\ptypes}}\] each element of $\universeOf{F}$ is of size 
  $2^{\bigO((\cardof{\sourcelabels}+\cardof{\ptypes}) \cdot \log K)}$.
  From the definition of the inference rules, Edge takes
  constant time to apply, Rename and Restrict take linear time, and
  Compose takes quadratic time in the size of their arguments. Then,
  evaluating the interpretation of any \hrtext{} function symbol takes
  $2^{\bigO((\cardof{\sourcelabels}+\cardof{\ptypes}) \cdot \log K)}$
  time. Because the domain of the algebra $\algof{F}$ is
  finite, the language $\alangof{}{\algof{F}}{\grammar}$ can be
  computed by a finite iteration of the Kleene sequence
  $\overrightarrow{\emptyset},
  \sem{\grammar}(\overrightarrow{\emptyset}),\sem{\grammar}(\overrightarrow{\emptyset})
  \cup \sem{\grammar}^2(\overrightarrow{\emptyset}), \ldots$, where
  $\sem{\grammar}$ is the function that maps any valuation of the
  nonterminals in $\grammar$ to the sets obtained by applying the
  rules of $\grammar$ and $\overrightarrow{\emptyset}$ assigns the
  empty set to each such nonterminal. Since each element of
  $\algof{F}$ is of size
  $2^{\bigO((\cardof{\sourcelabels}+\cardof{\ptypes}) \cdot \log K)}$,
  there are at most
  $2^{2^{\bigO((\cardof{\sourcelabels}+\cardof{\ptypes}) \cdot \log
      K)}}$ such elements, hence the Kleene iteration takes at most as
  many steps to reach a fixpoint. Moreover, computing each step of the
  Kleene iteration takes $\sizeof{\grammar} \cdot
  2^{\bigO((\cardof{\sourcelabels}+\cardof{\ptypes}) \cdot \log K)}$
  time, hence the entire language can be computed in time
  $2^{\sizeof{\grammar} \cdot
    2^{\bigO((\cardof{\sourcelabels}+\cardof{\ptypes}) \cdot \log
      K)}}$.\qed
\end{proofE}

\noindent Deciding whether $\mtarget$ is coverable by some instance
$\asys \in \alangof{}{\algof{S}}{\grammar}$, for a given \hrtext{}
grammar $\grammar$, is done by checking
$\restrict{\emptyset}{\algof{F}}(\alangof{}{\algof{F}}{\grammar}) \cap
\set{(0,0,\mtarget)} \stackrel{?}{=} \emptyset$. We apply
$\restrict{\emptyset}{}$ to $\alangof{}{\algof{F}}{\grammar}$ to
ensure that the sequences of edges considered lead to valid marking
footprints on every vertex of a system $(\asys,\sources) \in
\alangof{}{\algof{S}}{\grammar}$, including the sources from
$\img{\sources}$, that were exempt from satisfying this condition (see
the above definition of $\eta$). The latter emptiness check applies
the Filtering Theorem (\thmref{thm:filtering}), leading to an overall
\twoexptime\ upper bound.

The \pspace\ lower bound is obtained by a polynomial reduction from
the \pspace-complete emptiness problem for 2-way nondeterministic
finite automata (\twonfa). The idea of the reduction is to simulate a
run of a \twonfa{} by a grid-like system, such that the horizontal
axis corresponds to the length of the word and the vertical axis to
the number of control states of the automaton. A run of the \twonfa{}
is modeled by an execution of the system that moves a pebble
left/right to the next control state given by the transition relation
of the automaton. \ifLongVersion\else
Formal details can be found in \appref{app:lower-bound}. 
\fi

\begin{textAtEnd}[category=hardness]

In order to demonstrate that our restriction has not made the problem
trivial, we show that it remains \textsf{PSPACE}-hard to decide
coverability.  This is done by reduction from the known
\textsf{PSPACE}-complete problem of deciding the emptiness of the
language of a 2NFA (2-way Nondeterministic Finite Automaton,
\cite{2nfa}), which is essentially a read-only Turing Machine.

\begin{definition}{\textbf{2NFA.}}
  A 2NFA is a tuple $\automata = (Q, A, \delta, q_0, q_f)$,
  where $Q$ denotes the set of states (of which $q_0$ is the initial state
  and $q_f$ is the accepting state)
  and $A$ denotes the alphabet.
  Compared to a regular automaton, the transition function
  $\delta : Q \times (A \uplus \set{\langle,\rangle}) \to \pow{Q \times \set{\leftarrow, \rightarrow}}$
  also dedermines the direction in which the next letter is read
  ($\langle \cdots \rangle$ mark the beginning and end of words).
\end{definition}

\begin{lemma}
  The emptiness problem for a 2NFA reduces to coverability in pebble-passing systems.
  Knowing that emptiness is \textsf{PSPACE}-complete, this implies that
  coverability is \textsf{PSPACE}-complete.
\end{lemma}
\begin{proof}
  We need three process types: $\ptypes = \set{\ptype_i, \ptype_f, \ptype}$,
  representing respectively the initial, final, and any other state.
  We pick the set of sources $\sourcelabels \isdef Q \times \set{\leftarrow, \rightarrow} \cup {\mathsf{init}}$.
  The grammar $\grammar_\automata$ constructs systems that simulate the execution
  of $\automata$ on an arbitrary word of $\langle A^* \rangle$.

  Given a 2NFA $\automata$, construct the following grammar $\grammar_\automata$:
  \input{figure-2nfa/gram.tex}
  on which we ask the coverability query $\mtarget : \left\{\begin{array}{ll}q^{\ptype_f}_\top & \mapsto 1 \\ \_ & \mapsto 0\end{array}\right.$.
  When the grammar unfolds successively the nonterminals $X_{a_1}, T_{a_1 a_2}, X_{a_2}, T_{a_2 a_3}, X_{a_3}, ...$
  it simulates the behavior of $\automata$ on the word $a_1 a_2 a_3 \cdots$,
  in that if $\automata$ is reading the $i$'th letter while in state $q$
  then the token is currently on place $q_1$ of the process
  which at depth $i$ of unfolding the grammar was labeled by the source $(q, \rightarrow)$.
  Applied to $q_f$ this means that the place $q^{\ptype_f}_1$ can have a token
  exactly when the automaton reaches state $q_f$.
  The grammar ensures that any word in $A^*$ can be generated in this manner,
  and thus that $\mtarget$ is coverable if and only if the language is nonempty.
\end{proof}

\begin{example}
  In \figref{fig:2nfa-example} we show how this construction turns an example automata
  into a pebble-passing system that simulates its execution on a given word.
\end{example}

\input{figure-2nfa/rendered.tex}
\end{textAtEnd}

\newcommand{\toolrepo}{\mycomment{Neven}{REPO}}
\newcommand{\toolname}{ParCoSys}
\newcommand{\toolnameexplanation}{(\textbf{Par}ameterized \textbf{Co}verability)}

\section{Experiments}

We have implemented the counting abstraction method in the prototype
tool \toolname{}
\footnote{The source code is available as anonymized additional
  material.} \toolnameexplanation{}. The input of the tool is a
grammar describing the system and a safety property to be checked. The
output is a finite set of Petri nets that is fed to the LoLA
analyzer~\cite{lola}. Our choice for LoLA was driven by its robustness
and performance, but any Petri net analyzer can be used as back-end,
in principle.

The \texttt{ring} example is a standard token ring, on which we verify
that the process holding the token is unique.  The \texttt{philos}
example, is the dining philosophers problem on a ring, for which we
prove mutual exclusion properties between neighbors.  The
\texttt{consensus} example features processes arranged as a star performing a
2-valued consensus. The \texttt{leader-election} example is a
finite-valued model of a leader election having a pre-determined
winner.


In addition to these standard examples, we considered several mutual
exclusion protocols on stars (\texttt{lock}, \texttt{star}), binary
trees (\texttt{tree-dfs}, \texttt{tree-down}, \texttt{tree-halves}),
and more complex architectures, \eg \texttt{star-ring} is a star whose
points are chained, \texttt{server-loop} is a ring whith a star on
each node and \texttt{tree-nav} is a tree with chained leaves. It is
worth mentioning that \texttt{ring}, \texttt{star},
\texttt{tree-down}, \texttt{tree-nav} happen to be pebble-passing
systems (\secref{sec:pebble-passing}).


\begin{figure}[t!]
  \vspace*{-2\baselineskip}
  \begin{center}
    {\scriptsize\begin{tabular}{|c|c|c|c|c|}
      \hline
      Name & Architecture & Result & Size & Runtime (ms) \\
      \hline
      
        \texttt{ ring }
          & ring
          & $2 / 2$
          & $(\bigcirc9, \square11)\times 16$
          & $(94 + 17) \pm 5$
          \\
      
        \texttt{ star }
          & star
          & $3 / 3$
          & $(\bigcirc12, \square11)\times 4$
          & $(63 + 10) \pm 4$
          \\
      
        \texttt{ lock }
          & star
          & $1 / 3$
          & $(\bigcirc9, \square11)\times 4$
          & $(65 + 7) \pm 4$
          \\
      
        \texttt{ star-ring }
          & chained star
          & $3 / 3$
          & $(\bigcirc17, \square14)\times 2$
          & $(68 + 11) \pm 3$
          \\
      
        \texttt{ tree-dfs }
          & binary tree
          & $2 / 2$
          & $(\bigcirc19, \square16)\times 2$
          & $(61 + 2) \pm 5$
          \\
      
        \texttt{ tree-down }
          & binary tree
          & $1 / 1$
          & $(\bigcirc11, \square12)\times 20$
          & $(140 + 28) \pm 7$
          \\
      
        \texttt{ tree-halves }
          & binary tree
          & $4 / 4$
          & $(\bigcirc26, \square23)\times 8$
          & $(101 + 188) \pm 9$
          \\
      
        \texttt{ tree-nav }
          & chained binary tree
          & $2 / 2$
          & $(\bigcirc12, \square19)\times 12$
          & $(147 + 38) \pm 6$
          \\
      
        \texttt{ philos }
          & ring
          & $1 / 1$
          & $(\bigcirc16, \square14)\times 8$
          & $(144 + 17) \pm 6$
          \\
      
        \texttt{ consensus }
          & star
          & $5 / 5$
          & $(\bigcirc29, \square23)\times 6$
          & $(90 + 23) \pm 7$
          \\
      
        \texttt{ leader-election }
          & ring
          & $1 / 2$
          & $(\bigcirc27, \square32)\times 2$
          & $(64 + 36) \pm 7$
          \\
      
        \texttt{ server-loop }
          & ring of stars
          & $2 / 3$
          & $(\bigcirc19, \square19)\times 8$
          & $(118 + 4594) \pm 31$
          \\
      
        \texttt{ coverapprox }
          & star
          & $1 / 2$
          & $(\bigcirc4, \square8)\times 6$
          & $(77 + 56) \pm 8$
          \\
      
        \texttt{ simplify-me }
          & star
          & $1 / 2$
          & $(\bigcirc8, \square9)\times 4$
          & $(71 + 22) \pm 9$
          \\
      
        \texttt{ propagation }
          & ring
          & $1 / 2$
          & $(\bigcirc16, \square15)\times 4$
          & $(98 + 257) \pm 8$
          \\
      
        \texttt{ open }
          & ring
          & $0 / 1$
          & $(\bigcirc15, \square16)\times 4$
          & $(80 + 260) \pm 11$
          \\
      
      \hline
    \end{tabular}}
  \end{center}
  \vspace*{-\baselineskip}
  \caption{Table of experiments. ``Result'', $S/T$ means that of $T$
    safety properties, $S$ were successfully proven. ``Size'',
    $(\bigcirc p, \square t) \times n$ means that LoLA analyzed $n$
    nets, consisting of each approx. $p$ places and $t$ transitions.
    ``Runtime'', $(p+l)\pm e$ means that ParCoSys ran for $p+l$
    milliseconds with a standard deviation of $e$ milliseconds,
    including $p$ computing the abstraction and $l$ waiting for LoLA
    to respond.}
  \label{fig:benchmarks}
  \vspace*{-1.5\baselineskip}
\end{figure}

For most of these examples, we were able to automatically prove
several safety properties (\figref{fig:benchmarks}).  For instance, in
\texttt{lock}, we successfully verify the safety property that
$\mathit{proc}*p.(\mathit{on}) > 0 \wedge \mathit{proc}.(\mathit{on})
> 0$ is unreachable (\ie an instance of $\mathit{proc}$ may not have a token in
$\mathit{on}$ at the same time as any other instance of
$\mathit{proc}$), but we are unable to verify that
$\mathit{proc}.(\mathit{on}) > 1$ or $\mathit{proc}*p.(\mathit{on}) +
\mathit{proc}.(\mathit{on}) > 1$ are unreachable. Thus \texttt{lock}
is marked $1/3$, denoting 1 of 3 successful safety tests. In most
cases, we proved the stated properties (those where we failed require
future work to refine the abstraction).

The times were obtained on a Intel Ultra 7 laptop, with 16GiB RAM,
under Ubuntu 24.04. All benchmark specifications are provided as
additional material.




\enlargethispage{5mm}

\vspace*{-.7\baselineskip}
\section{Conclusions}

We present two orthogonal verification results for parameterized
process networks with topology specified using hyperedge-replacement
graph grammars. The first result is a finitary counting abstraction,
that consists in collapsing nodes of the same type in the
parameterized family of Petri nets that gives the semantics of
behaviors. The second result identifies a decidable fragment of the
(undecidable) parameterized verification problem and evaluates its
complexity bounds.

\bibliographystyle{abbrv}
\bibliography{refs}

\ifLongVersion\else
\appendix
\newpage
\section{Proofs}
\label{app:proofs}
\printProofs[proofs]

\section{Quotients}
\label{app:quotients}
\printProofs[quotients]

\section{Initial Markings}
\label{app:initial}
\printProofs[initial]

\section{Lower Bound}
\label{app:lower-bound}
\printProofs[hardness]
\fi

\end{document}